\documentclass[a4paper,fleqn]{article}
\usepackage{a4wide}
\usepackage{amsmath,amssymb,amsthm}
\usepackage{thmtools}
\usepackage{tikz}
\usepackage{enumerate}
\usepackage{xifthen}
\usepackage{graphicx}
\usepackage{abstract}
\usepackage{authblk}

\usepackage{hyperref}
\usepackage[T1]{fontenc}
\usepackage{mathpazo}
\makeatletter\@ifpackageloaded{mathpazo}\@tempswatrue\@tempswafalse
\if@tempswa
  \DeclareFontFamily{OT1}{pzc}{}
  \DeclareFontShape{OT1}{pzc}{m}{it}{<-> s * [1.15] pzcmi7t}{}
  \DeclareMathAlphabet{\mathpzc}{OT1}{pzc}{m}{it}
\fi\makeatother
\usepackage{bm}
\usepackage[utf8]{inputenc}

\usepackage[sorting=nyt,giveninits=true,backend=biber]{biblatex}
\makeatletter\@ifpackageloaded{biblatex}{%
  \usepackage{csquotes} 
  \bibliography{../../references}
  \renewbibmacro{in:}{%
    \ifentrytype{incollection}{\printtext{\bibstring{in}\intitlepunct}}{}}
  \renewbibmacro{publisher+location+date}{%
    \iflistundef{publisher}
      {\setunit*{\addcomma\space}}
      {\setunit*{\addcomma\space}}%
    \printlist{publisher}%
    \setunit*{\addcomma\space}%
    \printlist{location}%
    \setunit*{\addcomma\space}%
    \usebibmacro{date}%
    \newunit}
  \DeclareFieldFormat[article]{pages}{#1\isdot}
  \DeclareFieldFormat[article,incollection,inproceedings,unpublished]{title}{#1\isdot}
  \DeclareFieldFormat[thesis]{title}{\mkbibemph{#1\isdot}}
  \DeclareFieldFormat[unpublished]{date}{(#1)\isdot}
  \DeclareFieldFormat[unpublished]{note}{#1\nopunct} 
  \DeclareFieldFormat[article]{journaltitle}{\mkbibemph{#1\isdot}}
  
  \AtEveryBibitem{%
    \ifentrytype{book}{}{
      \clearname{editor}
    }
  }
  \newbibmacro*{bbx:parunit}{%
    \ifbibliography
      {\setunit{\bibpagerefpunct}\newblock
       \usebibmacro{pageref}%
       \clearlist{pageref}%
       \setunit{\adddot\par\nobreak}}
      {}
  }
  \renewbibmacro*{doi+eprint+url}{%
    \usebibmacro{bbx:parunit}
    \iftoggle{bbx:doi}
      {\printfield{doi}}
      {}%
    \iftoggle{bbx:eprint}
      {\usebibmacro{eprint}}
      {}%
    \iftoggle{bbx:url}
      {\usebibmacro{url+urldate}}
      {}
  }
  \renewbibmacro*{eprint}{%
    \usebibmacro{bbx:parunit}
    \iffieldundef{eprinttype}
      {\printfield{eprint}}
      {\printfield[eprint:\strfield{eprinttype}]{eprint}}
  }
  \renewbibmacro*{url+urldate}{%
    \usebibmacro{bbx:parunit}
    \printfield{url}%
    \iffieldundef{urlyear}
      {}
      {\setunit*{\addspace}%
       \printtext[urldate]{\printurldate}}
  }
}{}\makeatother

\declaretheorem[numberwithin=section,refname={theorem,theorems},Refname={Theorem,Theorems}]{theorem}
\declaretheorem[sibling=theorem,style=definition]{definition}
\declaretheorem[sibling=theorem,style=definition,name=Example]{example}

\declaretheorem[sibling=theorem,name=Lemma]{lemma}
\declaretheorem[sibling=theorem,name=Proposition]{proposition}
\declaretheorem[sibling=theorem,name=Corollary]{corollary}
\declaretheorem[sibling=theorem,name=Conjecture]{conjecture}
\declaretheorem[sibling=theorem,name=Question]{question}
\declaretheorem[sibling=theorem,name=Problem]{problem}

\makeatletter\@ifpackageloaded{hyperref}{%
  \usepackage{xcolor}
  \definecolor{dark-red}{rgb}{0.4,0.15,0.15}
  \definecolor{dark-blue}{rgb}{0.15,0.15,0.4}
  \definecolor{medium-blue}{rgb}{0,0,0.5}
  \hypersetup{
    colorlinks,
    linkcolor={dark-red},
    citecolor={dark-blue},
    urlcolor={medium-blue}%
  }

}{}\makeatother

\usepackage{sectsty}
\sectionfont{\large\bfseries}
\subsectionfont{\normalsize\bfseries}

\newcommand{\dio}[1]{\mathrm{dio}(#1)}
\newcommand{\ice}[1]{\mathrm{ice}(#1)}
\newcommand{\ind}[1]{\mathrm{ind}(#1)}
\newcommand{\s}[2][*]{#2_{#1}}
\newcommand{\rep}[2][]{\mathrm{rep}_{#1}(#2)}
\newcommand{\val}[2][]{\mathrm{val}_{#1}(#2)}
\newcommand{\irep}[2]{\mathrm{inrc}(#1, #2)}
\newcommand{\prep}[2]{\mathrm{pnrc}(#1, #2)}
\newcommand{\interval}[1]{\mathcal{I}_{#1}}
\newcommand{\infw}[1]{%
  \ifcat\noexpand#1\relax\bm{#1}
  \else\mathbf{#1}\fi}          
\providecommand{\abs}[1]{\lvert#1\rvert}
\providecommand{\Abs}[1]{\left\lvert#1\right\rvert}
\providecommand{\norm}[1]{\lVert#1\rVert}

\newcommand{\Lang}[2][]{\mathcal{L}_{#2}\ifthenelse{\isempty{#1}}{}{(#1)}}

\newcommand{\N}{\mathbb{N}}

\newcommand{\keywords}[1]{\par\noindent{\footnotesize{\em Keywords\/}: #1}}

\begin{document}
  \title{Initial nonrepetitive complexity of regular episturmian words and their Diophantine exponents}
  \author[,1,2,3]{Jarkko Peltomäki\footnote{Corresponding author.\\E-mail address: \href{mailto:r@turambar.org}{r@turambar.org} (J. Peltomäki).}}
  \affil[1]{The Turku Collegium for Science, Medicine and Technology TCSMT, University of Turku, Turku, Finland}
  \affil[2]{Turku Centre for Computer Science TUCS, Turku, Finland}
  \affil[3]{Department of Mathematics and Statistics, University of Turku, Turku, Finland}
  \date{}
  \maketitle
  \begin{centering}
  In memory of my beloved brother Sauli Peltomäki (1996--2021).\\
  \end{centering}
  \vspace{1em}
  \noindent
  \hrulefill
  \begin{abstract}
    \vspace{-1em}
    \noindent
    Regular episturmian words are episturmian words whose directive words have a regular and restricted form making
    them behave more like Sturmian words than general episturmian words. We present a method to evaluate the initial
    nonrepetitive complexity of regular episturmian words extending the work of Wojcik on Sturmian words. For this, we
    develop a theory of generalized Ostrowski numeration systems and show how to associate with each episturmian word a
    unique sequence of numbers written in this numeration system.

    The description of the initial nonrepetitive complexity allows us to obtain novel results on the Diophantine
    exponents of regular episturmian words. We prove that the Diophantine exponent of a regular episturmian word is
    finite if and only if its directive word has bounded partial quotients. Moreover, we prove that the Diophantine
    exponent of a regular episturmian word is strictly greater than $2$ if the sequence of partial quotients is
    eventually at least $3$.

    Given an infinite word $x$ over an integer alphabet, we may consider a real number $\xi_x$ having $x$ as a
    fractional part. The Diophantine exponent of $x$ is a lower bound for the irrationality exponent of $\xi_x$. Our
    results thus yield nontrivial lower bounds for the irrationality exponents of real numbers whose fractional parts
    are regular episturmian words. As a consequence, we identify a new uncountable class of transcendental numbers
    whose irrationality exponents are strictly greater than $2$. This class contains an uncountable subclass of
    Liouville numbers.
    \vspace{1em}
    \keywords{sturmian word, episturmian word, diophantine exponent, initial nonrepetitive complexity, irrationality exponent}
    \vspace{-1em}
  \end{abstract}
  \hrulefill


  \section{Introduction}
  The fractional part of the expansion of a real number in some base can be interpreted as a right-infinite word. Major
  open problems in number theory concern the expansions of well-known numbers such as $\sqrt{2}$, $\pi$, or $e$. The
  inverse problem of inferring properties of a number $\xi_\infw{x}$ whose fractional part matches a prescribed
  infinite word $\infw{x}$ has attracted much attention especially in the last two decades. One of the most significant
  results is that of Adamczewski and Bugeaud \cite{2007:on_the_complexity_of_algebraic_numbers_I_expansions} from
  $2007$ stating that if the factor complexity function $p(\infw{x}, n)$ of an aperiodic infinite word $\infw{x}$ (a
  purely combinatorial notion) is sublinear, then $\xi_\infw{x}$ is transcendental. More recently, Bugeaud and Kim
  \cite{2019:a_new_complexity_function_repetitions_in_sturmian} introduced the notion of the exponent of repetition of
  an infinite word and studied it in relation to Sturmian words proving, among other results, that if
  $\lim_{n\to\infty} (p(\infw{x}, n) - n) < \infty$, then the irrationality exponent of $\xi_\infw{x}$ is at least
  $5/3 + 4\sqrt{10}/15$. The exponent of repetition is closely linked to the notion of the Diophantine exponent of an
  infinite word. The significance of this notion is that the Diophantine exponent of $\infw{x}$ is a lower bound to the
  irrationality exponent of $\xi_\infw{x}$. In this paper, we consider the class of regular episturmian words and prove
  results on their Diophantine exponents by characterizing their initial nonrepetitive complexity function. This
  provides novel results on the irrationality exponents of numbers whose fractional parts match a regular episturmian
  word.

  \subsection{Episturmian Words}
  An infinite word is Sturmian if it has exactly $n+1$ distinct factors (subwords) of length $n$ for all $n$. Sturmian
  words have many equivalent definitions and they can be generalized in various ways depending on the definition being
  used. Generalizing the work of de Luca on iterated palindromic closure, Droubay, Justin, and Pirillo introduced in
  \cite{2001:episturmian_words_and_some_constructions_of_de} episturmian words that further generalize the so-called
  Arnoux-Rauzy words \cite{1991:representation_geometrique_de_suites_de_complexite_2n+1}. This purely combinatorial
  generalization is defined as follows. For a finite word $w$, let $w^{(+)}$ be the shortest palindrome having $w$ as a
  prefix. Let $\Delta = y_1 y_2 \dotsm$ be an infinite word and define a sequence $(u_k)$ of finite words as follows:
  \begin{align*}
    u_1 &= \varepsilon, \\
    u_{k+1} &= (u_k y_k)^{(+)}.
  \end{align*}
  The limit $\infw{c}_\Delta$ of the words $(u_k)$ is the standard episturmian word with directive word $\Delta$. An
  episturmian word with directive word $\Delta$ is an infinite word sharing the set of factors with $\infw{c}_\Delta$.
  Sturmian words correspond to directive words that are binary and not ultimately constant. The most famous standard
  episturmian word that is not Sturmian is the Tribonacci word
  \begin{equation*}
    01020100102010102010010201020100102010102010010201\dotsm
  \end{equation*}
  having directive word $012012012 \dotsm$. The main results of this paper are stated for regular episturmian words
  whose directive words are of special form. Write $\Delta$ in the form $x_1^{a_1} x_2^{a_2} \dotsm$ with
  $x_k \neq x_{k+1}$ and $a_k > 0$ for all $k$. If the sequence $x_1 x_2 \dotsm$ equals the periodic sequence with
  period $0 1 2 \dotsm (d-1)$ for some $d$, then we say that the episturmian words with directive word $\Delta$ are
  \emph{regular} with period $d$. This class of regular episturmian words contains Sturmian words (case $d = 2$) and
  $d$-bonacci words, generalizations of the Fibonacci and Tribonacci words.

  Episturmian words enjoy many of the good properties of Sturmian words as described in the foundational papers
  \cite{2001:episturmian_words_and_some_constructions_of_de,2002:episturmian_words_and_episturmian_morphisms,2004:episturmian_words_shifts_morphisms_and_numeration}
  and the survey \cite{2009:episturmian_words_a_survey}, but some properties such as interpretation as codings of
  irrational rotations are lost. Standard references for Sturmian words are
  \cite[Ch.~2]{2002:algebraic_combinatorics_on_words},
  \cite[Ch.~6]{2002:substitutions_in_dynamics_arithmetics_and_combinatorics}. We refer the reader to
  \cite[Ch.~4]{diss:jarkko_peltomaki} for an introduction to Sturmian words as codings of irrational rotations.

  \subsection{Initial Nonrepetitive Complexity}
  We define the \emph{initial nonrepetitive complexity function} $\irep{\infw{x}}{n}$ of an infinite word $\infw{x}$ by
  \begin{equation*}
    \irep{\infw{x}}{n} = \max\{m \colon \text{$\infw{x}[i,i+n-1] \neq \infw{x}[j,j+n-1]$ for all $i, j$ with $1 \leq i < j \leq m$}\}.
  \end{equation*}
  The number $\irep{\infw{x}}{n}$ is the maximum number of factors of length $n$ seen when $\infw{x}$ is read from left
  to right prior to the first repeated factor of length $n$. In other words, the prefix of $\infw{x}$ of length
  $\irep{\infw{x}}{n} + n$ is the shortest prefix of $\infw{x}$ containing two occurrences of some factor of length
  $n$.

  The notion of initial nonrepetitive complexity was introduced independently by Moothatu
  \cite{2012:eulerian_entropy_and_non-repetitive_subword_complexity} and Bugeaud and Kim
  \cite{2019:a_new_complexity_function_repetitions_in_sturmian}. Nicholson and Rampersad
  \cite{2016:initial_non-repetitive_complexity_of_infinite_words} examine the general properties of this function and
  determine it explicitly for certain words such as the Thue-Morse word, the Fibonacci word, and the Tribonacci word.
  Their results were generalized for all standard Arnoux-Rauzy words in
  \cite{2020:on_non-repetitive_complexity_of_arnoux-rauzy_words} by Medkov\'a et al.

  In this paper, we set up a framework that allows us in principle to determine the initial nonrepetitive complexity of
  any episturmian word. In order to do this, we generalize the $S$-adic expansion of Sturmian words derived in
  \cite{2006:initial_powers_of_sturmian_sequences} to all episturmian words and prove in
  \autoref{thm:ostrowski_val_limit} that an episturmian word $\infw{t}$ can be expressed in the form
  \begin{equation*}
    \infw{t} = \lim_{k\to\infty} T^{\rho_k}(\infw{c})
  \end{equation*}
  where $\infw{c}$ is the corresponding standard episturmian word and $(\rho_k)$ is an integer sequence expressed in a
  generalized Ostrowski numeration system. In other words, we show that each episturmian word is a limit of appropriate
  shifts of the corresponding standard word $\infw{c}$. This means that studying a property of episturmian words can be
  reduced to studying the property on shifts of standard episturmian words. Here we apply this principle and determine
  the initial nonrepetitive complexity of \emph{regular} episturmian words generalizing a result of Wojcik
  \cite{diss:caius_wojcik} on Sturmian words. This results in \autoref{thm:main} which is too complicated to be stated
  here. We leave the characterization of this function for all episturmian words open.

  \subsection{Diophantine Exponents and Main Results}\label{ssec:dio}
  The Diophantine exponent is a combinatorial exponent of infinite words. It is introduced in
  \cite{2007:dynamics_for_beta_shifts_and_diophantine_approximation}, but is used implicitly in the earlier works by
  the same authors.

  \begin{definition}
    Let $\infw{x}$ be an infinite word. We let its \emph{Diophantine exponent}, denoted by $\dio{\infw{x}}$, to be the
    supremum of all real numbers $\rho$ for which there exist arbitrarily long prefixes of $\infw{x}$ of the form
    $UV^e$, where $U$ and $V$ are finite words and $e$ is a real number, such that
    \begin{equation*}
      \frac{\abs{UV^e}}{\abs{UV}} \geq \rho.
    \end{equation*}
  \end{definition}

  The concept of Diophantine exponent has inherent interest to a word-combinatorist, and the concept has connections to
  other combinatorial exponents widely studied in combinatorics on words. What makes it special is the ingenious and
  simple result that $\dio{\infw{x}}$ is a lower bound for the irrationality exponent $\mu(\xi_{\infw{x},b})$ of the
  real number $\xi_{\infw{x},b}$ having $\infw{x}$ as the fractional part of its base-$b$ expansion.

  The crucial result here is that of Bugeaud and Kim \cite{2019:a_new_complexity_function_repetitions_in_sturmian}
  stating that
  \begin{equation*}
      \dio{\infw{x}} = 1 + \limsup_{n\to\infty} \frac{n}{\irep{\infw{x}}{n}}.
  \end{equation*}
  This, together with the characterization of the function $\irep{\infw{x}}{n}$ for regular episturmian words
  (\autoref{thm:main}), provides means to find lower bounds for $\mu(\xi_{\infw{x},b})$ when $\infw{x}$ is a regular
  episturmian word.

  Our two main results are as follows. The results were previously proved for Sturmian words by
  Adamczewski (based on the results of \cite{2006:initial_powers_of_sturmian_sequences}) and Komatsu
  \cite{1996:a_certain_power_series_and_the_inhomogeneous_continued} respectively.

  \begin{theorem}
    Let $\infw{t}$ be a regular episturmian word of period $d$ with directive word
    $\Delta = x_1^{a_1} x_2^{a_2} \dotsm$. If $d = 2$ or $\limsup_k a_k \geq 3$, then $\mu(\xi_{\infw{t},b}) > 2$.
  \end{theorem}

  \begin{theorem}
    Let $\infw{t}$ be a regular episturmian word with directive word $\Delta = x_1^{a_1} x_2^{a_2} \dotsm$. Then
    $\xi_{\infw{t},b}$ is a Liouville number if and only if the sequence $(a_k)$ is unbounded.
  \end{theorem}

  We thus identify a new uncountable class of numbers with irrationality exponent strictly greater than $2$ and a new
  uncountable class of Liouville numbers.

  In addition, we show in \autoref{sec:ie} that it is possible that $\mu(\xi_{\infw{t},b}) > \dio{\infw{t}}$ for an
  episturmian word over a $3$-letter alphabet. For Sturmian words, we have the equality
  $\mu(\xi_{\infw{t},b}) = \dio{\infw{t}}$ by a result of Bugeaud and Kim
  \cite{2019:a_new_complexity_function_repetitions_in_sturmian}. Additional results on the Diophantine exponents of
  $d$-bonacci words are provided in \autoref{sec:diophantine}.

  \section{Preliminaries from Combinatorics on Words}
  We use standard notions and notations from combinatorics on words. For a general reference, see, e.g.,
  \cite{2002:algebraic_combinatorics_on_words}.

  A \emph{word} is a finite sequence of symbols from some finite set of \emph{letters} called an \emph{alphabet}. If
  $w = w_1 \dotsm w_n$ with $w_i \in \mathcal{A}$, then we say that $w$ is a word of length $n$ over $\mathcal{A}$, and
  we set $\abs{w} = n$. The \emph{empty word}, the unique word of length $0$, is denoted by $\varepsilon$. The set of
  words over $\mathcal{A}$ is denoted by $\mathcal{A}^*$. If $u$ and $v$ are words such that $u = u_1 \dotsm u_n$ and
  $v = v_1 \dotsm v_m$ with $u_i, v_i \in \mathcal{A}$, then their \emph{concatenation} $uv$ is the word
  $u_1 \dotsm u_n v_1 \dotsm v_m$. If $w = uzv$, then $u$ is a \emph{prefix} of $w$, $v$ is a \emph{suffix} of $w$, and
  $z$ is a \emph{factor} of $w$. If $u \neq w$, then $u$ is a \emph{proper} prefix of $w$; proper suffix is defined
  analogously. If $w = uv$, then by $u^{-1} w$ and $w v^{-1}$ we respectively refer to the words $v$ and $u$. An
  \emph{occurrence} of $u$ in $w_1 \dotsm w_n$ is an index $i$ such that $u$ is a prefix of $w_i w_{i+1} \dotsm w_n$.
  By $w[i, j]$ we mean the factor $w_i \dotsm w_j$ whenever the indices $i$ and $j$ make sense.

  Let $w$ be a word over $\mathcal{A}$. By $w^n$ we refer to the concatenation $w \dotsm w$ where $w$ is repeated $n$
  times. This is an integer power, and by a fractional power $w^e$, $e \geq 1$, we mean the word $(uv)^n u$ with
  $uv = w$ and $e = n + \abs{u}/\abs{w}$. If $w = u^n$ only if $n = 1$, then we say that $w$ is \emph{primitive}. The
  word $w$ is primitive if and only if $w$ occurs exactly twice in $w^2$. If $w = uv$, then the word $vu$ is a
  \emph{conjugate} of $w$. If $w = w_1 \dotsm w_n$, $w_i \in \mathcal{A}$, then the \emph{reversal} of $w$ is the word
  $w_n \dotsm w_1$. If a word equals its reversal, then we say that it is a \emph{palindrome} (we count the empty word
  as a palindrome). If $w_i = w_{i+p}$ for all $i$ such that $0 \leq i < \abs{w} - p$, then $w$ has \emph{period} $p$.
  A mapping $\tau\colon \mathcal{A}^* \to \mathcal{A}^*$ is a \emph{morphism} if $\tau(uv) = \tau(u)\tau(v)$ for all
  $u, v \in \mathcal{A}^*$.

  An infinite sequence of letters $\infw{x}$ over $\mathcal{A}$ is called an \emph{infinite word}. The set of infinite
  words over $\mathcal{A}$ is denoted by $\mathcal{A}^\omega$. This set is naturally equipped with the product
  topology. If $\infw{x} = x_1 x_2 x_3 \dotsm$ with $x_i \in \mathcal{A}$, then its \emph{shift} $T(\infw{x})$ is the
  word $x_2 x_3 \dotsm$. The map $T$ is continuous with respect to the product topology. The \emph{language}
  $\Lang{\infw{x}}$ of $\infw{x}$ is its set of factors (the preceding notions of prefix, suffix, factor, and
  occurrence directly generalize to infinite words), and by $\Lang[n]{\infw{x}}$ we refer to the set of factors of
  $\infw{x}$ of length $n$. We say that the language $\Lang{\infw{x}}$ is \emph{closed under reversal} if the reversal
  of each $w$ in $\Lang{\infw{x}}$ is also in $\Lang{\infw{x}}$. If $wa, wb \in \Lang{\infw{x}}$ for distinct letters
  $a$ and $b$, then we say that the factor $w$ of $\infw{x}$ is \emph{right special}; \emph{left special} factors are
  defined analogously. If a factor is both right and left special, we say it is \emph{bispecial}. The set
  $\{\infw{w} \in \mathcal{A}^\omega : \Lang{\infw{w}} \subseteq \Lang{\infw{x}}\}$ is called the \emph{subshift}
  generated by $\infw{x}$. The subshift is a $T$-invariant and closed set. An infinite word $\infw{x}$ is
  \emph{ultimately periodic} if $\infw{x} = u v^\omega$ where $u$ and $v$ are finite words, $v \neq \varepsilon$, and
  $v^\omega = v v v \dotsm$. If $\infw{x}$ is not ultimately periodic, then it is \emph{aperiodic}. If
  $\infw{x} = x_1 x_2 \dotsm$ with $x_i \in \mathcal{A}$ and $\tau\colon A^* \to A^*$ is a morphism, then we define
  $\tau(\infw{x})$ to be the infinite word $\tau(x_1) \tau(x_2) \dotsm$.

  \section{Episturmian Words and Generalized Ostrowski Numeration Systems}
  Episturmian words were introduced in \cite{2001:episturmian_words_and_some_constructions_of_de} as generalizations of
  Sturmian words based on palindromic closure. 

  Let $\mathcal{A}$ be the integer alphabet $\{0, 1, \ldots, d-1\}$ of $d$ letters. Let $w^{(+)}$ be the shortest
  palindrome having the word $w$ as a prefix. Let $\Delta = y_1 y_2 \dotsm$ be an infinite word over $\mathcal{A}$ and
  define a sequence $(u_k)$ of finite words as follows:
  \begin{align*}
    u_1 &= \varepsilon, \\
    u_{k+1} &= (u_k y_k)^{(+)},
  \end{align*}
  and set $\infw{c}_\Delta = \lim_{k\to\infty} u_k$. We say that $\infw{c}_\Delta$ is a \emph{standard episturmian
  word} with directive word $\Delta$. In what follows, we often use the name \emph{epistandard} for a standard
  episturmian word. Each epistandard word has a unique directive word. The words $u_k$ are called \emph{central words}
  and they are exactly the palindromic prefixes of $\infw{c}_\Delta$. In fact, the words $u_k$ are exactly the
  bispecial factors of $\infw{c}_\Delta$. This means in particular that every prefix of $\infw{c}_\Delta$ is left
  special.

  An infinite word $\infw{t}$ is \emph{episturmian} with directive word $\Delta$ if
  $\Lang{\infw{t}} = \Lang{\infw{c}_\Delta}$. Equivalently, a word $\infw{t}$ is episturmian if $\Lang{\infw{t}}$ is
  closed under reversal and $\infw{t}$ has at most one right special factor of length $n$ for all $n$
  \cite[Thm.~5]{2001:episturmian_words_and_some_constructions_of_de}. If $\infw{t}$ is binary and aperiodic, then we
  call $\infw{t}$ \emph{Sturmian}. It is equivalent to require that $\Delta$ is binary and contains both $0$ and $1$
  infinitely often.

  An episturmian word is ultimately periodic if and only if the directive word $\Delta$ is eventually constant, that
  is, $y_n = a$ for some $a \in \mathcal{A}$ for all $n$ large enough
  \cite[Thm.~3]{2001:episturmian_words_and_some_constructions_of_de}. In this paper, we consider only aperiodic
  episturmian words, so we assume that $(y_n)$ is not eventually constant.

  \subsection{The Intercept of an Episturmian Word}
  Our next aim is to desubstitute an episturmian word with episturmian morphisms in a certain way that gives rise to
  the concept of the intercept of an episturmian word. This generalizes the arguments of
  \cite{2006:initial_powers_of_sturmian_sequences} for Sturmian words.

  Let $\mathcal{A}$ be the integer alphabet $\{0, 1, \ldots, d-1\}$ as before. For each $y \in \mathcal{A}$, define the
  morphisms $L_y$ as follows:
  \begin{equation*}
    L_y(x) = \begin{cases}
               y, &\text{if $x = y$}, \\
               yx, &\text{if $x \neq y$}.
             \end{cases}
  \end{equation*}
  These morphisms belong to the class of episturmian morphisms; see \cite[Sect.~3]{2009:episturmian_words_a_survey}.
  Let $\infw{t}$ be an episturmian word over $\mathcal{A}$ with directive word $y_1 y_2 \dotsm$. Then, depending on if
  the first letter of $\infw{t}$ is $y_1$, we have $\infw{t} = L_{y_1}(\infw{t}_1)$ or
  $\infw{t} = T(L_{y_1}(\infw{t}_1))$ for some unique infinite word $\infw{t}_1$. It is well-known that $\infw{t}_1$ is
  also an episturmian word over $\mathcal{A}$ (possibly over a strict subalphabet of $\mathcal{A}$)
  \cite[Thm.~3.10]{2002:episturmian_words_and_episturmian_morphisms}. Thus there exist an integer $b_1$ in $\{0, 1\}$
  and a unique episturmian word $\infw{t}_1$ such that $\infw{t} = T^{b_1} \circ L_{y_1}(\infw{t}_1)$. By repeating
  this decoding, we see that there exists a unique integer sequence $(b_k)$ and a unique sequence $(\infw{t}_k)$ of
  episturmian words such that
  \begin{equation}\label{eq:desubst}
    \infw{t} = (T^{b_1} \circ L_{y_1}) \circ \dotsm \circ (T^{b_k} \circ L_{y_k})(\infw{t}_k)
  \end{equation}
  for all $k \geq 1$. It is easy to see that if $y_k = y_{k+1}$ and $b_{k+1} = 0$, then $b_k = 0$. Indeed, if
  $b_k = 1$, then $\infw{t}_{k-1}$ does not begin with $y_k$ so, by the form of the morphism $L_{y_{k+1}}$, neither
  does $\infw{t}_k$ begin with $y_k$. Since $y_k = y_{k+1}$, this implies that $b_{k+1} = 1$.
  
  Let us write the directive word $\Delta$ more compactly as follows:
  \begin{equation}\label{eq:dw_multiplicative}
    \Delta = y_1 y_2 \dotsm = x_1^{a_1} x_2^{a_2} \dotsm
  \end{equation}
  with $a_k \geq 1$, and $x_k \neq x_{k+1}$ for all $k$. We call the sequence $(a_k)$ the sequence of \emph{partial
  quotients} of $\infw{t}$ (the choice of the name will become apparent below). Let $r_0 = 0$ and
  $r_k = a_1 + \dotsm + a_k$ for $k \geq 1$. By the property given at the end of the previous paragraph, we see that
  $b_{r_k+1} \dotsm b_{r_{k+1}}$ (viewed as a word over $\{0, 1\}$) is of the form $0^* 1^*$, so we may write
  \begin{equation*}
    b_1 b_2 \dotsm = 0^{a_1 - c_1} 1^{c_1} 0^{a_2 - c_2} 1^{c_2} \dotsm
  \end{equation*}
  for some integers $c_k$ such that $0 \leq c_k \leq a_k$ for all $k$. Therefore
  \begin{equation*}
    \infw{t} = \Big( L_{x_1}^{a_1 - c_1} \circ (T \circ L_{x_1})^{c_1} \Big) \circ \dotsm \circ \Big( L_{x_k}^{a_k-c_k}
    \circ (T \circ L_{x_k})^{c_k} \Big) (\infw{t}_{r_k})
  \end{equation*}
  for all $k$. It is easy to verify that $L_y \circ T \circ L_y = T \circ L_y \circ L_y$ for all letters $y$, so we
  find that
  \begin{equation}\label{eq:episturmian_s_adic}
    \infw{t} = T^{c_1} L_{x_1}^{a_1} \circ \dotsm \circ T^{c_k} L_{x_k}^{a_k}(\infw{t}_{r_k})
  \end{equation}
  for all $k$. We call the sequence $c_1 c_2 \dotsm$ the \emph{intercept} of $\infw{t}$. The choice of the name will
  become apparent after we discuss Sturmian words below. Notice that the intercept of $\infw{t}$ is unique. Notice also
  the important fact that the above derivation guarantees that
  $T^{c_1} L_{x_1}^{a_1} \circ \dotsm \circ T^{c_k} L_{x_k}^{a_k}(z)$ is nonempty when $z$ is the first letter of
  $\infw{t}_{r_k}$.

  \begin{lemma}\label{lem:standard_intercept}
    Let $\Delta$ be a directive word as in \eqref{eq:dw_multiplicative}. The intercept of the epistandard word
    $\infw{c}_\Delta$ is $0^\omega$.
  \end{lemma}
  \begin{proof}
    The word $\infw{c}_\Delta$ begins with $x_1$, so $b_1 = 0$. The claim follows by induction because $\infw{t}_1$ is
    epistandard by \cite[Thm.~9]{2001:episturmian_words_and_some_constructions_of_de}.
  \end{proof}

  Next we introduce several sequences of morphisms and words in order to define the important generalized standard
  words and show their connection to the central words $u_k$. See
  \cite[Sect.~2]{2002:episturmian_words_and_episturmian_morphisms} for a slightly more elaborate presentation.

  Let $\Delta$ be a directive word as in \eqref{eq:dw_multiplicative}. Set
  \begin{equation*}
    \mu_k = L_{y_1} \circ \dotsm \circ L_{y_k} \quad \text{and} \quad \tau_k = \mu_{r_k} = L_{x_1}^{a_1} \circ \dotsm \circ L_{x_k}^{a_k}
  \end{equation*}
  with $\mu_0$ and $\tau_0$ being the identity map, and define the sequences $(h_k)_{k \geq 0}$, $(s_k)_{k \geq 0}$,
  and $(q_k)_{k \geq 0}$ by setting
  \begin{equation*}
    h_k = \mu_k(y_{k+1}), \quad s_k = h_{r_k} = \tau_k(x_{k+1}), \quad \text{and} \quad q_k = \abs{s_k}.
  \end{equation*}
  The words $s_k$ are the (finite) \emph{standard words} associated with the directive word $\Delta$. By definition,
  the epistandard word $\infw{c}_\Delta$ is the limit of both $(h_k)$ and $(s_k)$. The words $s_k$ are primitive
  \cite[Prop.~2.8]{2002:episturmian_words_and_episturmian_morphisms}. Notice that if $i$ is such that
  $r_{k-1} \leq i < r_k$, then $h_k = s_{k-1}$.

  Let $(\infw{c}_n)$ be the sequence of epistandard words such that $\infw{c}_\Delta = \mu_n(\infw{c}_n)$ as in
  \eqref{eq:desubst} (see \autoref{lem:standard_intercept}) and $(\mu_{n,k})$, $(\tau_{n,k})$, $(u_{n,k})$,
  $(h_{n,k})$, and $(s_{n,k})$ be the respective sequences for $\infw{c}_n$. A simple induction argument (see
  \cite[Eq.~3]{2002:episturmian_words_and_episturmian_morphisms}) shows that
  \begin{equation}\label{eq:u_1}
    u_k = \mu_p(u_{p,k-p}) u_{p+1}
  \end{equation}
  for all $p$ such that $0 \leq p < k$. Replacing $k$ by $k+1$ and $p$ by $k-1$ in \eqref{eq:u_1} yields
  \begin{equation}\label{eq:u_2}
    u_{k+1} = h_{k-1} u_k
  \end{equation}
  for all $k \geq 1$. Using \eqref{eq:u_2} repeatedly, we obtain
  \begin{equation}\label{eq:u_3}
    u_k = h_{k-2} h_{k-3} \dotsm h_{p-1} u_p = h_{k-2} h_{k-3} \dotsm h_1 h_0
  \end{equation}
  for $k \geq 2$. The equation \eqref{eq:u_3} directly implies the following important formula for $k \geq 1$:
  \begin{equation}\label{eq:u_4}
    u_{r_k+1} = s_{k-1}^{a_k} s_{k-2}^{a_{k-1}} \dotsm s_0^{a_1}.
  \end{equation}


  \begin{definition}
    Let $\Delta$ be as in \eqref{eq:dw_multiplicative}, and let $P(k) = \max\{p < k : y_p = y_k\}$ if this integer
    exists, and leave $P(k)$ undefined otherwise. Define $j(k)$ as the largest $j$ such that $j \leq k$ and
    $x_j = x_{k+1}$ when $P(r_k+1)$ exists and leave $j(k)$ undefined otherwise.
  \end{definition}

  For the next lemma, we make the following convention which makes our formulas less ``noisy''. We often have formulas
  involving $s_k^{a_{k+1}}$ where the subscript $k+1$ in the superscript $a_{k+1}$ is one greater than in the subscript
  $k$. This conveys no essential information, so we will write $\smash[b]{s_k^{\s{a}}}$ instead whenever there is no
  risk of confusion. We also take $\smash[t]{s_k^{\s{a}-1}}$ to mean $\smash[t]{s_k^{a_{k+1}-1}}$ and
  $\smash[t]{s_k^{\s{a} - \s{c}}}$ to mean $\smash[t]{s_k^{a_{k+1}-c_{k+1}}}$ etc.

  \begin{lemma}\label{lem:s_recursion}
    Let $\Delta$ be a directive word as in \eqref{eq:dw_multiplicative} and $k \geq 0$. If $P(r_k+1)$ exists, then
    \begin{equation}\label{eq:s_recursion_1}
      s_k = s_{k-1}^{\s{a}} \dotsm s_{j(k)}^{\s{a}} s_{j(k)-1}.
    \end{equation}
    If $P(r_k+1)$ does not exist, then
    \begin{equation}\label{eq:s_recursion_2}
      s_k = s_{k-1}^{\s{a}} \dotsm s_0^{\s{a}} x_{k+1}.
    \end{equation}
  \end{lemma}
  \begin{proof}
    Suppose first that $P(r_k+1)$ exists. Set $j = j(k)$. First of all, as $x_k \neq x_{k+1}$, we have
    $s_k = \tau_{k-1}(L_{x_k}^{a_k}(x_{k+1})) = \tau_{k-1}(x_k^{a_k} x_{k+1}) = s_{k-1}^{a_k} \tau_{k-1}(x_{k+1})$. As
    long as $k-t > j$, we have $x_{k-t} \neq x_{k+1}$ and
    $\smash[t]{\tau_{k-t}(x_{k+1}) = s_{k-t-1}^{a_{k-t}} \tau_{k-t-1}(x_{k+1})}$ by a similar computation. Hence we
    find that
    \begin{equation*}
      s_k = s_{k-1}^{\s{a}} \dotsm s_{j}^{\s{a}} \tau_j(x_{k+1}) = s_{k-1}^{\s{a}} \dotsm s_j^{\s{a}} \tau_{j-1}(x_j) = s_{k-1}^{\s{a}} \dotsm s_j^{\s{a}} s_{j-1}
    \end{equation*}
    because $x_j = x_{k+1}$. Say $P(r_k+1)$ does not exist. Then $x_{k-t} \neq x_{k+1}$ for all $t$ such that
    $0 \leq t < k$ and the above arguments yield
    \begin{equation*}
      s_k = s_{k-1}^{\s{a}} \tau_{k-1}(x_{k+1}) = \ldots = s_{k-1}^{\s{a}} \dotsm s_0^{\s{a}} \tau_0(x_{k+1}).
    \end{equation*}
    The claim follows as $\tau_0(x_{k+1}) = x_{k+1}$.
  \end{proof}

  The equations \eqref{eq:s_recursion_1} and \eqref{eq:s_recursion_2} imply that
  \begin{equation}\label{eq:q_recursion_1}
    q_k = a_k q_{k-1} + \dotsm + a_{j(k)+1} q_{j(k)} + q_{j(k)-1}
  \end{equation}
  when $P(r_k+1)$ exists and
  \begin{equation}\label{eq:q_recursion_2}
    q_k = a_k q_{k-1} + \dotsm + a_1 q_0 + 1
  \end{equation}
  when $P(r_k+1)$ does not exist. If $\Delta$ is periodic with period $d$, then we see that $(q_k)$ satisfies a linear
  recurrence of order $d$.

%

  \subsection{Regular Episturmian Words}
  In this section, we define regular episturmian words which have directive words of special form. This subclass of
  episturmian words has not had specific attention except for the paper of Glen
  \cite{2007:powers_in_a_class_of_A_strict_standard_episturmian} where the powers occurring in these words are studied.

  \begin{definition}
    Let $\Delta$ be a directive word as in \eqref{eq:dw_multiplicative}. If there exists an integer $d$ such that
    $d \geq 2$, the letters $x_1$, $\dotsm$, $x_d$ are pairwise distinct, and
    $x_1 x_2 \dotsm = (x_1 \dotsm x_d)^\omega$, then we say that the directive word $\Delta$ is \emph{regular} with
    period $d$. An episturmian word is regular if its directive word is regular.
  \end{definition}

  In what follows, we often assume that $x_1 x_2 \dotsm = (012 \dotsm (d-1))^\omega$ for some integer $d$ such that $d
  \geq 2$. Notice that regular episturmian words are exactly the Sturmian words when $d = 2$. This class includes the
  $d$-bonacci words $\infw{f}_d$ which are the epistandard words having directive words $(012 \dotsm (d-1))^\omega$.
  The $2$-bonacci word is called the \emph{Fibonacci word}, and the $3$-bonacci word is called the \emph{Tribonacci
  word}.

  The main advantage in studying regular episturmian words is that the function $j(k)$ is simple: $j(k) = k - (d - 1)$
  when $j(k)$ is defined, i.e., when $k \geq d$. This simplifies many properties. For example, from
  \eqref{eq:s_recursion_2} and \eqref{eq:s_recursion_1}, we have
  \begin{equation}\label{eq:reg_s_1}
    s_k = s_{k-1}^{\s{a}} \dotsm s_0^{\s{a}} x_{k+1}
  \end{equation}
  for $1 \leq k < d$, and
  \begin{equation}\label{eq:reg_s_2}
    s_k = s_{k-1}^{\s{a}} \dotsm s_{k-(d-1)}^{\s{a}} s_{k-d}
  \end{equation}
  for $k \geq d$. Two consecutive applications of \eqref{eq:reg_s_2} show that
  \begin{equation*}
    (a_{k+1} + 1)q_k = q_{k+1} + (a_{k-(d-2)} - 1)q_{k-(d-1)} + q_{k-d}
  \end{equation*}
  for $k \geq d$. In particular, we see that
  \begin{equation}\label{eq:reg_q}
    q_{k+1} < (a_{k+1} + 1)q_k
  \end{equation}
  for $k \geq d$. In a similar fashion, combining \eqref{eq:reg_s_2} and \eqref{eq:u_4} yields that
  \begin{equation}\label{eq:u_5}
    \abs{u_{r_k}} \geq a_k q_{k-1} + q_{k-(d+2)}
  \end{equation}
  for $k \geq d + 2$.

  We believe that most of the results of this paper can be carried out for general episturmian words. However, this
  leads to very complicated arguments; the arguments are already tedious and complicated in the regular case.

  \subsection{Generalized Ostrowski Numeration Systems}
  Let us now define a representation for a nonnegative integer $n$ in terms of the shift $T^n(\infw{c}_\Delta)$ of the
  epistandard word $\infw{c}_\Delta$. First, we prove a generalization of a famous result of Brown
  \cite[Thm.~2]{1993:descriptions_of_the_characteristic_sequence_of_an_irrational}.

  \begin{proposition}\label{prp:ostrowski_prefix}
    Let $\Delta$ be a directive word as in \eqref{eq:dw_multiplicative}, and let $n$ be a positive integer. Let
    $c_1 c_2 \dotsm$ be the intercept of $T^n(\infw{c}_\Delta)$. Then there exists an integer $k$ such that
    $c_k \neq 0$ and $c_i = 0$ for all $i > k$. Moreover the prefix of $\infw{c}_\Delta$ of length $n$ equals
    $s_{k-1}^{c_k} s_{k-2}^{c_{k-1}} \dotsm s_0^{c_1}$.
  \end{proposition}
  \begin{proof}
    We prove the claim by induction on $n$. From \autoref{lem:standard_intercept}, it follows that
    $\infw{c}_\Delta = \tau_1(\infw{c}_{\Delta'})$ where $\Delta' = T^{a_1}(\Delta)$. Thus if $n \leq a_1$, then
    $n 0^\omega$ is a valid intercept for $T^n(\infw{c}_\Delta)$, and the uniqueness of the intercept implies that
    $c_1 = n$. By definition, the word $s_1$ is a prefix of $\infw{c}_\Delta$ and $s_1 = \tau_1(x_2) = x_1^{a_1} x_2$,
    so $s_0^{c_1}$ is a prefix of $\infw{c}_\Delta$. This establishes the base case.

    Suppose that $n > a_1$. Then it follows from \eqref{eq:episturmian_s_adic} and the arguments preceding
    it that $T^{n-c_1 q_0}(\infw{c}_\Delta)$ is a $\tau_1$-image of an episturmian word $\infw{t}$ with intercept
    $c_2 c_3 \dotsm$ (recall from above that $q_0 = 1$). In fact, as in the proof of \autoref{lem:standard_intercept},
    the word $\infw{t}$ is a suffix of the epistandard word $\infw{c}_{\Delta'}$. Let $w$ be the prefix of
    $\infw{c}_{\Delta'}$ such that $\abs{\tau_1(w)} = n - c_1 q_0$, that is, say
    $\infw{t} = T^{\abs{w}}(\infw{c}_{\Delta'})$. The word $w$ must be nonempty as $n > a_1$ and $c_1 \leq a_1$. Now
    $0 < \abs{w} < n$, so the induction hypothesis implies that there exists an integer $k$ such that $c_k \neq 0$,
    $c_i = 0$ for all $i > k$, and
    \begin{equation*}
      w = s_{\Delta',k-2}^{c_k} \dotsm s_{\Delta',0}^{c_2}
    \end{equation*}
    where $s_{\Delta',j}$ is the $j$th standard word for the directive word $\Delta'$. By definition, we have
    $\tau_1(s_{\Delta',j}) = s_{j+1}$ for all $j$, so
    \begin{equation*}
      \tau_1(w) = s_{k-1}^{c_k} s_{k-2}^{c_{k-1}} \dotsm s_1^{c_2}.
    \end{equation*}
    Since $T^{n-c_1 q_0}(\infw{c}_\Delta) = \tau_1(\infw{t})$, the word $T^{n-c_1 q_0}(\infw{c}_\Delta)$ must begin
    with $x_1^{a_1}$. It follows that $\tau_1(w) s_0^{c_1}$ is a prefix of $T^n(\infw{c}_\Delta)$. Since
    $\abs{\tau_1(w) s_0^{c_1}} = n$, the claim follows.
  \end{proof}

  The connection to Brown's result becomes clearer as we study greedy expansions below. Thanks to
  \autoref{prp:ostrowski_prefix}, we can give the following definitions.

  \begin{definition}
    Let $\Delta$ be a directive word as in \eqref{eq:dw_multiplicative}, and let $n$ be a positive integer. We let the
    \emph{representation}, or the Ostrowski expansion, $\rep[\Delta]{n}$ of $n$ to be the word $c_1 \dotsm c_k$ where
    $c_k \neq 0$ and $c_1 \dotsm c_k 0^\omega$ is the intercept of the word $T^n(\infw{c}_\Delta)$. In addition, we set
    $\rep[\Delta]{0} = \varepsilon$.
  \end{definition}

  \begin{definition}
    Let $\Delta$ be a directive word as in \eqref{eq:dw_multiplicative}. If $c_1 \dotsm c_k$ is a sequence of
    nonnegative integers, then we set the \emph{value} $\val[\Delta]{c_1 \dotsm c_k}$ of $c_1 \dotsm c_k$ to be the
    number
    \begin{equation*}
      \sum_{i=0}^k c_i q_{i-1}.
    \end{equation*}
  \end{definition}

  We often omit the subscript $\Delta$ in $\rep[\Delta]{n}$ and $\val[\Delta]{w}$ if the directive word $\Delta$ is
  clear from context.

  It follows from \autoref{prp:ostrowski_prefix} that if $\rep[\Delta]{n} = c_1 \dotsm c_k$, then
  $\val[\Delta]{c_1 \dotsm c_k} = n$. Therefore the Ostrowski expansion of an integer can be viewed as an expansion
  with respect to the numeration system associated with the sequence $(q_k)$ (for a gentle introduction to numeration
  systems, see the book \cite{2014:formal_languages_automata_and_numeration_systems_2}). However, we emphasize that the
  Ostrowski expansion of a number $n$ does not necessarily coincide with the greedy expansion of $n$ with respect to
  $(q_k)$ as indicated by the following example.

  \begin{example}\label{ex:greedy_fail}
    Consider the nonregular directive word $\Delta = 010(201)^\omega$. Then
    $(q_k) = (1, 2, 3, 7, \ldots)$, so the greedy expansion of the number $6$ is $002$ since
    $6 = 0 \cdot 1 + 0 \cdot 2 + 2 \cdot 3$. Now $\infw{c}_\Delta = 0100102 \dotsm$, so $T^6(\infw{c}_\Delta) = 2 \dotsm$
    and $c_1 = 1$. As in the proof of \autoref{prp:ostrowski_prefix}, the prefix $01001$ of $\infw{c}_\Delta$ is an
    $L_0$-image and $T^6(\infw{c}_\Delta) = T(L_0(T^3(\infw{c}_{\Delta'})))$ where $\Delta' = 10(201)^\omega$. Next
    $\infw{c}_{\Delta'} = 1012 \dots$, so $c_2 = 1$ and $T^3(\infw{c}_{\Delta'}) = T(L_1(T(\infw{c}_{\Delta''})))$ with
    $\Delta'' = 0(201)^\omega$. As the intercept of $\infw{c}_{\Delta''}$ is $0^\omega$, we find that
    $\rep[\Delta]{6} = 111$, so the Ostrowski expansion of $6$ is different from its greedy expansion.
  \end{example}

  Our next aim is to prove that the Ostrowski expansion of an integer coincides with the greedy expansion in the
  important special case of regular episturmian words. We leave the characterization open in the case of nonregular
  directive words. Whenever we discuss greedy expansions below, we assume that the greedy expansion is written the
  least significant digit first and without trailing zeros.

  \begin{definition}\label{def:ostrowski_conditions}
    Given a directive word $\Delta$ as in \eqref{eq:dw_multiplicative}, an infinite word $c_1 c_2 \dotsm$ over the
    alphabet $\{0, 1, 2, \ldots\}$ satisfies the \emph{Ostrowski conditions} if $0 \leq c_k \leq a_k$ for all $k$ and
    for all $k \geq 1$ the following implication holds:
    \begin{equation*}
      \text{$P(r_k + 1)$ exists and $c_i = a_i$ for all $i$ such that $j(k) < i \leq k$} \Longrightarrow c_{j(k)} = 0.
    \end{equation*}
    A finite word $c_1 \dotsm c_k$ satisfies the Ostrowski conditions if it is a prefix of an infinite word satisfying
    the Ostrowski conditions.
  \end{definition}

  If $\Delta$ is regular with period $d$, then the Ostrowski conditions state that $0 \leq c_k \leq a_k$ for all $k$
  and that if $c_i = a_i$ for all $i$ such that $k-(d-1) < i \leq k$ and $k \geq d$, then $c_{k-(d-1)} = 0$.

  \begin{lemma}\label{lem:regular_ostrowski_conditions}
    The intercept of a regular episturmian word satisfies the Ostrowski conditions.
  \end{lemma}
  \begin{proof}
    Let $c_1 c_2 \dotsm$ be the intercept of a regular episturmian word $\infw{t}$ with directive word $\Delta$ as in
    \eqref{eq:dw_multiplicative}. The property that $0 \leq c_k \leq a_k$ for all $k$ follows directly from the
    derivation of the intercept preceding \autoref{lem:standard_intercept}. Suppose that $P(r_k+1)$ exists, that is,
    say $k \geq d$, and assume that $c_i = a_i$ for all $i$ such that $k - (d - 1) < i \leq k$. Let $z$ be the first
    letter of $\infw{t}_{r_k}$. From \eqref{eq:episturmian_s_adic} and the discussion following it, we see that
    \begin{equation}\label{eq:prefix}
      T^{c_1} L_{x_1}^{a_1} \circ \dotsm \circ T^{c_k} L_{x_k}^{a_k}(z)
    \end{equation}
    is a nonempty prefix of $\infw{t}$. Suppose for a contradiction that there exists a largest $\ell$ such that
    $z = x_\ell$ with $k-(d-1) < \ell \leq k$. Thus $T^{c_i} L_{x_i}^{a_i}(z) = T^{a_i}(x_i^{a_i} z) = z$ when
    $i > \ell$. It follows that
    \begin{equation*}
      T^{c_\ell} L_{x_\ell}^{a_\ell} \circ \dotsm \circ T^{c_k} L_{x_k}^{a_k}(z) = T^{c_\ell} L_{x_\ell}^{a_\ell}(z) = \varepsilon
    \end{equation*}
    because $z = x_\ell$ and $c_\ell = a_\ell > 0$. This contradicts that \eqref{eq:prefix} is nonempty. Since the
    directive word $\Delta$ contains $d$ distinct letters, we deduce that $z = x_{k-(d-1)}$. Thus $z = x_{k+1}$ for
    $x_1 x_2 \dotsm$ is $d$-periodic. Clearly
    $\smash[t]{T^{c_i} L_{x_i}^{a_i}(x_{k+1}) = T^{a_i}(x_i^{a_i} x_{k+1}) = x_{k+1}}$ for all $i$ such that
    $k-(d-1) < i \leq k$. Since the prefix \eqref{eq:prefix} must be nonempty, we deduce that
    $\smash[t]{T^{c_{k-(d-1)}} L_{x_{k-(d-1)}}^{a_{k-(d-1)}}(x_{k+1}) \neq \varepsilon}$. As $x_{k-(d-1)} = x_{k+1}$,
    the only option is that $c_{k-(d-1)} = 0$. Therefore $c_1 c_2 \dotsm$ satisfies the Ostrowski conditions.
  \end{proof}

  The arguments presented so far generalize those of \cite{2006:initial_powers_of_sturmian_sequences}. In the case of
  Sturmian words, the sequence of partial quotients $(a_k)$ can be viewed as the continued fraction expansion
  $[0; a_1, a_2, \ldots]$ of a number $\alpha$ which equals the frequency of the letter $x_2$ in $\infw{t}$. A Sturmian
  word $\infw{t}$ can be seen as a coding of the rotation $x \mapsto x + \alpha$ of a point $\rho$ in the torus
  $[-\alpha, 1-\alpha)$. The point $\rho$, often called the intercept of $\infw{t}$ due to the connection to so-called
  mechanical words, can be expressed as a sum
  \begin{equation*}
    \sum_{k=1}^\infty c_k(q_{k-1} \alpha - p_{k-1})
  \end{equation*}
  where $p_{k-1}/q_{k-1}$ are the convergents of $\alpha$ and $(c_k)$ is an integer sequence satisfying the conditions
  $0 \leq c_k \leq a_k$ for all $k$ and $c_{k+1} = a_{k+1} \Longrightarrow c_k = 0$ for all $k$. The sequence $(c_k)$
  is exactly the intercept of $\infw{t}$ in the sense we defined above and the conditions match the Ostrowski
  conditions. Moreover, the denominators of the convergents match the lengths of the associated standard words. The sum
  representation of $\rho$ often goes by the name of Ostrowski expansion of a real number. For the proofs of these
  facts and more details, see \cite{2006:initial_powers_of_sturmian_sequences}. Thus the concepts we defined have nice
  number-theoretic interpretations in the context of Sturmian words. Unfortunately no such interpretations are known
  for general episturmian words. See \cite[Sect.~5]{2001:episturmian_words_and_some_constructions_of_de} for an
  intercept defined for episturmian words that are fixed points of morphisms.

  Equalities like \eqref{eq:episturmian_s_adic} make sense even when the word $c_1 c_2 \dotsm$ does not satisfy the
  Ostrowski conditions. A whole theory of so-called spinned directive words has been developed for studying alternative
  representations of episturmian words; see \cite[Sect.~4]{2009:episturmian_words_a_survey}. Our notion of intercept
  coincides with the normalized directive word of \cite{2008:quasiperiodic_and_lyndon_episturmian_words}.

  \begin{lemma}\label{lem:ostrowski_ub}
    Let $\Delta$ be a directive word as in \eqref{eq:dw_multiplicative}. If $c_1 \dotsm c_k$ satisfies the Ostrowski
    conditions, then $\val[\Delta]{c_1 \dotsm c_k} < q_k$.
  \end{lemma}
  \begin{proof}
    Let us prove the claim by induction on $k$. The base case is established by observing that
    $\val{\varepsilon} = 0 < q_0 = 1$. Say $P(r_k+1)$ exists, and set $j = j(k)$. Assume first that there exists a
    largest $i$ such that $j < i \leq k$ and $c_i < a_i$. Then
    \begin{align*}
      \val{c_1 \dotsm c_k} &= \val{c_1 \dotsm c_{i-1}} + c_i q_{i-1} + a_{i+1} q_i + \dotsm + a_k q_{k-1} \\
      &< q_{i-1} + (a_i - 1)q_{i-1} + a_{i+1} q_i + \dotsm + a_k q_{k-1} \\
      &\leq a_{j+1} q_j + \dotsm + a_k q_{k-1} \\
      &< q_k
    \end{align*}
    where the first inequality follows from the induction hypothesis and the final inequality follows from
    \eqref{eq:q_recursion_1}. Suppose then that no $i$ like above exists. The Ostrowski conditions now imply that
    $c_j = 0$, so
    \begin{align*}
      \val{c_1 \dotsm c_k} &= \val{c_1 \dotsm c_{j-1}} + a_{j+1} q_j + \dotsm + a_k q_{k-1} \\
      &< q_{j-1} + a_{j+1} q_j + \dotsm + a_k q_{k-1} \\
      &= q_k
    \end{align*}
    by the induction hypothesis and \eqref{eq:q_recursion_1}.

    Suppose then that $P(r_k+1)$ does not exist. Then
    \begin{equation*}
      \val{c_1 \dotsm c_k} \leq \val{a_1 \dotsm a_k} = a_1 q_0 + \dotsm a_k q_{k-1} < q_k
    \end{equation*}
    by \eqref{eq:q_recursion_2}. The claim follows.
  \end{proof}

  \begin{proposition}\label{prp:greedy_ostrowski}
    Let $\Delta$ be a regular directive word. Let $c_1 \dotsm c_k$ be the greedy expansion of a nonnegative integer $n$
    with respect to the numeration system associated with $(q_k)$. Then $\rep[\Delta]{n} = c_1 \dotsm c_k$.
  \end{proposition}
  \begin{proof}
    If $n = 0$, then the claim is clear as the greedy expansion of $0$ is $\varepsilon$ by convention. Suppose that
    $n > 0$. Let $\ell$ be the largest integer such that $q_{\ell-1} \leq n$. Then there exists a unique nonnegative
    integers $b_\ell$ and $r_{\ell-1}$ such that $n = b_\ell q_{\ell-1} + r_{\ell-1}$ with $r_{\ell-1} < q_{\ell-1}$.
    Writing similarly $r_i = b_i q_{i-1} + r_{i-1}$ for $i = \ell-1, \ldots, 1$ yields the greedy expansion
    $b_1 \dotsm b_\ell$ of $n$ with respect to $(q_k)$. It is evident that $\val[\Delta]{b_1 \dotsm b_i} < q_i$ for
    $i = 1, \ldots, \ell$. Let then $\rep[\Delta]{n} = c_1 \dotsm c_k$. It follows from \autoref{lem:ostrowski_ub} that
    $\val[\Delta]{c_1 \dotsm c_i} < q_i$ for $i = 1, \ldots, k$. For the claim it suffices to prove that
    $b_1 \dotsm b_\ell = c_1 \dotsm c_k$.

    Without loss of generality, we may assume that $k \leq \ell$. If $k < \ell$, then
    $\val{b_1 \dotsm b_\ell} \geq \val{0^{\ell-1}1} = q_{\ell-1}$. Then $\val{c_1 \dotsm c_k} < q_k \leq q_{\ell-1}$
    which contradicts that $\val[\Delta]{b_1 \dotsm b_\ell} = \val[\Delta]{c_1 \dotsm c_k} = n$. Therefore $k = \ell$.
    Suppose by symmetry that $c_k \leq b_\ell$. If $c_k < b_\ell$, then $b_1 \dotsm b_{\ell-1} (b_\ell - c_k)$ and
    $c_1 \dotsm c_{k-1} 0$ represent the same number which is impossible by what we just argued. Thus $b_\ell = c_k$.
    By repeating the argument for the words $c_1 \dotsm c_{k-1}$ and $b_1 \dotsm b_{\ell-1}$, which represent the same
    number, we see that $c_1 \dotsm c_k = b_1 \dotsm b_\ell$. The claim follows.
  \end{proof}

  \autoref{prp:greedy_ostrowski} allows easy computation of Ostrowski expansions in the case of regular directive
  words. For example, if $\Delta = (012)^\omega$, then $(q_k) = (1, 2, 4, 7, 13, 24, 44, \ldots)$ and
  $\rep[\Delta]{7} = 0001$ and $\rep[\Delta]{10} = 1101$ as $10 = 1 + 2 + 7$.

  Observe how the proof of \autoref{prp:greedy_ostrowski} fails in \autoref{ex:greedy_fail}. We have $\rep{6} = 111$,
  but it is not true that $\val{11} = q_0 + q_1 = 1 + 2 = 3 < q_2 = 3$. Thus a result akin to
  \autoref{lem:ostrowski_ub} is not valid for general directive words.

  \subsection{Auxiliary Results on Generalized Standard Words}
  In this section, we prove further results on generalized standard words needed in this paper. Most of the presented
  results appear in some form in \cite[Sect.~3.3]{2002:episturmian_words_and_episturmian_morphisms}; our proofs follow
  a somewhat different philosophy being based purely on properties of the numeration system. See also
  \cite{2004:episturmian_words_shifts_morphisms_and_numeration} on additional results on the numeration system.

  \autoref{prp:ostrowski_prefix} states that if in an episturmian word has intercept $c_1 \dotsm c_k 0^\omega$, then it
  equals $T^{\val{c_1 \dotsm c_k}}(\infw{c}_\Delta)$. We generalize this to arbitrary intercepts as follows. This
  result is found as \cite[Thm.~3.20]{2001:episturmian_words_and_some_constructions_of_de}.

  \begin{theorem}\label{thm:ostrowski_val_limit}
    If $\infw{t}$ is an episturmian word with directive word $\Delta$ as in \eqref{eq:dw_multiplicative} and intercept
    $c_1 c_2 \dotsm$, then
    \begin{equation*}
      \infw{t} = \lim_{k\to\infty} T^{\val{c_1 \dotsm c_k}}(\infw{c}_\Delta).
    \end{equation*}
  \end{theorem}

  For the proof, we need the following lemma.

  \begin{lemma}\label{lem:ostrowski_not_all_up}
    Let $\Delta$ be a directive word as in \eqref{eq:dw_multiplicative} and $c_1 c_2 \dotsm$ be an intercept. Then
    there exists an integer $k$ such that $c_k < a_k$.
  \end{lemma}
  \begin{proof}
    This proof is similar to that of \autoref{lem:regular_ostrowski_conditions}. Let $z_k$ be the first letter of
    $\infw{t}_{r_k}$ where $\infw{t}_{r_k}$ is as in \eqref{eq:episturmian_s_adic}. Since $\Delta$ contains finitely
    many distinct letters and is not eventually constant, there must exist integers $j$ and $k$ such that $j < k$,
    $x_j = z_k$, and $x_i \neq z_k$ when $j < i \leq k$. If $c_i < a_i$ for some $i$ such that $j < i \leq k$, then the
    claim is clear, so assume that $c_i = a_i$ when $j < i \leq k$. It follows that
    $T^{c_i} L_{x_i}^{a_i}(z_k) = T^{a_i}(x_i^{a_i} z_k) = z_k$ for all $i$ such that $j < i \leq k$. Therefore
    \begin{equation*}
      T^{c_j} L_{x_j}^{a_j} \circ \dotsm T^{c_k} L_{x_k}^{a_k}(z_k) = T^{c_j} L_{x_j}^{a_j}(z_k) = T^{c_j}(z_k)
    \end{equation*}
    since $x_j = z_k$. It follows from \eqref{eq:episturmian_s_adic} that $T^{c_j}(z_k)$ must be nonempty, so we
    conclude that $c_j = 0$. The claim follows.
  \end{proof}

  \begin{proof}[Proof of \autoref{thm:ostrowski_val_limit}]
    Since the suffix of an intercept is a valid intercept, \autoref{lem:ostrowski_not_all_up} implies that there exists
    an increasing integer sequence $(k_n)$ such that $c_{k_n} < a_{k_n}$ for all $n$. From
    \eqref{eq:episturmian_s_adic}, we have
    \begin{equation*}
      \infw{t} = T^{c_1} L_{x_1}^{a_1} \dotsm T^{c_{k_n-1}} L_{x_{k_n-1}}^{a_{k_n-1}}(x_{k_n}^{a_{k_n}-c_{k_n}} z_{k_n} \dotsm)
    \end{equation*}
    for a letter $z_{k_n}$. Thus $\infw{t}$ and the episturmian word with intercept $c_1 \dotsm c_{k_n} 0^\omega$
    have common prefix
    \begin{equation*}
      T^{c_1} L_{x_1}^{a_1} \dotsm T^{c_{k_n-1}} L_{x_{k_n-1}}^{a_{k_n-1}}(x_{k_n}^{a_{k_n}-c_{k_n}} z_{k_n}).
    \end{equation*}
    By definition, we have
    $\abs{T^{c_1} L_{x_1}^{a_1} \dotsm T^{c_{k_n-1}} L_{x_{k_n-1}}^{a_{k_n-1}}(x_{k_n})} \geq 1$, so
    this common prefix equals
    \begin{equation*}
      T^{c_1} L_{x_1}^{a_1} \dotsm T^{c_{k_n-1}} L_{x_{k_n-1}}^{a_{k_n-1}}(x_{k_n}) \cdot
      \tau_{k_n-1}(x_{k_n}^{a_{k_n}-c_{k_n}-1} z_{k_n}).
    \end{equation*}
    The length of this common prefix is thus has at least $\abs{\tau_{k_n-1}(z_{k_n})}$,
    and this length tends to infinity as $n \to \infty$ provided that the directive sequence $\Delta$ is not ultimately
    constant (we always assume this). If we denote the episturmian word with intercept $c_1 \dotsm c_k 0^\omega$ by
    $\infw{t}_k$, we have $\infw{t} = \lim_{k\to\infty} \infw{t}_k$. Since
    $\infw{t}_k = T^{\val{c_1 \dotsm c_k}}(\infw{c}_\Delta)$ by \autoref{prp:ostrowski_prefix}, the claim follows.
  \end{proof}

  The significance of \autoref{thm:ostrowski_val_limit} is that, in principle, the properties of a general
  episturmian word reduce to those of shifts of an epistandard word. This result suggests to consider the longest
  common prefixes of the words in the sequence $(T^{\val{c_1 \dotsm c_k}}(\infw{c}_\Delta))_k$. This is found in
  \autoref{lem:common_prefix_length}, but we need several auxiliary lemmas for the proof.

  \begin{lemma}\label{lem:u_prefix_s}
    The word $u_{r_k}$ is a proper prefix of $s_k$ for all $k \geq 1$.
  \end{lemma}
  \begin{proof}
    If $k = 1$, then $u_{r_k} = s_0^{a_1 - 1}$ by \eqref{eq:u_2} and \eqref{eq:u_4} and $s_1 = s_0^{a_1} x_2$, so the
    claim holds. Assume that $k > 1$. Say $P(r_k+1)$ exists. Since both $u_{r_k}$ and $s_k$ are prefixes of
    $\infw{c}_\Delta$, it suffices to show that $\abs{u_{r_k}} < \abs{s_k}$. By \eqref{eq:u_2} and \eqref{eq:u_4}, we
    have $\abs{u_{r_k}} = (a_k-1)q_{k-1} + a_{k-1} q_{k-2} + \dotsm + a_1 q_0$. Moreover, we have
    $q_k = a_k q_{k-1} + \dotsm + a_{j(k)+1} q_{j(k)} + q_{j(k)-1}$ by \eqref{eq:s_recursion_1}. Thus
    $\abs{s_k} - \abs{u_{r_k}} \geq q_{k-1} - \abs{u_{r_{j(k)}}}$. Since $j(k) < k$, we have
    $\abs{u_{r_{j(k)}}} \leq \abs{u_{r_{k-1}}}$. Thus by the induction hypothesis, we have
    $\abs{s_k} - \abs{u_{r_k}} \geq q_{k-1} - \abs{u_{r_{k-1}}} > 0$. When $P(r_k+1)$ does not exist, the claim follows
    by similar arguments.
  \end{proof}

  \begin{lemma}\label{lem:prefix_power_complete}
    We have the following implications for all $k \geq 0$.
    \begin{enumerate}[(i)]
      \item If $P(r_{k+1}+1)$ and $P(r_k+1)$ exist, then
            $s_{k+1} s_k = s_k^{\s{a}+1} u_{r_{j(k)}} u_{r_{j(k+1)}}^{-1} s_k$.
      \item If $P(r_{k+1}+1)$ exists and $P(r_k+1)$ does not exist, then
            $s_{k+1} s_k = s_k^{\s{a}+1} x_{k+1}^{-1} u_{r_{j(k+1)}}^{-1} s_k$.
      \item If $P(r_{k+1}+1)$ does not exist and $P(r_k+1)$ exists, then
            $s_{k+1} s_k = s_k^{\s{a}+1} u_{r_{j(k)}} x_{k+2} s_k$.
      \item If neither $P(r_{k+1}+1)$ nor $P(r_k+1)$ exists, then
            $s_{k+1} s_k = s_k^{\s{a}+1} x_{k+1}^{-1} x_{k+2} s_k$.
    \end{enumerate}
  \end{lemma}
  \begin{proof}
    Suppose that $P(r_{k+1}+1)$ exists. As $j(k+1) < k+1$, we see that $u_{r_{j(k+1)}}$ is a prefix of $s_k$. Applying
    \eqref{eq:s_recursion_1} and \eqref{eq:u_4}, we have
    \begin{equation*}
      (s_k^{\s{a}})^{-1} s_{k+1} s_k
      = s_{k-1}^{\s{a}} \dotsm s_{j(k+1)}^{\s{a}} s_{j(k+1)-1} \cdot u_{r_{j(k+1)}} u_{r_{j(k+1)}}^{-1} s_k
      = s_{k-1}^{\s{a}} \dotsm s_0^{\s{a}} \cdot u_{r_{j(k+1)}}^{-1} s_k.
    \end{equation*}
    When $P(r_k+1)$ exists, then $s_{k-1}^{\s{a}} \dotsm s_0^{\s{a}} = s_k u_{r_{j(k)}}$ and (i) holds. If $P(r_k+1)$
    does not exist, then we deduce from \eqref{eq:s_recursion_2} that
    \begin{equation*}
      (s_k^{\s{a}})^{-1} s_{k+1} s_k = s_k x_{k+1}^{-1} u_{r_{j(k+1)}}^{-1} s_k,
    \end{equation*}
    so (ii) holds.

    Suppose then that $P(r_{k+1}+1)$ does not exist. Then $s_{k+1} = s_k^{\s{a}} \dotsm s_0^{\s{a}} x_{k+2}$ by
    \eqref{eq:s_recursion_2}. If $P(r_k+1)$ exists, then $s_{k-1}^{\s{a}} \dotsm s_0^{\s{a}} = s_k u_{r_{j(k)}}$ like
    above, and we have
    \begin{equation*}
      (s_k^{\s{a}})^{-1} s_{k+1} s_k = s_k u_{r_{j(k)}} x_{k+2} s_k,
    \end{equation*}
    so (iii) holds. If $P(r_k+1)$ does not exist, then
    \begin{equation*}
      (s_k^{\s{a}})^{-1} s_{k+1} s_k = s_k x_{k+1}^{-1} x_{k+2} s_k,
    \end{equation*}
    and we see that (iv) holds.
  \end{proof}

  \begin{lemma}\label{lem:prefix_powers}
    For all $k \geq 0$, the word $s_k^{\s{a}+1}$ is a prefix of $\infw{c}_\Delta$ if and only if $P(r_k+1)$ exists.
    If $P(r_k+1)$ does not exist, then $s_k^{\s{a}+1} x_{k+1}^{-1} x_{k+2}$ is a prefix of $\infw{c}_\Delta$. The
    word $s_k^{\s{a}+2}$ is not a prefix of $\infw{c}_\Delta$.
  \end{lemma}
  \begin{proof}
    By \autoref{lem:u_prefix_s}, the word $\smash[b]{u_{r_{j(k+1)}}}$ is a proper prefix of $s_k$. Say $P(r_k+1)$
    exists. \autoref{lem:prefix_power_complete} implies that the word $s_k^{\s{a}+1}$ is a prefix of $s_{k+1} s_k$.
    Since $s_{k+1} s_k$ is a prefix of $\infw{c}_\Delta$, we see that $\smash[b]{s_k^{\s{a}+1}}$ is a prefix of
    $\infw{c}_\Delta$. Suppose that $P(r_k+1)$ does not exist. Then (ii) or (iv) of \autoref{lem:prefix_power_complete}
    holds. In the latter case, the word $\smash[b]{s_k^{\s{a}+1}}$ cannot be a prefix of $s_{k+1} s_k$ because
    $x_{k+1} \neq x_{k+2}$ by definition. In the former case, we have
    $\smash[b]{s_{k+1} s_k = s_k^{\s{a}+1} x_{k+1}^{-1} u_{r_{j(k+1)}}^{-1} s_k}$. Since $s_k$ is a prefix of
    $\infw{c}_\Delta$, we see that the words $s_k$ and $\smash[b]{u_{r_{j(k+1)}+1}}$ share the prefix
    $\smash[b]{u_{r_{j(k+1)}} x_{j(k+1)}}$. Therefore $s_{k+1} s_k$ has prefix
    $\smash[b]{s_k^{\s{a}+1} x_{k+1}^{-1} x_{j(k+1)}}$. By definition, we have $x_{j(k+1)} = x_{k+2}$, so we see that
    $\smash[b]{s_k^{\s{a}+1} x_{k+1}^{-1} x_{k+2}}$ is a prefix of $\infw{c}_\Delta$. Like above, the word
    $\smash[t]{s_k^{\s{a}+1}}$ is not a prefix of $\infw{c}_\Delta$. Thus we have proved the first and second claims.

    Let us then prove the final claim. If $P(r_k+1)$ does not exist, then the second claim shows that
    $s_k^{\s{a}+1}$ is not a prefix of $\infw{c}_\Delta$, so assume that $P(r_k+1)$ exists. Suppose first that
    $P(r_{k+1}+1)$ also exists. The first claim shows that $\smash[b]{s_{k+1}^2}$ is a prefix of $\infw{c}_\Delta$.
    From \autoref{lem:prefix_power_complete}, we see that
    \begin{equation*}
      s_{k+1}^2 = s_k^{\s{a}+1} u_{r_{j(k)}} u_{r_{j(k+1)}}^{-1} s_k^{\s{a}+1} u_{r_{j(k)}} u_{r_{j(k+1)}}^{-1}.
    \end{equation*}
    Because $s_k$ is a prefix of $\infw{c}_\Delta$ and $u_{r_{j(k+1)}}$ is a proper prefix of $s_k$ by
    \autoref{lem:u_prefix_s}, we see that the prefix $\smash[t]{s_k^{\s{a}+1}}$ of $\smash[t]{s_{k+1}^2}$ is followed
    by $u_{r_{j(k)}} x_{k+2}$. Similarly, the word $s_k$ has the word $u_{r_{j(k)}} x_{k+1}$ as a prefix. Since
    $x_{k+1} \neq x_{k+2}$, we conclude that $\smash[t]{s_k^{a_{k+1}+2}}$ is not a prefix of $\infw{c}_\Delta$.

    Assume then that $P(r_{k+1}+1)$ does not exist. Then $\smash[t]{s_{k+1} = s_k^{\s{a}+1} u_{r_{j(k)}} x_{k+2}}$.
    Since $s_{k+1}^2$ is a prefix of $\infw{c}_\Delta$ and $u_{r_{j(k)}}$ is a proper prefix of $s_k$, we see that if
    $\smash[t]{s_k^{\s{a}+2}}$ is a prefix of $\infw{c}_\Delta$, then $x_{k+2} = x_{j(k)}$. Because
    $x_{j(k)} = x_{k+1}$ and $x_{k+1} \neq x_{k+2}$, we conclude that $\smash[t]{s_k^{\s{a}+2}}$ is not a prefix of
    $\infw{c}_\Delta$.
  \end{proof}

  Based on \autoref{lem:prefix_powers}, we give the following definition.

  \begin{definition}
    If $P(r_k+1)$ exists, then we let $t_k$ to be the longest word such that $s_k^{\s{a}+1} t_k$ is a prefix of
    $\infw{c}_\Delta$ with period $q_k$. If $P(r_k+1)$ does not exist, then we set $t_k = x_{k+1}^{-1}$.
  \end{definition}

  Notice that the word $t_k$ is a proper prefix of $s_k$ when $P(r_k+1)$ exists, but $t_k$ can be empty.

  \begin{lemma}\label{lem:st_which_u}
    For all $k \geq 0$, we have $s_k^i t_k = u_{r_k+i}$ for all $i$ such that $1 \leq i \leq a_{k+1}$.
  \end{lemma}
  \begin{proof}
    Suppose that $P(r_k+1)$ exists. By \autoref{lem:prefix_powers}, the word $s_k^2$ is a prefix of $\infw{c}_\Delta$,
    so $s_k t_k$ is a right special prefix of $\infw{c}_\Delta$. Therefore $s_k t_k = u_{\ell+2}$ for some integer
    $\ell$ (recall that the central words $u_k$ are as bispecial factors exactly the right special prefixes of
    $\infw{c}_\Delta$). It follows from \cite[Eq.~2]{2002:episturmian_words_and_episturmian_morphisms} that
    $\smash[t]{u_{\ell+2} = L_{x_1}(u_{\ell+1}^{(1)}) x_1}$ where $\smash[t]{u_{\ell+1}^{(1)}}$ is the $(\ell+1)$th
    central word associated with the directive word $T(\Delta)$. The word $u_{\ell+2}$ can be repeatedly decoded
    like this for a total of $\ell+1$ times to obtain the empty word $\smash[t]{u_1^{(\ell+1)}}$. On the other hand, we
    have $s_k t_k = \mu_{r_k}(x_{k+1}) t_k$, so it must be possible to decode $s_k t_k$ at least $r_k$ times before
    obtaining an empty word. Therefore $r_k + 1 \leq \ell + 2$.

    Assume for a contradiction that $r_k + 1 < \ell + 2$, that is, $r_k \leq \ell$. By \eqref{eq:s_recursion_1}, we may
    write
    \begin{equation*}
      s_k t_k = s_{k-1}^{\s{a}} \dotsm s_{j(k)}^{\s{a}} s_{j(k)-1} t_k = h_{r_k-1} \dotsm h_{r_{j(k)}-1} t_k.
    \end{equation*}
    On the other hand, by \eqref{eq:u_3}, we have $s_k t_k = h_{\ell} \dotsm h_0$. As $r_k \leq \ell$, the word
    $h_{\ell} \dotsm h_0$ contains the product $h_{r_k-1} \dotsm h_{r_{j(k)}-1}$, and we see that
    $\abs{t_k} \geq \abs{h_\ell} \geq \abs{h_{r_k}} = q_k$. This contradicts the definition of $t_k$. We have thus
    proved that $s_k t_k = u_{r_k+1}$. The rest of the claim follows by observing from \eqref{eq:u_2} that the
    differences $\abs{u_{r_k+i+1}} - \abs{u_{r_k+i}}$ equal $q_k$ for $i = 1, \ldots, a_{k+1} - 1$. These correspond to
    the differences $\abs{s_k^{i+1} t_k} - \abs{s_k^i t_k}$ for $i = 1, \ldots, a_{k+1} - 1$.

    Let us then assume that $P(r_k+1)$ does not exist. \autoref{lem:prefix_powers} again implies that $s_k t_k$ is a
    central word. By \eqref{eq:s_recursion_2}, we have $s_k = s_{k-1}^{\s{a}} \dotsm s_0^{\s{a}} x_{k+1}$ so
    $s_k t_k = s_{k-1}^{\s{a}} \dotsm s_0^{\s{a}} = h_{r_k-1} \dotsm h_0 = u_{r_k+1}$ because
    $\smash[t]{t_k = x_{k+1}^{-1}}$. The rest of the claim follows as above.
  \end{proof}

  The following lemma was proved for Sturmian words in \cite[Prop.~5.2..2]{diss:caius_wojcik}.

  \begin{lemma}\label{lem:common_prefix_length}
    Let $\Delta$ be a regular directive word and $c_1 c_2 \dotsm$ be an intercept. Let $k \geq 1$, and assume that
    there exists a least positive integer $n$ such that $c_{k+n} \neq 0$. The length of the longest common prefix of
    $T^{\val{c_1 \dotsm c_k}}(\infw{c}_\Delta)$ and $T^{\val{c_1 \dotsm c_{k+n}}}(\infw{c}_\Delta)$ equals
    $\smash[t]{\abs{s_{k+n-1}^{\s{a}-\s{c}} s_{k+n-2}^{\s{a}} \dotsm s_0^{\s{a}}} - \val{c_1 \dotsm c_k}}$.
  \end{lemma}
  \begin{proof}
    As $\Delta$ is regular, it follows from Lemmas \ref{lem:regular_ostrowski_conditions} and \ref{lem:ostrowski_ub}
    that $q_{k+n-1} - \val{c_1 \dotsm c_k} \geq 0$. Thus we can let $v$ to be the suffix of $s_{k+n-1}$ of length
    $q_{k+n-1} - \val{c_1 \dotsm c_k}$. Suppose that $P(r_{k+n-1}+1)$ exists. Then \autoref{lem:prefix_powers} implies
    that $\smash[tb]{s_{k+n-1}^{\s{a}+1} t_{k+n-1}}$ is the longest prefix of $\infw{c}_\Delta$ with period
    $q_{k+n-1}$. Therefore the word $T^{\val{c_1 \dotsm c_k}}(\infw{c}_\Delta)$ has prefix
    $\smash[tb]{v s_{k+n-1}^{a_{k+n}} t_{k+n-1}}$. Since $c_{k+1} = \ldots = c_{k+n-1} = 0$, the word
    $T^{\val{c_1 \dotsm c_{k+n}}}(\infw{c}_\Delta)$ has prefix $\smash[t]{v s_{k+n-1}^{\s{a} - \s{c}} t_{k+n-1}}$.
    Since $c_{k+n} > 0$, we have by the definition of the word $t_{k+n-1}$ that the longest common prefix of the words
    $\smash[t]{T^{\val{c_1 \dotsm c_k}}(\infw{c}_\Delta)}$ and
    $\smash[t]{T^{\val{c_1 \dotsm c_{k+n}}}(\infw{c}_\Delta)}$ equals
    $\smash[t]{v s_{k+n-1}^{\s{a} - \s{c}} t_{k+n-1}}$. The claim follows by a short computation applying
    \autoref{lem:st_which_u}.

    The same arguments apply when $P(r_{k+n-1}+1)$ does not exist as then the longest prefix of $\infw{c}_\Delta$ with
    period $q_{k+n-1}$ equals $\smash[t]{s_{k+n-1}^{\s{a}+1} t_{k+n-1}}$ and $t_{k+n-1} = x_{k+n}^{-1}$. Notice that
    $\abs{x_{k+n}^{-1}} = -1$.
  \end{proof}

  Let us then find out when the length of the longest common prefix increases. Let $k \geq 1$ and $n_1$ and $n_2$ be
  the two least positive integers such that $n_1 < n_2$, $c_{k+n_1} \neq 0$, and $c_{k+n_2} \neq 0$. Let $v_1$ be the
  longest common prefix of $\smash[t]{T^{\val{c_1 \dotsm c_k}}(\infw{c}_\Delta)}$ and
  $\smash[t]{T^{\val{c_1 \dotsm c_{k+n_1}}}(\infw{c}_\Delta)}$ and $v_2$ be the longest common prefix of
  $\smash[t]{T^{\val{c_1 \dotsm c_{k+n_1}}}(\infw{c}_\Delta)}$ and
  $\smash[t]{T^{\val{c_1 \dotsm c_{k+n_2}}}(\infw{c}_\Delta)}$. By \autoref{lem:common_prefix_length}, we have
  \begin{equation*}
    \abs{v_2} - \abs{v_1} = \abs{s_{k+n_2-1}^{\s{a}-\s{c}} s_{k+n_2-2}^{\s{a}} \dotsm s_0^{\s{a}}} -
    \abs{s_{k+n_1-1}^{\s{a}} s_{k+n_1-2}^{\s{a}} \dotsm s_0^{\s{a}}},
  \end{equation*}
  so $\abs{v_2} - \abs{v_1} > 0$ unless $n_2 = n_1 + 1$ and $a_{k+n_2} = c_{k+n_2}$.

  \begin{lemma}\label{lem:common_prefix_length_increasing}
    Let $\Delta$ be a regular word and $c_1 c_2 \dotsm$ be an intercept. Define
    \begin{equation*}
      \eta_k = \abs{s_k^{\s{a}-\s{c}} s_{k-1}^{\s{a}} \dotsm s_0^{\s{a}}} - \val{c_1 \dotsm c_k}
    \end{equation*}
    for all $k \geq 0$. The sequence $(\eta_k)$ is nondecreasing and $\lim_{k\to\infty} \eta_k = \infty$.
  \end{lemma}
  \begin{proof}
    Let $k \geq 0$. Then
    \begin{align*}
      \eta_{k+1} &= \abs{s_{k+1}^{a_{k+2}-c_{k+2}} s_k^{a_{k+1}} \dotsm s_0^{a_1}} - \val{c_1 \dotsm c_{k+1}} \\
                 &\geq \abs{s_k^{a_{k+1}} \dotsm s_0^{a_1}} - \val{c_1 \dotsm c_{k+1}} \\
                 &= \abs{s_k^{a_{k+1} - c_{k+1}} s_{k-1}^{a_k} \dotsm s_0^{a_1}} - \val{c_1 \dotsm c_k} \\
                 &= \eta_k,
    \end{align*}
    so $(\eta_k)$ is nondecreasing. From the discussion preceding this lemma, we see that $\eta_{k+1} = \eta_k$ if and
    only if $a_{k+2} = c_{k+2}$. By \autoref{lem:ostrowski_not_all_up}, there exist infinitely many $k$ such that
    $c_{k+2} < a_{k+2}$. Since $(\eta_k)$ is nondecreasing, this shows that $\lim_{k\to\infty} \eta_k = \infty$.
  \end{proof}

  Many additional properties of the words $s_k$ could be easily derived, but we do not need them in this paper, so we
  will stop after the following required result.

  \begin{lemma}\label{lem:left_shift_intercept}\cite[Thm.~3.17]{2002:episturmian_words_and_episturmian_morphisms}
    Let $\Delta$ be a directive word as in \eqref{eq:dw_multiplicative} and $y$ be a letter occurring infinitely many
    times in $\Delta$. Define an intercept $c_1 c_2 \dotsm$ as follows: $c_k = 0$ if $x_k = y$ and $c_k = a_k$ if
    $x_k \neq y$. Then $c_1 c_2 \dotsm$ is the intercept of the word $y \infw{c}_\Delta$.
  \end{lemma}
  \begin{proof}
    Since $y$ occurs infinitely many times in $\Delta$, arbitrarily long central words $u_k$ are followed by the letter
    $y$. As the language of $\infw{c}_\Delta$ is closed under reversal and $u_k$ are palindromes, we see that
    $y u_k$ is a factor of $\infw{c}_\Delta$ for infinitely many $k$. Since $\infw{c}_\Delta$ is the limit of the
    sequence $(u_k)$, it follows that $y \infw{c}_\Delta$ is an episturmian word with directive word $\Delta$.

    Let $d_1 d_2 \dotsm$ be the intercept of $y \infw{c}_\Delta$. Say $y \neq x_1$. Then
    $y \infw{c}_\Delta = y x_1^{a_1} x_2 \dotsm$, so $y \infw{c}_\Delta = T \circ L_{x_1}(y \infw{c}_{\Delta'})$ where
    $\Delta' = T(\Delta)$. It follows that $d_1 = a_1$. Suppose then that $y = x_1$. Now it is plain that
    $y \infw{c}_\Delta = L_{x_1}(y \infw{c}_{\Delta'})$, so we have $d_1 = 0$. It follows from induction that
    $y \infw{c}_\Delta$ has the claimed intercept.
  \end{proof}

  \section{Rauzy Graphs of Episturmian Words}
  Let $\infw{x}$ be an infinite word. The Rauzy graph $\Gamma(n)$ of order $n$ associated with the language of
  $\infw{x}$ is a directed graph with vertices $\Lang[n]{\infw{x}}$ and edges $\Lang[n+1]{\infw{x}}$. There is an edge
  $e$ from vertex $u$ to vertex $v$ if and only if $e$ has prefix $u$ and suffix $v$. Each word with the language
  $\Lang{\infw{x}}$ corresponds to an infinite path in the graph $\Gamma(n)$ starting from its prefix of length $n$. The
  initial nonrepetitive complexity $\irep{\infw{x}}{n}$ can be determined from $\Gamma(n)$: start from the vertex
  corresponding to the prefix of $\infw{x}$ of length $n$ and follow the path dictated by $\infw{x}$ until a vertex is
  repeated for the first time. In general, this is of no help as the graph $\Gamma(n)$ can be very complicated. However,
  when $\infw{x}$ has low factor complexity, there are only few right special factors of length $n$ and the analysis is
  more likely to succeed. This is indeed the case with episturmian words whose Rauzy graphs have especially nice form.

  An episturmian word $\infw{t}$ with directive word $\Delta$ can be equivalently defined as an infinite word such that
  its language is closed under reversal and it has exactly one right special factor of each length
  \cite[Thm.~5]{2001:episturmian_words_and_some_constructions_of_de}. The reversal of the right special factor must
  thus be left special, so there is exactly one left special factor and exactly one right special factor of each
  length. A moment's thought shows that this means that $\Gamma(n)$ is composed of cycles sharing a common part, called
  the \emph{central path}, like in \autoref{fig:rauzy_graph}. The number of cycles depends on the number of letters
  that eventually appear in the directive word. If $\infw{t}$ is regular with period $d$, then there are exactly $d$
  cycles. Indeed, each central word $u_k$ has as a suffix all shorter central words, so each central word is followed
  by each letter $0$, $1$, $\ldots$, $d-1$ in $\infw{c}_\Delta$. The suffixes of the central words yield right special
  factors for each length, and the claim follows. Notice that we just argued that the central words are right special.
  This means that they are left special because they are palindromes. Therefore the central path of $\Gamma(\abs{u_k})$
  reduces to a single vertex. Notice that the graph $\Gamma(n)$ ``evolves'' to $\Gamma(n+1)$ in a deterministic fashion
  whenever $n$ does not equal the length of a central word: the central path is shortened by one edge and all cycles
  maintain their number of edges. When $n$ equals some $\abs{u_k}$, the evolution or ``bursting of the bispecial
  factor'' is determined by $\Delta$. In the case of episturmian words, determining $\irep{\infw{t}}{n}$ is thus rather
  straightforward: find out the location of the vertex $v$ corresponding to the prefix of $\infw{t}$ of length $n$ and
  determine the length $L$ of the next cycle taken. If $v$ is on the central path, then $\irep{\infw{t}}{n}$ equals
  $L$. Otherwise we need to add to $L$ the number of edges that need to be traversed from $v$ to the vertex of the left
  special factor.

  \begin{figure}
    \centering
    \begin{tikzpicture}[> = stealth]
      \tikzstyle{state}=[draw=black,circle,minimum size=1.8em]
      \tikzstyle{arc}=[line width=1pt,color=black]

      \node[state] (l) at (0,0) {$\ell$};
      \node[state] (2) at (2,0) {};
      \node (3) at (4,0) {$\dotsm$};
      \node[state] (4) at (6,0) {};
      \node[state] (r) at (8,0) {$r$};

      \node[state] (5) at (1.5,2) {};
      \node (6) at (4,2.5) {$\dotsm$};
      \node[state] (7) at (6.5,2) {};

      \node[state] (8) at (0.5,4) {};
      \node (9) at (4,4.5) {$\dotsm$};
      \node[state] (10) at (7.5,4) {};

      \node[state] (11) at (1.5,-2) {};
      \node (12) at (4,-2.5) {$\dotsm$};
      \node[state] (13) at (6.5,-2) {};

      \draw[->,arc] (l) -> (2);
      \draw[->,arc] (2) -> (3);
      \draw[->,arc] (3) -> (4);
      \draw[->,arc] (4) -> (r);

      \draw[->,arc] (r) -> (7);
      \draw[->,arc] (7) -> (6);
      \draw[->,arc] (6) -> (5);
      \draw[->,arc] (5) -> (l);

      \draw[->,arc] (r) -> (10);
      \draw[->,arc] (10) -> (9);
      \draw[->,arc] (9) -> (8);
      \draw[->,arc] (8) -> (l);

      \draw[->,arc] (r) -> (13);
      \draw[->,arc] (13) -> (12);
      \draw[->,arc] (12) -> (11);
      \draw[->,arc] (11) -> (l);
    \end{tikzpicture}
    \caption{The Rauzy graph of an episturmian word. The left special factor corresponds to the vertex $\ell$ and the
    right special to the vertex $r$. The directed path from $\ell$ to $r$ is the central path.}
    \label{fig:rauzy_graph}
  \end{figure}
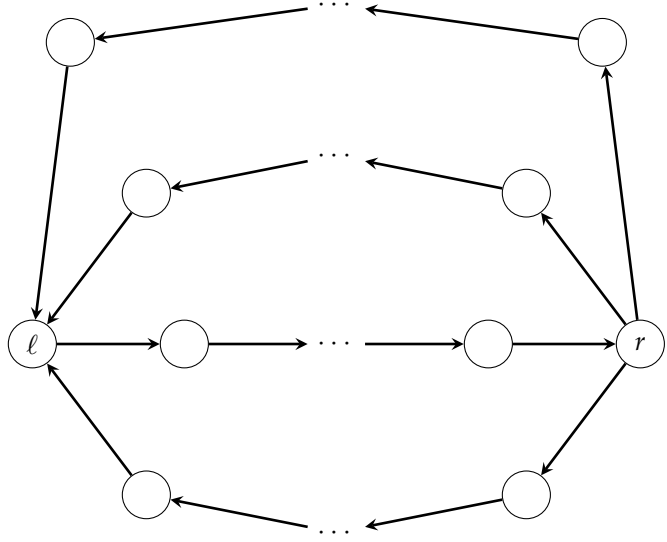

  \section{Initial Nonrepetitive Complexity of Regular Episturmian Words}\label{sec:inrc}
  In this section, we derive a complete description of the initial nonrepetitive complexity of regular episturmian
  words; see \autoref{thm:main}. We specialize the most significant propositions to the case of Sturmian words. Our
  proof method generalizes that of Wojcik who determined the initial nonrepetitive complexity of Sturmian words
  \cite[Sect.~5.3]{diss:caius_wojcik}.

  The majority of the results presented need the assumption that the directive word $\Delta$ is regular. We make the
  convention that this is implicitly assumed in the following discussion, but we make the assumption explicit in the
  statements of lemmas, propositions, etc.

  It is natural to partition the positive integers according to the sequence $(q_k)$, but it is in fact better to do it
  using the central words $u_k$. We set
  \begin{equation*}
    \interval{k} = \{\abs{u_{r_k}} + 1, \abs{u_{r_k}} + 2, \ldots, \abs{u_{r_{k+1}}}\}
  \end{equation*}
  for $k \geq 0$ with the convention $\abs{u_{r_0}} = -1$. Clearly $\N$ is a disjoint union of these intervals. We
  further subdivide each interval $\interval{k}$ into $a_{k+1}$ subintervals by setting
  \begin{equation*}
    \interval{k,\ell} = \{\abs{u_{r_k+\ell}} + 1, \ldots, \abs{u_{r_k+\ell+1}}\}.
  \end{equation*}
  for $\ell = 0, \ldots, a_{k+1}-1$. Notice the following peculiarity: the first subinterval $\interval{k,0}$ has
  $q_{k-1}$ elements while the remaining intervals have $q_k$ elements (when $k > 0$). Indeed by \eqref{eq:u_2}, we
  have $u_{r_k+1} = h_{r_k-1} u_{r_k}$, so $\abs{u_{r_k+1}} - \abs{u_{r_k}} = \abs{h_{r_k-1}} = q_{k-1}$. Similarly
  $\abs{u_{r_k+t}} - \abs{u_{r_k+t-1}} = q_k$ for $t$ such that $2 \leq t \leq a_{k+1}$. The chief reason for defining
  the intervals like this is that when the directive word $\Delta$ is regular with period $d$, we are guaranteed that
  $q_{k+d-1} > \abs{u_{r_{k+1}}}$, that is, $q_{k+d-1}$ exceeds the right endpoint of $\interval{k}$. If we had defined
  the right endpoint of $\interval{k}$ to be $\abs{u_{r_{k+1}+1}}$, which is a priori more natural, this is not true
  when $d = 2$ as indicated by the proof of the following lemma.

  \begin{lemma}\label{lem:right_endpoint_greater}
    Let $\Delta$ be a regular directive word with period $d$. We have $q_{k+d-1} - 1 > \abs{u_{r_{k+1}}}$ when
    $k \geq 0$ and $d \geq 2$. If $d > 2$, then $q_{k+d-1} - 1 > \abs{u_{r_{k+1}+1}}$ for $k \geq 0$.
  \end{lemma}
  \begin{proof}
    Let us first prove the latter claim. By \eqref{eq:u_4} and \eqref{eq:reg_s_2}, we have
    $u_{r_{k+1}+1} = s_k^{\s{a}} \dotsm s_0^{\s{a}}$ and
    $\smash[t]{s_{k+d-1} = s_{k+d-2}^{\s{a}} \dotsm s_k^{\s{a}} s_{k-1}}$. When $d > 2$, the word $s_{k+d-1}$ has
    the factors $s_{k-1+d-1}$ and $\smash[t]{s_k^{\s{a}} s_{k-1}}$. We conclude that if $d > 2$ and
    $q_{k-1+d-1} > \abs{u_{r_k+1}}$, then $q_{k+d-1} - 1 > \abs{u_{r_{k+1}+1}}$. It then suffices to check that
    $q_{d-1} > \abs{u_{r_1+1}}$, but this is trivially true as $\smash[t]{u_{r_1+1} = s_0^{\s{a}}}$ and
    $\smash[t]{s_{d-1} = s_{d-2}^{\s{a}} \dotsm s_0^{\s{a}} x_d}$ by \eqref{eq:reg_s_1}.

    Say $d = 2$. Like above, we have $\smash[t]{u_{r_{k+1}} = s_k^{\s{a}-1} s_{k-1}^{\s{a}} \dotsm s_0^{\s{a}}}$ and
    $\smash[t]{s_{k+1} = s_k^{\s{a}} s_{k-1}}$. Consequently $q_{k+1} - 1 > \abs{u_{r_{k+1}}}$ if and only if
    $\smash[t]{\abs{s_k} - 1 > \abs{s_{k-1}^{\s{a}-1} s_{k-2}^{\s{a}} \dotsm s_0^{\s{a}}}}$. By repeating the argument,
    we see that the claim is true if and only if $\smash[t]{\abs{s_1} - 1 > \abs{s_0^{\s{a}-1}}}$. Since
    $s_1 = s_0^{a_1} x_2$, the claim follows.
  \end{proof}

  We set out to figure out $\irep{\infw{t}}{n}$ for a regular episturmian word $\infw{t}$ when $n \in \interval{k}$ for
  $k \geq 0$. In view of \autoref{thm:ostrowski_val_limit}, the aim is to reduce finding this number to the study of
  shifts of $\infw{c}_\Delta$. In fact, if $\infw{t}$ is regular with period $d$ and intercept $c_1 c_2 \dotsm$, then
  in most cases $\irep{\infw{t}}{n}$ is determined by the word $T^{\val{c_1 \dotsm c_{k+d-1}}}(\infw{c}_\Delta)$; see
  \autoref{prp:sufficient_shift} for the complete details. See \autoref{fig:plot} for example plots of the function
  $\irep{\infw{t}}{n}$.

  \begin{figure}
    \centering
    \includegraphics[width=\textwidth]{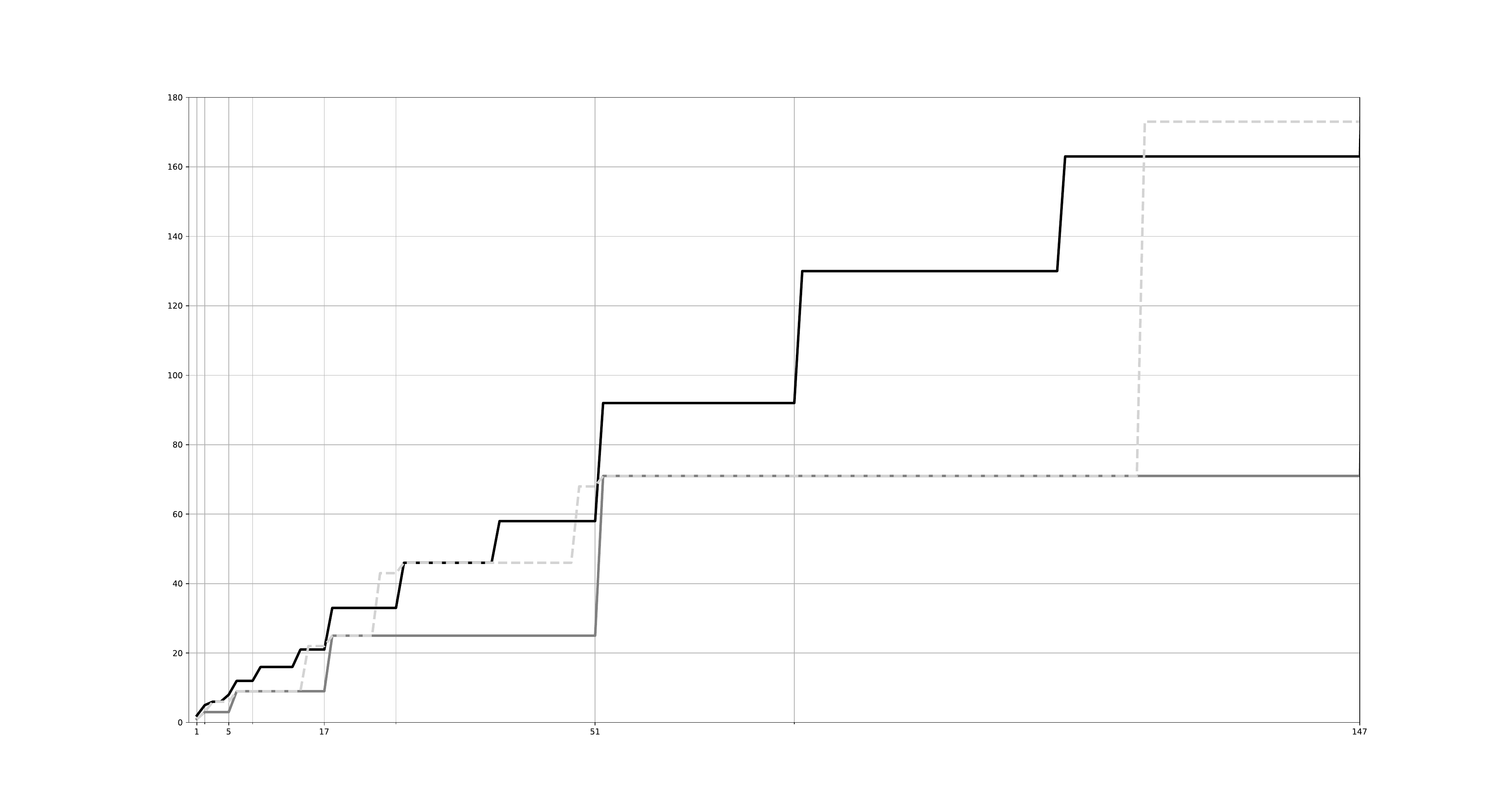}
    \caption{Plots of the initial nonrepetitive complexity of the episturmian words with directive word $(001122)^\omega$ having intercepts $0^\omega$ (dark gray), $(01)^\omega$ (dashed light gray), and $1^\omega$ (black). The major ticks $1$, $5$, $17$, $51$, $147$ on the $x$-axis are the endpoints of the intervals $\interval{k}$, and the minor ticks are the endpoints of the intervals $\interval{k,0}$.}\label{fig:plot}
  \end{figure}

  Let $n \in \interval{k,\ell}$ for $k$ and $\ell$ such that $k \geq 0$ and $0 \leq \ell < a_{k+1}$, and let $\theta_n$
  be the length of the central path of the Rauzy graph $\Gamma(n)$ of $\infw{c}_\Delta$ (the number of edges on the
  central path). The number $\irep{T^m(\infw{c}_\Delta)}{n}$ is determined by the cycle sequence taken in the Rauzy
  graph $\Gamma(n)$ when the word $\infw{c}_\Delta$ is read. We denote by $C_y$ the cycle of
  $\Gamma(\abs{u_{r_k+\ell+1}})$ containing the edge corresponding to the factor $u_{r_k+\ell+1} y$. We denote the
  length of $C_y$ by $\norm{C_y}$. Notice that the graph $\Gamma(n)$ has the same cycle lengths as the graph
  $\Gamma(\abs{u_{r_k+1+\ell}})$.

  Let us next show that $\norm{C_y} = \abs{\mu_{r_k+\ell}(y)}$. Say we start at the vertex $u_{r_k+\ell+1}$ of
  $\Gamma(\abs{u_{r_k+\ell+1}})$, take the cycle $C_y$, and return to the vertex $u_{r_k+\ell+1}$. This
  sequence of vertices corresponds to a factor $w$ of length $\abs{u_{r_k+\ell-1}} + \norm{C_y}$ such that $w$ contains
  exactly two occurrences of $u_{r_k+\ell+1}$, one as a prefix and one as a suffix. If follows from
  \cite[Eq.~2]{2002:episturmian_words_and_episturmian_morphisms} that
  $\smash[b]{u_{r_k+\ell+1} = L_{x_1}(v_1) x_1}$ where $v_1$ is the $(r_k+\ell)$th central word associated with the
  directive word $T(\Delta)$. Hence the word obtained from $w$ by removing its last letter decodes to a word $w_1$ such
  that $w_1$ has exactly two occurrences of $v_1$, one as a prefix and one as a suffix. Moreover, we deduce from the
  form of the morphism $L_{x_1}$ that the prefix $v_1$ of $w_1$ is followed by $y$. This procedure may be repeated
  $r_k+\ell$ times to obtain $w_{r_k+\ell} = y$. The procedure removes the suffix $u_{r_k+\ell+1}$ completely, so it
  must be that $\abs{\mu_{r_k+\ell}(y)} + \abs{u_{r_k+\ell+1}} = \abs{w}$, that is,
  $\norm{C_y} = \abs{\mu_{r_k+\ell}(y)}$.

  Next we partition the interval $\{0, 1, \ldots, q_{k+d-1} - 1\}$ into intervals
  $\lambda_i$, and we further divide these intervals into subintervals $\lambda_{i,j}$ according to the cycle sequence
  as described below. Our aim is to show that the initial nonrepetitive complexity has simple description on each
  $\lambda_{i,j}$. See \autoref{prp:inrc_shift}.

  Let
  \begin{equation*}
    \Delta' = T^{r_{k+1}}(\Delta) = x_{k+2}^{a_{k+2}} \dotsm \quad \text{and} \quad
    \Delta'' = T^{r_k+\ell}(\Delta) = x_{k+1}^{a_{k+1}-\ell} x_{k+2}^{a_{k+2}} \dotsm.
  \end{equation*}
  Preceding the first occurrence of $\smash[t]{x_{k+1}^{a_{k+1}-\ell+1}}$, the standard word $\infw{c}_{\Delta''}$ is
  formed of blocks $\smash[t]{x_{k+1}^{a_{k+1}-\ell} y}$ with $y \neq x_{k+1}$. We say that such a block is of type
  $y$. Notice that the block types are given by the letters of the standard word with directive word $\Delta'$. It is
  straightforward to verify that the number of blocks including the first block of type $x_{k+d}$ equals $K_d$, where
  \begin{equation*}
    K_d = \prod_{i=2}^{d-1} (a_{k+i} + 1).
  \end{equation*}
  In particular, $K_2 = 1$.

  Corresponding to the $i$th block, $1 \leq i \leq K_d$, we define an interval $\lambda_i$ as follows. If $i = 1$, then
  we let $L_i = 0$, and otherwise we let $L_i - 1$ to be the largest element of $\lambda_{i-1}$. We define
  \begin{equation*}
    \lambda_i = \{L_i, \ldots, L_i + \abs{\tau_{k+1}(y_i)} - 1\}
  \end{equation*}
  where $y_i$ is the type of the $i$th block. The number of elements of $\lambda_i$ is simply the length of the
  $\mu_{r_k+\ell}$-image of the block since
  $\smash[t]{\mu_{r_k+\ell}(x_{k+1}^{a_{k+1}-\ell} y_i) = \mu_{r_k+\ell} L_{x_{k+1}}^{a_{k+1}-\ell}(y_i) = \tau_{k+1}(y_i)}$.

  Next we subdivide the interval $\lambda_i$ into four adjacent intervals that respectively have sizes
  $(a_{k+1} - \ell - 1)\norm{C_{x_{k+1}}} + \theta_n + 1$, $\norm{C_{x_{k+1}}} - (\theta_n + 1)$, $\theta_n+1$, and
  $\norm{C_{y_i}} - (\theta_n + 1)$. More formally, we define
  \begin{align*}
    \lambda_{i,1} &= \{L_i, \ldots, L_i + (a_{k+1} - \ell - 1)\norm{C_{x_{k+1}}} + \theta_n\}, \\
    \lambda_{i,2} &= \{L_i + (a_{k+1} - \ell - 1)\norm{C_{x_{k+1}}} + \theta_n + 1, \ldots, L_i + (a_{k+1} - \ell)\norm{C_{x_{k+1}}} - 1\}, \\
    \lambda_{i,3} &= \{L_i + (a_{k+1} - \ell)\norm{C_{x_{k+1}}}, \ldots, L_i + (a_{k+1} - \ell)\norm{C_{x_{k+1}}} + \theta_n\}, \text{ and} \\
    \lambda_{i,4} &= \{L_i + (a_{k+1} - \ell)\norm{C_{x_{k+1}}} + \theta_n + 1, \ldots, L_i + (a_{k+1} - \ell)\norm{C_{x_{k+1}}} + \norm{C_{y_i}} - 1\}.
  \end{align*}

  Let us find out the size of the union of the intervals $\lambda_i$ for $i = 1, \ldots, K_d$. The intervals are
  clearly disjoint and adjacent, so the size equals $\abs{\tau_{k+1}(y)}$ summed over the block types $y$. Thus the
  size of the union equals the length of the $\tau_{k+1}$-image of the prefix of $\infw{c}_{\Delta'}$ having the first
  occurrence of $x_{k+d}$ as a suffix. By \eqref{eq:reg_s_1}, this prefix equals $v_{d-2}$ where $v_{d-2}$ is the
  $(d-2)$th standard word for the directive word $\Delta'$. Then the $\tau_{k+1}$-image of $v_{d-2}$ equals
  $s_{k+d-1}$. Indeed, by definition, we have
  $v_{d-2} = L_{x_{k+2}}^{a_{k+2}} \circ \dotsm \circ L_{x_{k+d-1}}^{a_{k+d-1}}(x_{k+d})$, so
  $\tau_{k+1}(v_{d-2}) = \tau_{k+d-1}(x_{k+d}) = s_{k+d-1}$. It follows that
  \begin{equation*}
    \bigcup_{i=1}^{K_d} \lambda_i = \{0, 1, \ldots, q_{k+d-1}-1\}.
  \end{equation*}

  \begin{example}\label{ex:sturmian_intervals}(Sturmian Case)
    When $\Delta$ is binary and $d = 2$, we have $K_d = 1$, so there is only one block. Now
    $\norm{C_{x_{k+1}}} = \abs{\tau_k(x_{k+1})} = q_k$. The type $y$ of the block is clearly $x_{k+2}$, so
    \begin{equation*}
      \norm{C_y} = \abs{\mu_{r_k+\ell}(x_{k+2})} = \abs{\mu_{r_k}(x_{k+1}^\ell x_{k+2})} = \ell q_k + \abs{\mu_{r_k}(x_k)} = \ell q_k + q_{k-1}.
    \end{equation*}
    By recalling that $L_1 = 0$, we find that the intervals $\lambda_{1,j}$ are as follows:
    \begin{align*}
      \lambda_{1,1} &= \{0, \ldots, (a_{k+1} - \ell - 1)q_k + \theta_n\}, \\
      \lambda_{1,2} &= \{(a_{k+1} - \ell - 1)q_k + \theta_n + 1, \ldots, (a_{k+1} - \ell)q_k - 1\}, \\
      \lambda_{1,3} &= \{(a_{k+1} - \ell)q_k, \ldots, (a_{k+1} - \ell)q_k + \theta_n\}, \text{ and} \\
      \lambda_{1,4} &= \{(a_{k+1} - \ell)q_k + \theta_n + 1, \ldots, q_{k+1} - 1\}.
    \end{align*}
    The union of the intervals equals $\{0, 1, \ldots, q_{k+1}-1\}$.
  \end{example}

  \begin{proposition}\label{prp:inrc_shift}
    Let $\Delta$ be a regular directive word. Let $n$ and $i$ be integers such that $n \in \interval{k,\ell}$ with
    $k \geq 0$ and $0 \leq \ell < a_{k+1}$ and $1 \leq i \leq K_d$. Suppose that the $i$th block has type $y_i$.
    \begin{enumerate}[(i)]
      \item If $m \in \lambda_{i,1}$, then $\irep{T^m(\infw{c}_\Delta)}{n} = q_k$.
      \item If $m \in \lambda_{i,2}$, then $\irep{T^m(\infw{c}_\Delta)}{n} = \abs{\tau_{k+1}(y_i)} + L_i - m$.
      \item If $m \in \lambda_{i,3}$, then $\irep{T^m(\infw{c}_\Delta)}{n} = \ell q_k + \abs{\tau_k(y_i)}$.
      \item If $m \in \lambda_{i,4}$, then $\irep{T^m(\infw{c}_\Delta)}{n} = q_k + L_{i+1} - m$.
    \end{enumerate}
  \end{proposition}
  \begin{proof}
    We are concerned with the cycles taken in the graph $\Gamma(n)$, which evolves to $\Gamma(u_{r_k+1+\ell})$. Recall
    that the cycle sequence taken in $\Gamma(u_{r_k+1+\ell})$ is determined by the letters of the standard word
    $\infw{c}_{\Delta''}$ with
    $\smash[t]{\Delta'' = T^{r_k+\ell}(\Delta) = x_{k+1}^{a_{k+1}-\ell} x_{k+2}^{a_{k+2}} \dotsm}$ and that the cycle
    lengths are given by the lengths of the $\mu_{r_k+\ell}$-images of these letters. Let $v$ be the prefix of
    $T^m(\infw{c}_\Delta)$ of length $n$.

    \textbf{Case A.} Suppose that $m \in \lambda_{i,1}$. By the discussion preceding \autoref{ex:sturmian_intervals},
    the number $L_i$ equals the length of the $\mu_{r_k+\ell}$-image of the first $i-1$ blocks. Then, from the remark
    at the beginning of this proof, we see that reading off $L_i$ letters from the beginning of $\infw{c}_\Delta$
    amounts to traveling complete cycles in $\Gamma(u_{r_k+1+\ell})$. Since all prefixes of $\infw{c}_\Delta$ are left
    special, we see that the prefix of $T^{L_i}(\infw{c}_\Delta)$ of length $n$ corresponds to the left special vertex
    of $\Gamma(n)$. Assume that $L_i \leq m < L_i + (a_{k+1} - \ell - 1)\norm{C_{x_{k+1}}}$. Since the cycle $C_1$ of
    $\Gamma(n)$ having length $\norm{C_{x_{k+1}}}$ is initially taken $a_{k+1}-\ell$ times (as the word
    $\smash[t]{x_{k+1}^{a_{k+1}-\ell}}$ is a prefix of $\infw{c}_\Delta''$), it follows that $v$ lies on $C_1$ and
    $C_1$ is traversed at least once more after the prefix $v$ of $T^m(\infw{c}_\Delta)$. It follows that
    \begin{equation*}
     \irep{T^m(\infw{c}_\Delta)}{n} = \norm{C_1} = \norm{C_{x_{k+1}}} = \abs{\mu_{r_k+\ell}(x_{k+1})} = q_k.
    \end{equation*}
    Assume then that
    $L_i + (a_{k+1} - \ell - 1)\norm{C_{x_{k+1}}} \leq m \leq L_i + (a_{k+1} - \ell - 1)\norm{C_{x_{k+1}}} + \theta_n$.
    Then $v$ lies on the central path of $\Gamma(n)$ during the $(a_{k+1}-\ell)$th traversal of $C_1$. Hence
    $\irep{T^m(\infw{c}_\Delta)}{n} = q_k$ in this case as well.

    \textbf{Case B.} Suppose that $m \in \lambda_{i,2}$. By the arguments in the latter case of the previous paragraph,
    we see that $v$ lies on the cycle $C_1$ during the $(a_{k+1}-\ell)$th traversal of $C_1$. Moreover, the vertex $v$
    is not on the central path of $\Gamma(n)$ and exactly $\norm{C_1} - (m - (L_i+(a_{k+1}-\ell-1)\norm{C_1}))$ edges
    need to be traversed to return to the left special vertex of $\Gamma(n)$. The $i$th block equals
    $\smash[t]{x_{k+1}^{a_{k+1}-\ell} y_i}$, so the cycle $C_1$ is followed by a cycle $C_y$ having length
    $\abs{\mu_{r_k+\ell}(y_i)}$. Therefore the initial nonrepetitive complexity is determined by the return to the left
    special vertex of $\Gamma(n)$ when traversing the cycle $C_y$, that is, we have
    \begin{align*}
      \irep{T^m(\infw{c}_\Delta)}{n} &= \norm{C_1} + \norm{C_y} - (m - (L_i+(a_{k+1}-\ell-1)\norm{C_1})) \\
                                     &= \abs{\mu_{r_k+\ell}(y_i)} - (m - (L_i+(a_{k+1}-\ell)q_k)) \\
                                     &= \abs{\tau_k(x_{k+1}^{\ell} y_i)} - (m - (L_i+(a_{k+1}-\ell)q_k)) \\
                                     &= \ell q_k + \abs{\tau_k(y_i)} - (m - (L_i+(a_{k+1}-\ell)q_k)) \\
                                     &= a_{k+1}q_k + \abs{\tau_k(y_i)} + L_i - m \\
                                     &= \abs{\tau_{k+1}(y_i)} + L_i - m.
    \end{align*}

    \textbf{Case C.} Suppose that $m \in \lambda_{i,3}$. Now $v$ lies on the central path of $\Gamma(n)$ and the
    next cycle to be traversed is $C_y$. Since $v$ is on the central path, the initial nonrepetitive complexity is
    simply $\norm{C_y}$. The claim follows from the preceding computations.

    \textbf{Case D.} Suppose that $m \in \lambda_{i,4}$. In this case $v$ lies on $C_y$ but not on the central path.
    From the form of the word $\infw{c}_{\Delta''}$, we deduce that the next cycle taken is $C_1$. Exactly
    $L_{i+1} - m$ edges need to be traversed to arrive at the left special factor of $\Gamma(n)$. Therefore
    \begin{equation*}
      \irep{T^m(\infw{c}_\Delta)}{n} = L_{i+1} - m + \norm{C_1} = q_k + L_{i+1} - m. \qedhere
    \end{equation*}
  \end{proof}

  In the Sturmian case, \autoref{prp:inrc_shift} simplifies as follows. This result was obtained in
  \cite[Prop.~5.3.0.1]{diss:caius_wojcik}.

  \begin{proposition}[Sturmian Case]
    Assume that $\Delta$ is a binary directive word. Let $n \in \interval{k,\ell}$ for $k \geq 0$ and
    $0 \leq \ell < a_{k+1}$.
    \begin{enumerate}[(i)]
      \item If $0 \leq m \leq (a_{k+1}-\ell-1)q_k + \theta_n$, then $\irep{T^m(\infw{c}_\Delta)}{n} = q_k$.
      \item If $(a_{k+1}-\ell-1)q_k + \theta_n < m < (a_{k+1} - \ell)q_k$, then $\irep{T^m(\infw{c}_\Delta)}{n} = q_{k+1} - m$.
      \item If $(a_{k+1} - \ell)q_k \leq m \leq (a_{k+1} - \ell)q_k + \theta_n$, then $\irep{T^m(\infw{c}_\Delta)}{n} = \ell q_k + q_{k-1}$.
      \item If $(a_{k+1} - \ell)q_k + \theta_n < m < q_{k+1}$, then $\irep{T^m(\infw{c}_\Delta)}{n} = q_{k+1} + q_k - m$.
    \end{enumerate}
  \end{proposition}
  \begin{proof}
    The claim follows by short computations using the information provided in \autoref{ex:sturmian_intervals}.
  \end{proof}

  The following proposition gives the nonrepetitive initial complexity of certain shifts of the regular standard
  episturmian words for the lengths in the interval $\interval{k}$. The statement is quite complicated. We advise the
  reader to read the proof and study the implications \eqref{eq:inrc_n_1}, \eqref{eq:inrc_n_2}, \eqref{eq:inrc_n_3},
  and \eqref{eq:inrc_n_4} rather than spending much time on the statement itself.

  \begin{proposition}\label{prp:inrc_n}
    Let $\Delta$ be a regular directive word and $\Delta' = T^{r_{k+1}}(\Delta)$ for $k \geq 0$. Suppose that
    $n \in \interval{k}$, and let $m$ be an integer such that $0 \leq m < q_{k+d-1}$ and
    $\rep[\Delta]{m} = c_1 \dotsm c_{k+d-1}$ (possibly with trailing zeros). Let $y_i$ be the $i$th letter of
    $\infw{c}_{\Delta'}$ when $i = \val[\Delta']{c_{k+2} \dotsm c_{k+d-1}} + 1$.
    \begin{enumerate}[(i)]
      \item If $0 < c_{k+1} < a_{k+1} - 1$ or $c_{k+1} = a_{k+1} - 1 > 0$ and $c_k = 0$, then we have the following
            implications:
            \begin{itemize}
              \item $\abs{u_{r_k}} < n \leq \abs{u_{r_{k+1}}} - \val[\Delta]{c_1 \dotsm c_{k+1}} \\
                     \Longrightarrow \, \irep{T^m(\infw{c}_\Delta)}{n} = q_k$,
              \item $\abs{u_{r_{k+1}}} - \val[\Delta]{c_1 \dotsm c_{k+1}} < n \leq \abs{u_{r_{k+1}}} - c_{k+1}q_k \\
                     \Longrightarrow \, \irep{T^m(\infw{c}_\Delta)}{n} = \abs{\tau_{k+1}(y_i)} - \val[\Delta]{c_1 \dotsm c_{k+1}}$,
              \item $\abs{u_{r_{k+1}}} - c_{k+1}q_k < n \leq \abs{u_{r_{k+1}}} + q_k - \val[\Delta]{c_1 \dotsm c_{k+1}} \\
                     \Longrightarrow \, \irep{T^m(\infw{c}_\Delta)}{n} = (a_{k+1}-c_{k+1})q_k + \abs{\tau_k(y_i)}$,
              \item $\abs{u_{r_{k+1}}} + q_k - \val[\Delta]{c_1 \dotsm c_{k+1}} < n \leq \abs{u_{r_{k+1}}} \\
                     \Longrightarrow \, \irep{T^m(\infw{c}_\Delta)}{n} = q_k + \abs{\tau_{k+1}(y_i)} - \val[\Delta]{c_1 \dotsm c_{k+1}}$.
            \end{itemize}
      \item If $c_{k+1} = a_{k+1}$, then we have the following implications:
            \begin{itemize}
              \item $\abs{u_{r_k}} < n \leq \abs{u_{r_{k+1}}} + q_k - \val[\Delta]{c_1 \dotsm c_k} \\
                     \Longrightarrow \, \irep{T^m(\infw{c}_\Delta)}{n} = \abs{\tau_k(y_i)}$,
              \item $\abs{u_{r_{k+1}}} + q_k - \val[\Delta]{c_1 \dotsm c_k} < n \leq \abs{u_{r_{k+1}}} \\
                     \Longrightarrow \, \irep{T^m(\infw{c}_\Delta)}{n} = q_k + \abs{\tau_k(y_i)} - \val[\Delta]{c_1 \dotsm c_k}$.
            \end{itemize}
      \item If $c_{k+1} = a_{k+1} - 1 > 0$ and $c_k \neq 0$, then we have the following implications:
            \begin{itemize}
              \item $\abs{u_{r_k}} < n \leq \abs{u_{r_k+1}} \\
                     \Longrightarrow \, \irep{T^m(\infw{c}_\Delta)}{n} = q_k + \abs{\tau_k(y_i)} - \val[\Delta]{c_1 \dotsm c_k}$,
              \item $\abs{u_{r_k+1}} < n \leq \abs{u_{r_k+1}} + q_k - \val[\Delta]{c_1 \dotsm c_k} \\
                     \Longrightarrow \, \irep{T^m(\infw{c}_\Delta)}{n} = q_k + \abs{\tau_k(y_i)}$,
              \item $\abs{u_{r_k+1}} + q_k - \val[\Delta]{c_1 \dotsm c_k} < n \leq \abs{u_{r_{k+1}}} \\
                     \Longrightarrow \, \irep{T^m(\infw{c}_\Delta)}{n} = 2q_k + \abs{\tau_k(y_i)} - \val[\Delta]{c_1 \dotsm c_k}$.
            \end{itemize}
      \item If $c_{k+1} = 0$ and $c_k = 0$, then we have the following implications:
            \begin{itemize}
              \item $\abs{u_{r_k}} < n \leq \abs{u_{r_{k+1}}} - \val[\Delta]{c_1 \dotsm c_{k-1}} \\
                     \Longrightarrow \, \irep{T^m(\infw{c}_\Delta)}{n} = q_k$,
              \item $\abs{u_{r_{k+1}}} - \val[\Delta]{c_1 \dotsm c_{k-1}} < n \leq \abs{u_{r_{k+1}}} \\
                     \Longrightarrow \, \irep{T^m(\infw{c}_\Delta)}{n} = \abs{\tau_{k+1}(y_i)} - \val[\Delta]{c_1 \dotsm c_{k-1}}$.
            \end{itemize}
      \item If $c_{k+1} = a_{k+1} - 1 = 0$ and $c_k \neq 0$, then we have the following implications:
            \begin{itemize}
              \item $\abs{u_{r_k}} < n \leq \abs{u_{r_{k+1}}} \\
                     \Longrightarrow \, \irep{T^m(\infw{c}_\Delta)}{n} = \abs{\tau_{k+1}(y_i)} - \val[\Delta]{c_1 \dotsm c_k}$.
            \end{itemize}
      \item If $a_{k+1} - 1 > c_{k+1} = 0$ and $c_k \neq 0$, then we have the following implications:
            \begin{itemize}
              \item $\abs{u_{r_k}} < n \leq \abs{u_{r_{k+1}}} - \val[\Delta]{c_1 \dotsm c_k} \\
                     \Longrightarrow \, \irep{T^m(\infw{c}_\Delta)}{n} = q_k$.
              \item $\abs{u_{r_{k+1}}} - \val[\Delta]{c_1 \dotsm c_k} < n \leq \abs{u_{r_{k+1}}} \\
                     \Longrightarrow \, \irep{T^m(\infw{c}_\Delta)}{n} = \abs{\tau_{k+1}(y_i)} - \val[\Delta]{c_1 \dotsm c_k}$.
            \end{itemize}
    \end{enumerate}
  \end{proposition}
  \begin{proof}
    Since $0 \leq m < q_{k+d-1}$, we see that $m \in \lambda_i$ for some $i$. The discussion preceding
    \autoref{ex:sturmian_intervals} tells us that the left endpoint $L_i$ of $\lambda_i$ equals the $\tau_{k+1}$-image
    of the prefix of $\infw{c}_{\Delta'}$ of length $i-1$. Let $v_1 = \rep[\Delta']{i-1}$ and $v_2 = \rep[\Delta']{i}$.
    It follows that $0^{k+1} v_1 \leq_{\mathrm{lex}} \rep[\Delta]{m} <_{\mathrm{lex}} 0^{k+1} v_2$ (here
    $<_{\mathrm{lex}}$ is the lexicographic order on $\N$). Since $v_1$ and $v_2$ are the representations of two
    consecutive integers, we see that $\rep[\Delta]{m}$ must end with $v_1$. In other words, we have showed that
    $i = \val[\Delta']{c_{k+2} \dotsm c_{k+d-1}} + 1$. Thus the type of the $i$th block is $y_i$ where $y_i$ is the
    $i$th letter of $\infw{c}_{\Delta'}$.

    Assume that $n \in \interval{k,l}$ for some $\ell$ such that $0 \leq \ell < a_{k+1}$, and write
    \begin{equation}\label{eq:n_cp}
      n = \abs{u_{r_k+\ell+1}} - \theta_n = \abs{u_{r_k+1}} + \ell q_k - \theta_n.
    \end{equation}
    Recall that $\theta_n$ is the length of the central path of $\Gamma(n)$ and that the intervals $\interval{k,\ell}$
    have size $q_k$ except in the case $\ell = 0$ when the size is $q_{k-1}$.

    \textbf{Case A.} Let us first consider the case where $m \in \lambda_{i,1}$. Now
    $L_i = \val[\Delta]{0^{k+1} c_{k+2} \dotsm c_{k+d-1}}$, so $m \in \lambda_{i,1}$ if and only if
    \begin{equation*}
      0 \leq \val[\Delta]{c_1 \dotsm c_{k+1}} \leq (a_{k+1} - \ell - 1)q_k + \theta_n.
    \end{equation*}
    Substituting $\theta_n$ from \eqref{eq:n_cp} to the right inequality yields
    \begin{align*}
      n &\leq (a_{k+1} - \ell - 1)q_k + \abs{u_{r_k+1}} + \ell q_k - \val[\Delta]{c_1 \dotsm c_{k+1}} \\
        &= (a_{k+1} -  1)q_k + \abs{u_{r_k+1}} - \val[\Delta]{c_1 \dotsm c_{k+1}} \\
        &= \abs{u_{r_{k+1}}} - \val[\Delta]{c_1 \dotsm c_{k+1}}
    \end{align*}
    The left inequality $\val[\Delta]{c_1 \dotsm c_{k+1}} \geq 0$ is trivially true, so we conclude using
    \autoref{prp:inrc_shift} that
    \begin{equation}\label{eq:inrc_n_1}
      \abs{u_{r_k}} < n \leq \abs{u_{r_{k+1}}} - \val[\Delta]{c_1 \dotsm c_{k+1}} \, \Longrightarrow \, \irep{T^m(\infw{c}_\Delta)}{n} = q_k.
    \end{equation}
    Now $\abs{u_{r_{k+1}}} = \abs{u_{r_k}} + q_{k-1} + (a_{k+1}-1)q_k$, so the antecedent is false only if
    $\val[\Delta]{c_1 \dotsm c_{k+1}} \geq (a_{k+1}-1)q_k + q_{k-1}$. Using
    \autoref{lem:ostrowski_ub}, we see that this happens exactly when $c_{k+1} = a_{k+1}$ or $c_{k+1} = a_{k+1} - 1$
    and $c_k \neq 0$.

    \textbf{Case B.} Assume that $m \in \lambda_{i,2}$. Like above, this holds if and only if
    \begin{equation*}
      (a_{k+1} - \ell - 1)q_k + \theta_n < \val[\Delta]{c_1 \dotsm c_{k+1}} < (a_{k+1} - \ell)q_k.
    \end{equation*}
    The only option is that $c_{k+1} = a_{k+1}-\ell-1$, which cannot happen if $a_{k+1} = c_{k+1}$. Substituting
    $\theta_n$ from \eqref{eq:n_cp} to the left inequality implies that
    \begin{equation*}
      n > (a_{k+1}-1)q_k + \abs{u_{r_k+1}} - \val[\Delta]{c_1 \dotsm c_{k+1}} = \abs{u_{r_{k+1}}} - \val[\Delta]{c_1 \dotsm c_{k+1}}.
    \end{equation*}
    The right inequality $\val[\Delta]{c_1 \dotsm c_{k+1}} < (a_{k+1} - \ell)q_k$ is trivially true. Since
    $\abs{u_{r_k+1}} + \ell q_k = \abs{u_{r_k+1}} + (a_{k+1} - c_{k+1} - 1)q_k = \abs{u_{r_{k+1}}} - c_{k+1}q_k$, we
    conclude from \autoref{prp:inrc_shift} that
    \begin{align}\label{eq:inrc_n_2}
      &\abs{u_{r_{k+1}}} - \val[\Delta]{c_1 \dotsm c_{k+1}} < n \leq \abs{u_{r_{k+1}}} - c_{k+1}q_k \nonumber \\
      &\Longrightarrow \, \irep{T^m(\infw{c}_\Delta)}{n} = \abs{\tau_{k+1}(y_i)} - \val[\Delta]{c_1 \dotsm c_{k+1}}.
    \end{align}
    Notice indeed that here
    $L_i - m = \val[\Delta]{0^{k+1} c_{k+2} \dotsm c_{k+d-1}} - m = -\val[\Delta]{c_1 \dotsm c_{k+1}}$.

    \textbf{Case C.} Suppose that $m \in \lambda_{i,3}$. This holds if and only if
    \begin{equation*}
      (a_{k+1} - \ell)q_k \leq \val[\Delta]{c_1 \dotsm c_{k+1}} \leq (a_{k+1} - \ell)q_k + \theta_n.
    \end{equation*}
    Since $\theta_n < q_k$, it must be that $c_{k+1} = a_{k+1} - \ell$, and this cannot happen if
    $c_{k+1} = 0$ as $0 \leq \ell < a_{k+1}$. The left inequality is trivial. Utilizing again \eqref{eq:n_cp}, the
    right inequality transforms into
    \begin{equation*}
      n \leq a_{k+1}q_k + \abs{u_{r_k+1}} - \val[\Delta]{c_1 \dotsm c_{k+1}} = \abs{u_{r_{k+1}}} + q_k - \val[\Delta]{c_1 \dotsm c_{k+1}}.
    \end{equation*}
    As $\abs{u_{r_k+\ell}} = \abs{u_{r_k+1}} + (\ell-1)q_k = \abs{u_{r_k+1}} + (a_{k+1} - c_{k+1} - 1)q_k = \abs{u_{r_{k+1}}} - c_{k+1}q_k$,
    we have the following:
    \begin{align}\label{eq:inrc_n_3}
      &\abs{u_{r_k+1}} - c_{k+1}q_k < n \leq \abs{u_{r_{k+1}}} + q_k - \val[\Delta]{c_1 \dotsm c_{k+1}} \nonumber \\
      &\Longrightarrow \, \irep{T^m(\infw{c}_\Delta)}{n} = (a_{k+1}-c_{k+1})q_k + \abs{\tau_k(y_i)}.
    \end{align}

    \textbf{Case D.} Assume finally that $m \in \lambda_{i,4}$. This is true only if
    \begin{equation*}
      (a_{k+1} - \ell)q_k + \theta_n < \val[\Delta]{c_1 \dotsm c_{k+1}} < (a_{k+1} - \ell)q_k + \abs{\tau_{k+1}(y_i)}.
    \end{equation*}
    From the left inequality, we obtain the following inequality:
    \begin{equation*}
      n > a_{k+1}q_k + \abs{u_{r_k+1}} - \val[\Delta]{c_1 \dotsm c_{k+1}} = \abs{u_{r_{k+1}}} + q_k - \val[\Delta]{c_1 \dotsm c_{k+1}}.
    \end{equation*}
    Since $n \in \interval{k,\ell}$, we see that this is possible only when $q_k < \val[\Delta]{c_1 \dotsm c_{k+1}}$,
    that is, when $c_{k+1} \neq 0$. The right inequality holds trivially, so we have
    \begin{align}\label{eq:inrc_n_4}
      &\abs{u_{r_{k+1}}} + q_k - \val[\Delta]{c_1 \dotsm c_{k+1}} < n \leq \abs{u_{r_{k+1}}} \nonumber \\
      &\Longrightarrow \, \irep{T^m(\infw{c}_\Delta)}{n} = q_k + \abs{\tau_{k+1}(y_i)} - \val[\Delta]{c_1 \dotsm c_{k+1}}.
    \end{align}
    Here the facts $L_{i+1} = L_i + \abs{\tau_{k+1}(y_i)}$ and $L_i - m = -\val[\Delta]{c_1 \dotsm c_{k+1}}$ were used.

    Let us then put the above results together. If $c_{k+1}$ satisfies $0 < c_{k+1} < a_{k+1}-1$, then the antecedents
    of \eqref{eq:inrc_n_1}, \eqref{eq:inrc_n_2}, \eqref{eq:inrc_n_3}, and \eqref{eq:inrc_n_4} are all satisfied.
    Clearly the interval $\{\abs{u_{r_k}} + 1, \ldots, \abs{u_{r_{k+1}}}\}$ is partitioned by these four cases
    and the initial nonrepetitive complexity is determined on each partition. Exactly the same happens if
    $c_{k+1} = a_{k+1} - 1$ and $c_k = 0$. This gives (i). Suppose then that $c_{k+1} = a_{k+1}$. Then, as we saw
    above, the antecedents of \eqref{eq:inrc_n_1} and \eqref{eq:inrc_n_2} are not satisfied, so these cases are
    omitted. The left inequality is trivial in the Case C and $n \in \interval{k,0}$, so we may now deduce that 
    $\irep{T^m(\infw{c}_\Delta)}{n} = \abs{\tau_k(y_i)}$ provided that
    $\abs{u_{r_k}} < n \leq \abs{u_{r_{k+1}}} + q_k - \val[\Delta]{c_1 \dotsm c_{k+1}}$. The Case D directly applies,
    and we have (ii). The remaining cases are similar.
  \end{proof}

  The Sturmian case was determined in \cite[Cor.~5.3.0.2]{diss:caius_wojcik}.

  \begin{proposition}[Sturmian Case]
    Let $\Delta$ be a binary directive word. Suppose that $n \in \interval{k}$ for $k \geq 0$, and let $m$ be an
    integer such that $0 \leq m < q_{k+1}$ and $\rep{m} = c_1 \dotsm c_{k+1}$ (possibly with trailing zeros).
    \begin{enumerate}[(i)]
      \item If $0 < c_{k+1} < a_{k+1} - 1$ or $c_{k+1} = a_{k+1} - 1 > 0$ and $c_k = 0$, then we have the following
            implications:
            \begin{itemize}
              \item $q_k - 2 < n \leq q_{k+1} - 2 - \val{c_1 \dotsm c_{k+1}} \\
                     \Longrightarrow \, \irep{T^m(\infw{c}_\Delta)}{n} = q_k$,
              \item $q_{k+1} - 2 - \val{c_1 \dotsm c_{k+1}} < n \leq q_{k+1} - 2 - c_{k+1}q_k \\
                     \Longrightarrow \, \irep{T^m(\infw{c}_\Delta)}{n} = q_{k+1} - \val{c_1 \dotsm c_{k+1}}$,
              \item $q_{k+1} - 2 - c_{k+1}q_k < n \leq q_{k+1} + q_k - 2 - \val{c_1 \dotsm c_{k+1}} \\
                     \Longrightarrow \, \irep{T^m(\infw{c}_\Delta)}{n} = q_{k+1} - c_{k+1}q_k$,
              \item $q_{k+1} + q_k - 2 - \val{c_1 \dotsm c_{k+1}} < n \leq q_{k+1} - 2 \\
                     \Longrightarrow \, \irep{T^m(\infw{c}_\Delta)}{n} = q_{k+1} + q_k - \val{c_1 \dotsm c_{k+1}}$.
            \end{itemize}
      \item If $c_{k+1} = a_{k+1}$, then we have the following implications:
            \begin{itemize}
              \item $q_k - 2 < n \leq q_{k+1} - 2 - \val{c_1 \dotsm c_k} \\
                     \Longrightarrow \, \irep{T^m(\infw{c}_\Delta)}{n} = q_{k-1}$,
              \item $q_{k+1} - 2 - \val{c_1 \dotsm c_k} < n \leq q_{k+1} - 2 \\
                     \Longrightarrow \, \irep{T^m(\infw{c}_\Delta)}{n} = q_k + q_{k-1} - \val{c_1 \dotsm c_k}$.
            \end{itemize}
      \item If $c_{k+1} = a_{k+1} - 1 > 0$ and $c_k \neq 0$, then we have the following implications:
            \begin{itemize}
              \item $q_k - 2 < n \leq q_k + q_{k-1} - 2 \\
                     \Longrightarrow \, \irep{T^m(\infw{c}_\Delta)}{n} = q_k + q_{k-1} - \val{c_1 \dotsm c_k}$,
              \item $q_k + q_{k-1} - 2 < n \leq 2q_k + q_{k-1} - 2 - \val{c_1 \dotsm c_k} \\
                     \Longrightarrow \, \irep{T^m(\infw{c}_\Delta)}{n} = q_k + q_{k-1}$,
              \item $2q_k + q_{k-1} - 2 - \val{c_1 \dotsm c_k} < n \leq q_{k+1} - 2 \\
                     \Longrightarrow \, \irep{T^m(\infw{c}_\Delta)}{n} = 2q_k + q_{k-1} - \val{c_1 \dotsm c_k}$.
            \end{itemize}
      \item If $c_{k+1} = 0$ and $c_k = 0$, then we have the following implications:
            \begin{itemize}
              \item $q_k - 2 < n \leq q_{k+1} - 2 - \val{c_1 \dotsm c_{k-1}} \\
                     \Longrightarrow \, \irep{T^m(\infw{c}_\Delta)}{n} = q_k$,
              \item $q_{k+1} - 2 - \val{c_1 \dotsm c_{k-1}} < n \leq q_{k+1} - 2 \\
                     \Longrightarrow \, \irep{T^m(\infw{c}_\Delta)}{n} = q_{k+1} - \val{c_1 \dotsm c_{k-1}}$.
            \end{itemize}
      \item If $c_{k+1} = a_{k+1} - 1 = 0$ and $c_k \neq 0$, then we have the following implications:
            \begin{itemize}
              \item $q_k - 2 < n \leq q_{k+1} - 2 \\
                     \Longrightarrow \, \irep{T^m(\infw{c}_\Delta)}{n} = q_{k+1} - \val{c_1 \dotsm c_k}$.
            \end{itemize}
      \item If $a_{k+1} - 1 > c_{k+1} = 0$ and $c_k \neq 0$, then we have the following implications:
            \begin{itemize}
              \item $q_k - 2 < n \leq q_{k+1} - 2 - \val{c_1 \dotsm c_k} \\
                     \Longrightarrow \, \irep{T^m(\infw{c}_\Delta)}{n} = q_k$.
              \item $q_{k+1} - 2 - \val{c_1 \dotsm c_k} < n \leq q_{k+1} - 2 \\
                     \Longrightarrow \, \irep{T^m(\infw{c}_\Delta)}{n} = q_{k+1} - \val{c_1 \dotsm c_k}$.
            \end{itemize}
    \end{enumerate}
  \end{proposition}
  \begin{proof}
    An easy induction argument shows that $\abs{u_{r_k}} = q_k - 2$ for all $k \geq 1$ in the Sturmian case. Utilize
    again the computations of \autoref{ex:sturmian_intervals}.
  \end{proof}

  Keeping in mind \autoref{thm:ostrowski_val_limit}, we now determine which shifts of $\infw{c}_\Delta$ need to be
  considered in order to determine the initial nonrepetitive complexity. Notice that if precise information is not
  required, the first case of the following proposition can be omitted since the longest common prefix of $\infw{t}$
  and $T^{\val{c_1 \dotsm c_{k+d}}}(\infw{c}_\Delta)$ is at least as long as that of $\infw{t}$ and
  $T^{\val{c_1 \dotsm c_{k+d-1}}}(\infw{c}_\Delta)$ according to \autoref{lem:common_prefix_length_increasing}. Notice
  also that when the case (i) of the proposition applies, we can compute the initial nonrepetitive complexity using
  \autoref{prp:inrc_n}. This result was proved for Sturmian words in \cite[Prop.~5.3.0.3]{diss:caius_wojcik}.

  \begin{proposition}\label{prp:sufficient_shift}
    Let $\infw{t}$ be a regular episturmian word with directive word as in \eqref{eq:dw_multiplicative} and intercept
    $c_1 c_2 \dotsm$. Let $n \in \interval{k}$ for $k \geq 0$.
    \begin{enumerate}[(i)]
      \item If $c_i < a_i$ for some $i$ such that $k+2 \leq i \leq k+d$, then
            $\irep{\infw{t}}{n} = \irep{T^{\val{c_1 \dotsm c_{k+d-1}}}(\infw{c}_\Delta)}{n}$.
      \item If $c_i = a_i$ for all $i$ such that $k+2 \leq i \leq k+d$, then
            $\irep{\infw{t}}{n} = \irep{T^{\val{c_1 \dotsm c_{k+d}}}(\infw{c}_\Delta)}{n}$.
    \end{enumerate}
  \end{proposition}

  For the proof, we need some auxiliary lemmas. They include more information than we need, but the complete statements
  could be useful in some other contexts.

  \begin{lemma}\label{lem:regular_help}
    Let $\Delta = x_1^{a_1} x_2^{a_2} \dotsm$ be a regular directive word with period $d$. Let $k$ and $i$ be such that
    $1 \leq k \leq d$ and $0 \leq i < k$. Then
    $\tau_k(x_{k-i}) = s_{k-1}^{\s{a}} \dotsm s_{k-i}^{\s{a}} s_{k-(i+1)}$.
  \end{lemma}
  \begin{proof}
    Let $k = 1$ and $i = 0$. Now $\tau_k(x_{k-i}) = \tau_1(x_1) = x_1 = s_0$, so the base case is established. Let then
    $k > 1$ and $i$ be such that $1 \leq i < k$. If $i > 0$, then $x_k \neq x_{k-i}$ because $\Delta$ is regular, so
    $\tau_k(x_{k-i}) = s_{k-1}^{\s{a}} \tau_{k-1}(x_{k-i})$, and the claim follows from the induction hypothesis. If
    $i = 0$, then $\tau_{k-1}(x_{k-i}) = s_{k-1}$, and the claim follows.
  \end{proof}

  \begin{lemma}\label{lem:regular_tau_images}
    Let $\Delta = x_1^{a_1} x_2^{a_2} \dotsm$ be a regular directive word with period $d$. For $k$ and $\ell$ such that
    $0 \leq k < d$ and $1 \leq \ell \leq a_{k+1}$, we have
    \begin{enumerate}[(i)]
      \item $\tau_k L_{x_{k+1}}^\ell(x_{k+1}) = s_k$,
      \item $\tau_k L_{x_{k+1}}^\ell(x_i) = s_k^\ell s_{k-1}^{\s{a}} \dotsm s_i^{\s{a}} s_{i-1}$ when
            $1 \leq i \leq k$, and
      \item $\tau_k L_{x_{k+1}}^\ell(x_i) = s_k^\ell s_{k-1}^{\s{a}} \dotsm s_0^{\s{a}} x_i$ when $i > k+1$.
    \end{enumerate}
    Let $k$, $\ell$, and $i$ be such that $k \geq d$, $1 \leq \ell \leq a_{k+1}$, and $1 \leq i \leq d$. Then
    \begin{enumerate}[(i)]
      \item[(iv)] $\tau_k L_{x_{k+1}}^\ell(x_i) = s_k$ if $i \equiv k + 1 \pmod{d}$ and
      \item[(v)] $\tau_k L_{x_{k+1}}^\ell(x_i) = s_k^{\ell} s_{k-1}^{\s{a}} \dotsm s_{k-j}^{\s{a}} s_{k-j-1}$ when
                 $i \not\equiv k + 1 \pmod{d}$ and $j$ is the smallest integer such that $k-j \equiv i \pmod{d}$.
    \end{enumerate}
  \end{lemma}
  \begin{proof}
    Assume that $k$ and $\ell$ are such that $0 \leq k < d$ and $1 \leq \ell \leq a_{k+1}$. Notice that since $\Delta$
    is regular, we have $x_{k+1} \neq x_i$ for $i = 1, \ldots, k$ and $i = k+2, \ldots, k+d$. Clearly
    $\tau_k L_{x_{k+1}}^\ell(x_{k+1}) = \tau_k(x_{k+1}) = s_k$. It is straightforward to show that
    $\tau_k(x_i) = s_{k-1}^{\s{a}} \dotsm s_0^{\s{a}} x_i$ for $i$ such that $k+1 < i < d$, so we obtain
    $\tau_k L_{x_{k+1}}^\ell(x_i) = \tau_k(x_{k+1}^\ell x_i) = s_k^\ell \tau_k(x_i) = s_k^\ell s_{k-1}^{\s{a}} \dotsm s_0^{\s{a}} x_i$
    for these $i$. Let then $i$ be such that $0 \leq i < k$. Since $x_{k-i} \neq x_{k+1}$, we infer from
    \autoref{lem:regular_help} that
    \begin{equation*}
      \tau_k L_{x_{k+1}}^\ell(x_{k-i}) = s_k^\ell \tau_k(x_{k-i}) = s_k^\ell s_{k-1}^{\s{a}} \dotsm s_{k-i}^{\s{a}} s_{k-(i+1)}.
    \end{equation*}
    This proves the first part of the claim.

    Suppose then that $k \geq d$, $1 \leq \ell \leq a_{k+1}$, and $1 \leq i \leq d$. Again we have
    $\tau_k L_{x_{k+1}}^\ell(x_{k+1}) = s_k$ which gives (iv). Assume that $i \not\equiv k+1 \pmod{d}$. Let $j$ be the
    smallest number such that $k-j \equiv i \pmod{d}$. Applying computations as above, we obtain that
    \begin{equation*}
      \tau_k L_{x_{k+1}}^\ell(x_i) = s_k^{\ell} s_{k-1}^{\s{a}} \tau_{k-1}(x_i) = \ldots = s_k^{\ell} s_{k-1}^{\s{a}} \dotsm s_{k-j}^{\s{a}} \tau_{k-j}(x_i) = s_k^{\ell} s_{k-1}^{\s{a}} \dotsm s_{k-j}^{\s{a}} s_{k-j-1}.
    \end{equation*}
    This proves (v).
  \end{proof}

  The point of \autoref{lem:regular_tau_images} is that the words $\tau_k L_{x_{k+1}}^\ell(x_i)$ are ordered by length
  in a predictable pattern. Let us consider the words $\smash[t]{\tau_k L_{x_{k+1}}^\ell(x_i)}$ ordered by the index
  $i$ in the natural order of $\{1, 2, \ldots, d\}$. Then the first part of \autoref{lem:regular_tau_images} states
  that the lengths $\smash[t]{\abs{\tau_k L_{x_{k+1}}^\ell(x_i)}}$ strictly decrease when $i$ increases from $1$ to
  $k + 1$. The remaining words are of equal length that is strictly greater than the preceding values. In other words,
  the shortest word $\tau_k L_{x_{k+1}}^\ell(x_{k+1})$ is pushed to the right in a cyclical fashion. The second part of
  the lemma states that this cyclical pattern continues modulo $d$.

  \begin{proof}[Proof of \autoref{prp:sufficient_shift}]
    Suppose that $n \in \interval{k,\ell}$ for some $\ell$ such that $0 \leq \ell < a_{k+1}$. Let
    $m = \val{c_1 \dotsm c_{k+d-1}}$. Now $m$ belongs to the $i$th block where
    $i = \val[\Delta']{c_{k+2} \dotsm c_{k+d-1}} + 1$ and $\Delta' = T^{r_{k+1}}(\Delta)$ (see the first paragraph of
    the proof of \autoref{prp:inrc_n}). Say the block has type $y_i$. Write
    $n = \abs{u_{r_k+1}} + \ell q_k - \theta_n$.

    We assume first that there exists a largest $i$ such that $k+2 \leq i \leq k+d$ and $c_i < a_i$. In order to show
    that $\irep{\infw{t}}{n} = \irep{T^m(\infw{c}_\Delta)}{n}$, it suffices to demonstrate that the longest common
    prefix of these words has length at least $\irep{T^m(\infw{c}_\Delta)}{n} + n$. Below we do this depending on which
    interval $\lambda_{i,j}$ the number $m$ belongs to. If $\infw{t} = T^m(\infw{c}_\Delta)$, there is nothing to
    prove, so we assume that there exists a least positive integer $j$ such that $c_{k+d-1+j} \neq 0$. By
    \autoref{lem:common_prefix_length}, the longest common prefix of $T^m(\infw{c}_\Delta)$ and
    $\smash[t]{T^{\val{c_1 \dotsm c_{k+d-1+j}}}(\infw{c}_\Delta)}$ has length $P$ where
    \begin{equation*}
      P = \abs{s_{k+d-1+j-1}^{\s{a}-\s{c}} s_{k+d-1+j-2}^{\s{a}} \dotsm s_0^{\s{a}}} - \val{c_1 \dotsm c_{k+d-1}}.
    \end{equation*}
    It follows from \autoref{lem:common_prefix_length_increasing} that $\infw{t}$ and
    $T^{\val{c_1 \dotsm c_{k+d-1+j}}}(\infw{c}_\Delta)$ also have a common prefix of length $P$. It thus suffices to
    show that $P \geq \irep{T^m(\infw{c}_\Delta)}{n} + n$ to establish (i).

    \textbf{Case A.} Assume that $m \in \lambda_{i,1}$. As in the proof of \autoref{prp:inrc_n}, we see that this means
    that $0 \leq \val{c_1 \dotsm c_{k+1}} \leq (a_{k+1} - \ell - 1)q_k + \theta_n$, so
    $-\theta_n \leq (a_{k+1} - \ell - 1)q_k - \val{c_1 \dotsm c_{k+1}}$. Hence we obtain from \autoref{prp:inrc_shift}
    that
    \begin{align*}
      \irep{T^m(\infw{c}_\Delta)}{n} + n &= q_k + n \\
                                         &= q_k + \abs{u_{r_{k+1}}} + \ell q_k - \theta_n \\
                                         &\leq \abs{u_{r_k+1}} + a_{k+1} q_k - \val{c_1 \dotsm c_{k+1}} \\
                                         &= \abs{u_{r_{k+1}+1}} - \val{c_1 \dotsm c_{k+1}}.
    \end{align*}
    By applying \autoref{lem:common_prefix_length_increasing} and the preceding inequality, we obtain
    \begin{align*}
      P &\geq \abs{s_{k+1}^{\s{a}-\s{c}} s_k^{\s{a}} \dotsm s_0^{\s{a}}} - \val{c_1 \dotsm c_{k+1}} \\
        &\geq \abs{s_k^{\s{a}} \dotsm s_0^{\s{a}}} - \val{c_1 \dotsm c_{k+1}} \\
        &= \abs{u_{r_{k+1}+1}} - \val{c_1 \dotsm c_{k+1}} \\
        &\geq \irep{T^m(\infw{c}_\Delta)}{n} + n,
    \end{align*}
    so we conclude that $\irep{\infw{t}}{n} = \irep{T^m(\infw{c}_\Delta)}{n}$.

    \textbf{Case B.} Assume that $m \in \lambda_{i,2}$. As in the proof of \autoref{prp:inrc_n}, we have
    \begin{equation*}
      (a_{k+1} - \ell - 1)q_k + \theta_n < \val{c_1 \dotsm c_{k+1}} < (a_{k+1} - \ell)q_k
    \end{equation*}
    and $c_{k+1} = a_{k+1}-\ell-1$. Thus \autoref{prp:inrc_shift} gives
    \begin{align*}
      \irep{T^m(\infw{c}_\Delta)}{n} + n &= \abs{\tau_{k+1}(y_i)} - \val{c_1 \dotsm c_{k+1}} + \abs{u_{r_k+1}} + \ell q_k + \theta_n \\
                                         &= \abs{\tau_{k+1}(y_i)} - \val{c_1 \dotsm c_{k+1}} + \abs{u_{r_k+1}} + (a_{k+1} - c_{k+1} - 1)q_k + \theta_n \\
                                         &= \abs{\tau_{k+1}(y_i)} - \val{c_1 \dotsm c_{k+1}} + \abs{u_{r_{k+1}+1}} - (c_{k+1}+1)q_k + \theta_n.
    \end{align*}
    Recall that there exists a largest $i$ such that $k+2 \leq i \leq k+d$ and $c_i < a_i$. An application of
    \autoref{lem:common_prefix_length_increasing} yields
    \begin{align*}
      P &\geq \abs{s_{i-1}^{\s{a} - \s{c}} s_{i-2}^{\s{a}} \dotsm s_0^{\s{a}}} - \val{c_1 \dotsm c_{i-1}} \\
        &\geq q_{i-1} + \abs{s_{i-2}^{\s{a}} \dotsm s_0^{\s{a}}} - \val{c_1 \dotsm c_{i-1}} \\
        &\geq q_{k+1} + \abs{s_k^{\s{a}} \dotsm s_0^{\s{a}}} - \val{c_1 \dotsm c_{k+1}} \\
        &= q_{k+1} + \abs{u_{r_{k+1}+1}} - \val{c_1 \dotsm c_{k+1}}.
    \end{align*}
    It thus suffices to prove that $\abs{\tau_{k+1}(y_i)} \leq q_{k+1}$ in order to conclude that
    $\irep{\infw{t}}{n} = \irep{T^m(\infw{c}_\Delta)}{n}$. Since $y_i \neq x_{k+1}$, the $j$ in (v) of
    \autoref{lem:regular_tau_images} is at most $d-2$ which implies that $\tau_{k+1}(y_i)$ is a prefix of
    $s_k^{\s{a}} \dotsm s_{k-(d-2)}^{\s{a}} s_{k-(d-1)}$. Since
    $s_k^{\s{a}} \dotsm s_{k+1-d+1}^{\s{a}} s_{k+1-d} = s_{k+1}$ by \eqref{eq:reg_s_1} and \eqref{eq:reg_s_2},
    it follows that $\abs{\tau_{k+1}(y_i)} \leq q_{k+1}$.

    \textbf{Case C.} Suppose that $m \in \lambda_{i,3}$. Then
    $(a_{k+1} - \ell)q_k \leq \val{c_1 \dotsm c_{k+1}} \leq (a_{k+1} - \ell)q_k + \theta_n$,
    $c_{k+1} = a_{k+1} - \ell$, and $c_{k+1} \neq 0$. \autoref{prp:inrc_shift} implies that
    \begin{align*}
      \irep{T^m(\infw{c}_\Delta)}{n} + n &= \abs{\tau_k(y_i)} + \ell q_k + \abs{u_{r_k+1}} + \ell q_k - \theta_n \\
                                         &= \abs{\tau_{k+1}(y_i)} - (a_{k+1}-\ell)q_k + \abs{u_{r_k+1}} + \ell q_k - \theta_n \\
                                         &= \abs{\tau_{k+1}(y_i)} - a_{k+1} q_k + 2\ell q_k + \abs{u_{r_k+1}} - \theta_n \\
                                         &= \abs{\tau_{k+1}(y_i)} - a_{k+1} q_k + 2(a_{k+1} - c_{k+1})q_k + \abs{u_{r_k+1}} - \theta_n \\
                                         &= \abs{\tau_{k+1}(y_i)} + \abs{u_{r_{k+1}+1}} - 2c_{k+1}q_k - \theta_n.
    \end{align*}
    Like in the previous case, we have $P \geq q_{k+1} + \abs{u_{r_{k+1}+1}} - \val{c_1 \dotsm c_{k+1}}$, so we just
    need to show that $\val{c_1 \dotsm c_{k+1}} \leq 2c_{k+1}q_k + \theta_n$. This is equivalent with
    $\val{c_1 \dotsm c_k} \leq c_{k+1}q_k + \theta_n$ and this in turn is true when $c_{k+1} \neq 0$ because
    $\val{c_1 \dotsm c_k} < q_k$ by \autoref{lem:ostrowski_ub}. We conclude that
    $\irep{\infw{t}}{n} = \irep{T^m(\infw{c}_\Delta)}{n}$.

    \textbf{Case D.} Suppose that $m \in \lambda_{i,4}$. Now
    \begin{equation*}
      (a_{k+1} - \ell)q_k + \theta_n < \val{c_1 \dotsm c_{k+1}} < (a_{k+1} - \ell)q_k + \abs{\tau_{k+1}(y_i)}
    \end{equation*}
    and $c_{k+1} \neq 0$. From \autoref{prp:inrc_shift}, we obtain that
    \begin{align*}
      \irep{T^m(\infw{c}_\Delta)}{n} + n &= q_k + \abs{\tau_{k+1}(y_i)} - \val{c_1 \dotsm c_{k+1}} + \abs{u_{r_k+1}} + \ell q_k - \theta_n \\
                                         &\leq q_k + q_{k+1} + \abs{u_{r_{k+1}+1}} - \val{c_1 \dotsm c_{k+1}} - a_{k+1} q_k + \ell q_k - \theta_n \\
                                         &\leq P + q_k - a_{k+1} q_k + \ell q_k - \theta_n \\
                                         &\leq P - (a_{k+1} - \ell - 1)q_k \\
                                         &\leq P
    \end{align*}
    because $P \geq q_{k+1} + \abs{u_{r_{k+1}+1}} - \val{c_1 \dotsm c_{k+1}}$, $\theta_n \geq 0$, and $\ell < a_{k+1}$.
    Therefore $\irep{\infw{t}}{n} = \irep{T^m(\infw{c}_\Delta)}{n}$ also in this case.

    We have now finished the first part of the proof. Let $m' = \val{c_1 \dotsm c_{k+d}}$, and assume that $c_i = a_i$
    for all $i$ such that $k+2 \leq i \leq k+d$. Since the intercept $c_1 c_2 \dotsm$ satisfies the Ostrowski
    conditions, it follows that $c_{k+1} = 0$ and $c_{k+d+1} < a_{k+d+1}$. Thus we immediately see that the preceding
    Cases C and D do not occur as then we had $c_{k+1} \neq 0$. The arguments given in the Case A still work, and we
    see that $\irep{\infw{t}}{n} = \irep{T^m(\infw{c}_\Delta)}{n}$ when $m \in \lambda_{i,1}$. Since the longest common
    prefix of $\infw{t}$ and $T^{m'}(\infw{c}_\Delta)$ is at least as long as that of $\infw{t}$ and
    $T^m(\infw{c}_\Delta)$ by \autoref{lem:common_prefix_length_increasing}, we deduce that
    $\irep{\infw{t}}{n} = \irep{T^{m'}(\infw{c}_\Delta)}{n}$ if $m \in \lambda_{i,1}$. We are thus left with the case
    $m \in \lambda_{i,2}$. Earlier in Case B we deduced that $\ell = a_{k+1} - c_{k+1} - 1$, that is,
    $\ell = a_{k+1} - 1$. This means that $n = \abs{u_{r_{k+1}}} - \theta_n$.

    If $T^{m'}(\infw{c}_\Delta) = \infw{t}$, then the claim is clear, so we assume that there exists a least positive
    integer $j'$ such that $c_{k+d+j'} \neq 0$ (notice that since $c_{k+d} = a_{k+d}$, we have $j = 1$). By
    \autoref{lem:common_prefix_length}, the longest common prefix of $\infw{t}$ and $T^{m'}(\infw{c}_\Delta)$ has
    length $P'$ where
    \begin{align*}
      P' &= \abs{s_{k+d+j'-1}^{\s{a}-\s{c}} s_{k+d+j'-2}^{\s{a}} \dotsm s_0^{\s{a}}} - \val{c_1 \dotsm c_{k+d}} \\
         &\geq \abs{s_{k+d}^{\s{a}-\s{c}} s_{k+d-1}^{\s{a}} \dotsm s_0^{\s{a}}} - \val{c_1 \dotsm c_{k+d}} \\
         &\geq q_{k+d} + \abs{s_{k+d-1}^{\s{a}} \dotsm s_0^{\s{a}}} - \val{c_1 \dotsm c_{k+d}} \\
         &= q_{k+d} + \abs{s_k^{\s{a}} \dotsm s_0^{\s{a}}} - \val{c_1 \dotsm c_{k+1}} \\
         &= q_{k+d} + \abs{u_{r_{k+1}+1}} - \val{c_1 \dotsm c_k}.
    \end{align*}
    We used above the facts that $c_{k+d+1} < a_{k+d+1}$, $c_{k+1} = 0$, and $c_i = a_i$ for all $i$ such that
    $k+2 \leq i \leq k+d$. Apply \autoref{prp:inrc_n} (ii) to obtain that
    \begin{equation*}
      \irep{T^{m'}(\infw{c}_\Delta)}{n'} = \abs{\tau_{k+1}(y)}
    \end{equation*}
    for a letter $y$ when $\abs{u_{r_{k+1}}} < n' \leq \abs{u_{r_{k+1}}} + q_k - \val{c_1 \dotsm c_k}$. Since the
    initial nonrepetitive complexity is nondecreasing and $\abs{\tau_{k+1}(y)} \leq q_{k+1}$ (see Case B), we find that
    \begin{equation*}
      \irep{T^{m'}(\infw{c}_\Delta)}{n} + n \leq q_{k+1} + n = q_{k+1} + \abs{u_{r_{k+1}}} - \theta_n.
    \end{equation*}
    Therefore
    \begin{align*}
      P' &\geq q_{k+d} + \abs{u_{r_{k+1}+1}} - \val{c_1 \dotsm c_k} \\
         &= q_{k+d} + \abs{u_{r_{k+1}}} + q_k - \val{c_1 \dotsm c_k} \\
         &> q_{k+1} + \abs{u_{r_{k+1}}} \\
         &\geq \irep{T^{m'}(\infw{c}_\Delta)}{n} + n.
    \end{align*}
    It follows that $\irep{\infw{t}}{n} = \irep{T^{m'}(\infw{c}_\Delta)}{n}$. This establishes (ii) and ends the proof.
  \end{proof}

  We are finally in the position to give a complete description of the initial nonrepetitive complexity of a regular
  episturmian word. This mostly amounts to putting together Propositions \ref{prp:sufficient_shift} and
  \ref{prp:inrc_n}. The statement is again complicated, but all listed cases are different. The initial nonrepetitive
  complexity of a general epistandard word is determined in
  \cite[Thm.~16]{2020:on_non-repetitive_complexity_of_arnoux-rauzy_words}.

  \begin{theorem}\label{thm:main}
    Let $\infw{t}$ be a regular episturmian word with a directive word $\Delta$ of period $d$ and intercept
    $c_1 c_2 \dotsm$. Suppose that $n \in \interval{k}$ for $k \geq 0$, and let $\Delta' = T^{r_{k+1}}(\Delta)$ and
    $y_i$ be the $i$th letter of $\infw{c}_{\Delta'}$ when $i = \val[\Delta']{c_{k+2} \dotsm c_{k+d-1}} + 1$.
    \begin{enumerate}[(i)]
      \item If $0 < c_{k+1} < a_{k+1} - 1$ or $c_{k+1} = a_{k+1} - 1 > 0$ and $c_k = 0$, then we have the following
            implications:
            \begin{itemize}
              \item $\abs{u_{r_k}} < n \leq \abs{u_{r_{k+1}}} - \val[\Delta]{c_1 \dotsm c_{k+1}} \\
                     \Longrightarrow \, \irep{\infw{t}}{n} = q_k$,
              \item $\abs{u_{r_{k+1}}} - \val[\Delta]{c_1 \dotsm c_{k+1}} < n \leq \abs{u_{r_{k+1}}} - c_{k+1}q_k \\
                     \Longrightarrow \, \irep{\infw{t}}{n} = \abs{\tau_{k+1}(y_i)} - \val[\Delta]{c_1 \dotsm c_{k+1}}$,
              \item $\abs{u_{r_{k+1}}} - c_{k+1}q_k < n \leq \abs{u_{r_{k+1}}} + q_k - \val[\Delta]{c_1 \dotsm c_{k+1}} \\
                     \Longrightarrow \, \irep{\infw{t}}{n} = (a_{k+1}-c_{k+1})q_k + \abs{\tau_k(y_i)}$,
              \item $\abs{u_{r_{k+1}}} + q_k - \val[\Delta]{c_1 \dotsm c_{k+1}} < n \leq \abs{u_{r_{k+1}}} \\
                     \Longrightarrow \, \irep{\infw{t}}{n} = q_k + \abs{\tau_{k+1}(y_i)} - \val[\Delta]{c_1 \dotsm c_{k+1}}$.
            \end{itemize}
      \item If $c_{k+1} = a_{k+1}$, then we have the following implications:
            \begin{itemize}
              \item $\abs{u_{r_k}} < n \leq \abs{u_{r_k+1}} - \val[\Delta]{c_1 \dotsm c_k} \\
                     \Longrightarrow \, \irep{\infw{t}}{n} = \abs{\tau_k(y_i)}$,
              \item $\abs{u_{r_k+1}} - \val[\Delta]{c_1 \dotsm c_k} < n \leq \abs{u_{r_{k+1}}} \\
                     \Longrightarrow \, \irep{\infw{t}}{n} = q_k + \abs{\tau_k(y_i)} - \val[\Delta]{c_1 \dotsm c_k}$.
            \end{itemize}
      \item If $c_{k+1} = a_{k+1} - 1 > 0$ and $c_k \neq 0$, then we have the following implications:
            \begin{itemize}
              \item $\abs{u_{r_k}} < n \leq \abs{u_{r_k+1}} \\
                     \Longrightarrow \, \irep{\infw{t}}{n} = q_k + \abs{\tau_k(y_i)} - \val[\Delta]{c_1 \dotsm c_k}$,
              \item $\abs{u_{r_k+1}} < n \leq \abs{u_{r_k+1}} + q_k - \val[\Delta]{c_1 \dotsm c_k} \\
                     \Longrightarrow \, \irep{\infw{t}}{n} = q_k + \abs{\tau_k(y_i)}$,
              \item $\abs{u_{r_k+1}} + q_k - \val[\Delta]{c_1 \dotsm c_k} < n \leq \abs{u_{r_{k+1}}} \\
                     \Longrightarrow \, \irep{\infw{t}}{n} = 2q_k + \abs{\tau_k(y_i)} - \val[\Delta]{c_1 \dotsm c_k}$.
            \end{itemize}
      \item Suppose that $c_{k+1} = 0$ and there exists $i$ such that $k+2 \leq i \leq k+d$ and $c_i < a_i$.
            \begin{enumerate}
              \item[(iv.a)] If $c_k = 0$, then we have the following implications:
                \begin{itemize}
                  \item $\abs{u_{r_k}} < n \leq \abs{u_{r_{k+1}}} - \val[\Delta]{c_1 \dotsm c_{k-1}} \\
                         \Longrightarrow \, \irep{\infw{t}}{n} = q_k$,
                  \item $\abs{u_{r_{k+1}}} - \val[\Delta]{c_1 \dotsm c_{k-1}} < n \leq \abs{u_{r_{k+1}}} \\
                         \Longrightarrow \, \irep{\infw{t}}{n} = \abs{\tau_{k+1}(y_i)} - \val[\Delta]{c_1 \dotsm c_{k-1}}$.
                \end{itemize}
              \item[(iv.b)] If $a_{k+1} = 1$ and $c_k \neq 0$, then we have the following implication:
                \begin{itemize}
                  \item $\abs{u_{r_k}} < n \leq \abs{u_{r_{k+1}}} \\
                         \Longrightarrow \, \irep{\infw{t}}{n} = \abs{\tau_{k+1}(y_i)} - \val[\Delta]{c_1 \dotsm c_k}$.
                \end{itemize}
              \item[(iv.c)] If $a_{k+1} > 1$ and $c_k \neq 0$, then we have the following implications:
                \begin{itemize}
                  \item $\abs{u_{r_k}} < n \leq \abs{u_{r_{k+1}}} - \val[\Delta]{c_1 \dotsm c_k} \\
                         \Longrightarrow \, \irep{\infw{t}}{n} = q_k$.
                  \item $\abs{u_{r_{k+1}}} - \val[\Delta]{c_1 \dotsm c_k} < n \leq \abs{u_{r_{k+1}}} \\
                         \Longrightarrow \, \irep{\infw{t}}{n} = \abs{\tau_{k+1}(y_i)} - \val[\Delta]{c_1 \dotsm c_k}$.
                \end{itemize}
            \end{enumerate}
      \item If $c_{k+1} = 0$ and $a_i = c_i$ for all $i$ such that $k+2 \leq i \leq k+d$, then we have the following
            implication:
            \begin{itemize}
              \item $\abs{u_{r_k}} < n \leq \abs{u_{r_{k+1}}} \\
                     \Longrightarrow \, \irep{\infw{t}}{n} = q_k$.
            \end{itemize}
    \end{enumerate}
  \end{theorem}
  \begin{proof}
    Let $m = \val{c_1 \dotsm c_{k+d-1}}$ and $m' = \val{c_1 \dotsm c_{k+d}}$.

    Suppose that $0 < c_{k+1} < a_{k+1} - 1$ or $c_{k+1} = a_{k+1} - 1 > 0$ and $c_k = 0$. Then the Ostrowski
    conditions guarantee that there exists $i$ such that $k+2 \leq i \leq k+d$ and $c_i < a_i$. Therefore
    \autoref{prp:sufficient_shift} implies that $\irep{\infw{t}}{n} = \irep{T^m(\infw{c}_\Delta)}{n}$. Thus (i) follows
    directly from \autoref{prp:inrc_n} (i). We apply \autoref{prp:inrc_n} similarly when $c_{k+1} = a_{k+1}$ or
    $c_{k+1} = a_{k+1} - 1 > 0$ and $c_k \neq 0$ to obtain (ii) and (iii). When there exists $i$ such that
    $k+2 \leq i \leq k+d$ and $c_i < a_i$, then we obtain (iv) directly from \autoref{prp:inrc_n}.

    Consider the final case (v), that is, assume that $c_i = a_i$ for all $i$ such that $k+2 \leq i \leq k+d$. Then the
    Ostrowski conditions imply that $c_{k+1} = 0$. In the latter part of the proof of \autoref{prp:sufficient_shift},
    we saw that our assumptions implied that $m \notin \lambda_{i,3}, \lambda_{i,4}$. Moreover, we proved that if
    $m \in \lambda_{i,1}$, then $\irep{\infw{t}}{n} = \irep{T^m(\infw{c}_\Delta)}{n}$. If $m \in \lambda_{i,1}$, then
    the Case A of the proof of \autoref{prp:inrc_n} applies, and we have $\irep{T^m(\infw{c}_\Delta)}{n} = q_k$ when
    $\abs{u_{r_k}} < n \leq \abs{u_{r_k+1}} - \val[\Delta]{c_1 \dotsm c_k}$. Assume then that $m \in \lambda_{i,2}$. As
    we saw in the Case B of the proof of \autoref{prp:sufficient_shift}, this means that $n \in \interval{k,\ell}$ with
    $\ell = a_{k+1}-1$. Now $\irep{\infw{t}}{n} = \irep{T^{m'}(\infw{c}_\Delta)}{n}$ by \autoref{prp:sufficient_shift},
    so it suffices to prove that $\irep{T^{m'}(\infw{c}_\Delta)}{n} = q_k$ to obtain (v).

    Recall that the number $\irep{T^m(\infw{c}_\Delta)}{n}$ is determined by the cycle sequence taken in the graph
    $\Gamma(n)$ and that the cycle sequence is determined by the letters of the standard word $\infw{c}_{\Delta''}$
    with $\Delta'' = T^{r_k+\ell}(\Delta) = x_{k+1} x_{k+2}^{a_{k+2}} \dotsm$. Recall also that the prefix of
    $T^{L_i}(\infw{c}_\Delta)$ of length $n$ corresponds to the left special vertex of $\Gamma(n)$ and that $L_i$
    equals the length of the $\mu_{r_k+\ell}$-images of the first $i-1$ blocks. Phrased alternatively, $L_i$ is the
    length of the $\tau_{k+1}$-image of the prefix of $\Delta'$ of length $i-1$ where
    $\Delta' = T^{r_{k+1}}(\Delta) = T(\Delta'')$. Now $\smash[t]{m' - m = c_{k+d} q_{k+d-1}}$, so reading off the
    first $\smash[t]{L_i + m' - m}$ letters of $\infw{c}_\Delta$ corresponds exactly to taking the cycles determined by
    the $\smash[t]{L_{x_{k+1}}^{a_{k+1}-\ell}}$-image of $\infw{c}_{\Delta'}$ of length
    $\val[\Delta']{c_{k+2} \dotsm c_{k+d}}$.

    Let us then find out the letter $y$ of $\infw{c}_{\Delta'}$ at position
    $\val[\Delta']{c_{k+2} \dotsm c_{k+d}} + 1$. By \eqref{eq:reg_s_1}, we have $v_{d-1} = v_{d-2}^{a_{k+d}} \dotsm
    v_0^{a_{k+2}} x_{k+1}$ for standard words $v_t$ with the directive word $\Delta'$. The length of the prefix of
    $v_{d-1}$ obtained by removing the suffix $x_{k+1}$ equals $\val[\Delta']{c_{k+2} \dotsm c_{k+d}}$, so
    $y = x_{k+1}$. The next cycles taken after reading off the first $L_i + m' - m$ letters of $\infw{c}_\Delta$ are
    thus determined by the $\smash[t]{L_{x_{k+1}}^{a_{k+1}-\ell}}$-image of the letter $y$ and the letter immediately
    after it. Since the $\smash[t]{L_{x_{k+1}}^{a_{k+1}-\ell}}$-image of any letter begins with $x_{k+1}$, we see that
    $\smash[t]{T^{L_i + m' - m}(\infw{c}_\Delta)}$ initially takes the cycle $C_{x_{k+1}}$ twice.

    Since $m \in \lambda_{i,2}$, we see as in the Case B of the proof of \autoref{prp:inrc_shift} that the prefix of
    $T^m(\infw{c}_\Delta)$ of length $n$ lies on the cycle $C_{x_{k+1}}$ but not on the central path. In other words,
    when $m - L_i$ edges are traversed on the cycle $C_{x_{k+1}}$, starting from the left special vertex, the whole
    cycle is not traversed. Consider now the word $\smash[t]{T^{L_i + m' - m}(\infw{c}_\Delta})$ that takes the cycle
    $C_{x_{k+1}}$ twice. When $m - L_i$ more letters are read, the current vertex is vertex $v$ on the cycle
    $C_{x_{k+1}}$. When $\norm{C_{x_{k+1}}}$ more letters are read, the vertex $v$ is encountered again. Since
    $L_i + m' - m + m - L_i = m'$, we see that $\irep{T^{m'}(\infw{c}_\Delta)}{n} = \norm{C_{x_{k+1}}}$. Now
    $\norm{C_{x_{k+1}}} = q_k$, and the claim is proved.
  \end{proof}

  The Sturmian case was determined in \cite[Th\'{e}or.~5.3.0.4]{diss:caius_wojcik}.

  \begin{theorem}[Sturmian Case]\label{thm:main_sturmian}
    Let $\infw{t}$ be a Sturmian word with directive word $\Delta$ and intercept $c_1 c_2 \dotsm$. Suppose that
    $n \in \interval{k}$ for $k \geq 0$.
    \begin{enumerate}[(i)]
      \item If $0 < c_{k+1} < a_{k+1} - 1$ or $c_{k+1} = a_{k+1} - 1 > 0$ and $c_k = 0$, then we have the following
            implications:
            \begin{itemize}
              \item $q_k - 2 < n \leq q_{k+1} - 2 - \val{c_1 \dotsm c_{k+1}} \\
                     \Longrightarrow \, \irep{\infw{t}}{n} = q_k$,
              \item $q_{k+1} - 2 - \val{c_1 \dotsm c_{k+1}} < n \leq q_{k+1} - 2 - c_{k+1}q_k \\
                     \Longrightarrow \, \irep{\infw{t}}{n} = q_{k+1} - \val{c_1 \dotsm c_{k+1}}$,
              \item $q_{k+1} - 2 - c_{k+1}q_k < n \leq q_{k+1} + q_k - 2 - \val{c_1 \dotsm c_{k+1}} \\
                     \Longrightarrow \, \irep{\infw{t}}{n} = q_{k+1} - c_{k+1}q_k$,
              \item $q_{k+1} + q_k - 2 - \val{c_1 \dotsm c_{k+1}} < n \leq q_{k+1} - 2 \\
                     \Longrightarrow \, \irep{\infw{t}}{n} = q_{k+1} + q_k - \val{c_1 \dotsm c_{k+1}}$.
            \end{itemize}
      \item If $c_{k+1} = a_{k+1}$, then we have the following implications:
            \begin{itemize}
              \item $q_k - 2 < n \leq q_{k+1} - 2 - \val{c_1 \dotsm c_k} \\
                     \Longrightarrow \, \irep{\infw{t}}{n} = q_{k-1}$,
              \item $q_{k+1} - 2 - \val{c_1 \dotsm c_k} < n \leq q_{k+1} - 2 \\
                     \Longrightarrow \, \irep{\infw{t}}{n} = q_k + q_{k-1} - \val{c_1 \dotsm c_k}$.
            \end{itemize}
      \item If $c_{k+1} = a_{k+1} - 1 > 0$ and $c_k \neq 0$, then we have the following implications:
            \begin{itemize}
              \item $q_k - 2 < n \leq q_k + q_{k-1} - 2 \\
                     \Longrightarrow \, \irep{\infw{t}}{n} = q_k + q_{k-1} - \val{c_1 \dotsm c_k}$,
              \item $q_k + q_{k-1} - 2 < n \leq 2q_k + q_{k-1} - 2 - \val{c_1 \dotsm c_k} \\
                     \Longrightarrow \, \irep{\infw{t}}{n} = q_k + q_{k-1}$,
              \item $2q_k + q_{k-1} - 2 - \val{c_1 \dotsm c_k} < n \leq q_{k+1} - 2 \\
                     \Longrightarrow \, \irep{\infw{t}}{n} = 2q_k + q_{k-1} - \val{c_1 \dotsm c_k}$.
            \end{itemize}
          \item Suppose that $c_{k+1} = 0$ and $c_{k+2} < a_{k+2}$.
            \begin{enumerate}
              \item[(iv.a)] If $c_k = 0$, then we have the following implications:
                \begin{itemize}
                  \item $q_k - 2 < n \leq q_{k+1} - 2 - \val{c_1 \dotsm c_{k-1}} \\
                         \Longrightarrow \, \irep{\infw{t}}{n} = q_k$,
                  \item $q_{k+1} - 2 - \val{c_1 \dotsm c_{k-1}} < n \leq q_{k+1} - 2 \\
                         \Longrightarrow \, \irep{\infw{t}}{n} = q_{k+1} - \val{c_1 \dotsm c_{k-1}}$.
                \end{itemize}
              \item[(iv.b)] If $a_{k+1} - 1 = c_{k+1}$ and $c_k \neq 0$, then we have the following implication:
                \begin{itemize}
                  \item $q_k - 2 < n \leq q_{k+1} - 2 \\
                         \Longrightarrow \, \irep{\infw{t}}{n} = q_{k+1} - \val{c_1 \dotsm c_k}$.
                \end{itemize}
              \item[(iv.c)] If $a_{k+1} - 1 > c_{k+1}$ and $c_k \neq 0$, then we have the following implications:
                \begin{itemize}
                  \item $q_k - 2 < n \leq q_{k+1} - 2 - \val{c_1 \dotsm c_k} \\
                         \Longrightarrow \, \irep{\infw{t}}{n} = q_k$.
                  \item $q_{k+1} - 2 - \val{c_1 \dotsm c_k} < n \leq q_{k+1} - 2 \\
                         \Longrightarrow \, \irep{\infw{t}}{n} = q_{k+1} - \val{c_1 \dotsm c_k}$.
                \end{itemize}
            \end{enumerate}
        \item If $c_{k+1} = 0$ and $a_{k+2} = c_{k+2}$, then we have the following implication:
            \begin{itemize}
              \item $q_k - 2 < n \leq q_{k+1} - 2 \\
                     \Longrightarrow \, \irep{\infw{t}}{n} = q_k$.
            \end{itemize}
    \end{enumerate}
  \end{theorem}
  \begin{proof}
    Utilize again the computations of \autoref{ex:sturmian_intervals}.
  \end{proof}

  \section{Diophantine Exponents}\label{sec:diophantine}
  Recall the definition of the Diophantine exponent $\dio{\infw{x}}$ of an infinite word $\infw{x}$ from
  \autoref{ssec:dio}. It is clear from the definition that $\dio{\infw{x}} \geq 1$ for all infinite words $\infw{x}$.
  It is possible that $\dio{\infw{x}} = \infty$ (as we shall see). It is proved in
  \cite{2009:binary_words_with_a_given_diophantine_exponent} that for each real number $\theta$ such that
  $\theta \geq 1$, there exists an infinite binary word $\infw{x}$ such that $\dio{\infw{x}} = \theta$. Propositions
  4.3 and 6.2 of the recent paper \cite{2020:on_the_computational_complexity_of_algebraic_numbers_the_hartmanis}
  provide lower and upper bounds for Diophantine exponents of infinite words generated by morphisms.

  Diophantine exponents relate to our work through the following result that allows us to compute the Diophantine
  exponent of a regular episturmian word with the help of \autoref{thm:main}.

  \begin{proposition}\label{prp:dio_inrc}\cite[Lemma~10.3]{2019:a_new_complexity_function_repetitions_in_sturmian}
    If $\infw{x}$ is an infinite word, then
    \begin{equation*}
      \dio{\infw{x}} = 1 + \limsup_{n\to\infty} \frac{n}{\irep{\infw{x}}{n}}.
    \end{equation*}
  \end{proposition}
  \begin{proof}
    There are some notational differences. The claim follows directly from
    \cite[Lemma~10.3]{2019:a_new_complexity_function_repetitions_in_sturmian} with the observation that
    $\irep{\infw{x}}{n} = r(n,\infw{x}) - n$ where $r$ is defined as on p. 3282 of
    \cite{2019:a_new_complexity_function_repetitions_in_sturmian}.
  \end{proof}

  By \autoref{prp:dio_inrc} and \autoref{thm:main}, finding $\dio{\infw{t}}$ for a regular episturmian word
  $\infw{t}$ amounts to determining the ratio $n/\irep{\infw{t}}{n}$ on the right endpoint of each of the subintervals
  of $\interval{k}$ described by \autoref{thm:main}.

  Let us then define the closely-related notions of initial critical exponent and index.

  \begin{definition}
    Let $\infw{x}$ be an infinite word. The \emph{initial critical exponent} of $\infw{x}$, denoted by
    $\ice{\infw{x}}$, is the supremum of the rational numbers $\rho$ for which there exist arbitrarily long prefixes
    $V$ of $\infw{x}$ such that $V^\rho$ is a prefix of $\infw{x}$.
  \end{definition}

  A formula for the the initial critical exponent of a Sturmian word is given in
  \cite[Cor.~3.5]{2006:initial_powers_of_sturmian_sequences}. According to our knowledge, no attempt to study the
  initial critical exponents of general episturmian words has been attempted. The methods used to prove
  \autoref{thm:main} could be used for such a study. It is worth mentioning that a slight modification of the proof
  of \cite[Lemma~10.3]{2019:a_new_complexity_function_repetitions_in_sturmian} yields
  \begin{equation*}
    \ice{\infw{x}} = 1 + \limsup_{n\to\infty} \frac{n}{\prep{\infw{x}}{n}}
  \end{equation*}
  for an infinite word $\infw{x}$. Here $\prep{\infw{x}}{n}$ is the \emph{prefix nonrepetitive complexity function} of
  $\infw{x}$ defined by setting
  \begin{equation*}
    \prep{\infw{x}}{n} = \max\{m \colon \text{$\infw{x}[i,i+n-1] \neq \infw{x}[1,n-1]$ for all $2 \leq i < m$}\}.
  \end{equation*}
  Thus the prefix of $\infw{x}$ of length $\prep{\infw{x}}{n} + n$ is the shortest prefix of $\infw{x}$ containing two
  occurrences of the prefix of $\infw{x}$ of length $n$. We have the obvious relation
  $\irep{\infw{x}}{n} \leq \prep{\infw{x}}{n}$ for all $n$.

  \begin{definition}
    Let $\infw{x}$ be an infinite word. The \emph{index} (or critical exponent) of $\infw{x}$, denoted by
    $\ind{\infw{x}}$, is the number
    \begin{equation*}
      \sup\{e \in \mathbb{Q} \cap [1, \infty) : \text{$\infw{x}$ has an $e$-power as a factor}\}.
    \end{equation*}
  \end{definition}

  The indices have been determined for various classes of infinite words. It suffices to say here that a formula for
  the index of Sturmian words was derived independently in
  \cite{2000:special_factors_periodicity_and_an_application_to_sturmian,2001:fractional_powers_in_sturmian_words,2002:the_index_of_sturmian_sequences}
  (see also \cite{2015:characterization_of_repetitions_in_sturmian_words_a_new}). The index of a regular episturmian
  word can be found by applying \cite[Lemma~6.5,Thm.~6.19]{2007:powers_in_a_class_of_A_strict_standard_episturmian}. It
  seems that the discussion in Sections 4.1 and 5.5 of \cite{2002:episturmian_words_and_episturmian_morphisms} is all
  that has been said of the indices of general episturmian words.

  From the definitions, we infer the following inequalities:
  \begin{equation*}
    \ice{\infw{x}} \leq \dio{\infw{x}} \leq \ind{\infw{x}}.
  \end{equation*}
  It is possible that $\ice{\infw{x}} < \dio{\infw{x}}$ (see \cite[p.~18]{2006:initial_powers_of_sturmian_sequences})
  or $\dio{\infw{x}} < \ind{\infw{x}}$ (see \autoref{prp:standard_dio}).

  Let us then show that the Diophantine exponent is shift-invariant. The same is not true for the initial critical
  exponent; see \cite[Prop.~2.1,~Sect.~4.2]{2006:initial_powers_of_sturmian_sequences}.

  \begin{proposition}\label{prp:dio_shift_invariant}
    Let $\infw{x}$ be an infinite word. Then $\dio{T^m(\infw{x})} = \dio{\infw{x}}$ for all $m \geq 0$.
  \end{proposition}
  \begin{proof}
    Let $(U_n)$, $(V_n)$, and $(e_n)$ be sequences such that $\abs{U_n V_n^{e_n}}/\abs{U_n V_n}$ converges to
    $\dio{\infw{x}}$ and $(\abs{U_n V_n})$ is increasing. Since either $\abs{U_n} \to \infty$ or
    $\abs{V_n} \to \infty$, we may factorize the prefix of $T^m(\infw{x})$ of length $\abs{U_n V_n^{e_n}} - m$ as
    $\smash[t]{U'_n {V'_n}^{e'_n}}$ for a conjugate $V'_n$ of $V_n$ for $n$ large enough. For $n$ large enough, we have
    \begin{equation*}
      \frac{\abs{U'_n {V'_n}^{e'_n}}}{\abs{U'_n V'_n}} = \frac{\abs{U_n V_n^{e_n}} - m}{\abs{U_n V_n} - m} =
      \frac{\abs{U_n V_n^{e_n}}/\abs{U_n V_n} - m/\abs{U_n V_n}}{1 - m/\abs{U_n V_n}} \to \lim_{n\to\infty}
      \frac{\abs{U_n V_n^{e_n}}}{\abs{U_n V_n}} = \dio{\infw{x}}.
    \end{equation*}
    Therefore $\dio{T^m(\infw{x})} \geq \dio{\infw{x}}$. A similar argument establishes that
    $\ice{T^m(\infw{x})} \geq \ice{\infw{x}}$.

    A symmetric argument where we prepend a factor of length $m$ to words $U_n V_n^{e_n}$ such that
    $(\abs{U_n V_n})$ is increasing and $\abs{U_n V_n^{e_n}}/\abs{U_n V_n}$ converges to $\dio{T^m(\infw{x})}$
    establishes that $\dio{\infw{x}} \geq \dio{T^m(\infw{x})}$. This argument does not apply to the initial critical
    exponent unless the prepended factor is compatible with the prefix powers of $T^m(\infw{x})$.
  \end{proof}

  \subsection{Diophantine Exponents of Regular Episturmian Words}
  As the first application of \autoref{thm:main}, we find the Diophantine exponent of a regular epistandard word.

  \begin{proposition}\label{prp:standard_dio}
    Let $\Delta$ be a regular directive word as in \eqref{eq:dw_multiplicative}. Then
    \begin{equation*}
      \dio{\infw{c}_\Delta} = 1 + \limsup_{k\to\infty} (a_{k+1} + \abs{u_{r_{k-(d-1)}}}/q_k) = \ind{\infw{c}_\Delta} - 1.
    \end{equation*}
  \end{proposition}
  \begin{proof}
    The intercept of $\infw{c}_\Delta$ is $0^\omega$. Thus we find from the first implication of \autoref{thm:main}
    (iv.a) that
    \begin{equation*}
      \max_{n\in\interval{k}} \frac{n}{\irep{\infw{c}_\Delta}{n}} = \frac{\abs{u_{r_{k+1}}}}{q_k}
    \end{equation*}
    for $k \geq 0$. Since
    $u_{r_{k+1}} = s_k^{\s{a}-1} s_{k-1}^{\s{a}} \dotsm s_0^{\s{a}} = s_k^{\s{a}} u_{r_{k-d+1}}$ by \eqref{eq:u_4} and
    \eqref{eq:reg_s_2}, it follows from \autoref{prp:dio_inrc} that
    $\dio{\infw{c}_\Delta} = 1 + \limsup_{k\to\infty} (a_{k+1} + \abs{u_{r_{k-(d-1)}}}/q_k)$. If $\infw{c}_\Delta$ has
    unbounded partial quotients, that is, the sequence $(a_k)$ is unbounded, then $\dio{\infw{c}_\Delta} = \infty$.
    Since $\dio{\infw{c}_\Delta} \leq \ind{\infw{c}_\Delta}$, the claim follows. If $\infw{c}_\Delta$ has bounded
    partial quotients, then it follows from
    \cite[Lemma~6.5,Thm.~6.19]{2007:powers_in_a_class_of_A_strict_standard_episturmian} that
    \begin{equation*}
      \ind{\infw{c}_\Delta} = 2 + \limsup_{k\to\infty} (a_{k+1} + \abs{u_{r_{k-(d-1)}}}/q_k) \qedhere
    \end{equation*}
  \end{proof}

  We saw in the preceding proof that $\dio{\infw{c}_\Delta} = \infty$ if and only if $\infw{c}_\Delta$ has unbounded
  partial quotients. This is in fact true for all words in the subshift generated by $\infw{c}_\Delta$ as the next
  theorem states. This theorem was originally proved for Sturmian words in
  \cite[Prop.~11.1]{2011:nombres_reels_de_complexite_sous-lineaire} and an alternative proof was given in
  \cite[Thm.~3.3]{2019:a_new_complexity_function_repetitions_in_sturmian}. Our more general proof uses yet another
  approach.

  \begin{theorem}\label{thm:dio_upq}
    Let $\infw{t}$ be a regular episturmian word. Then $\dio{\infw{t}} < \infty$ if and only if $\infw{t}$ has
    bounded partial quotients.
  \end{theorem}

  Before the proof, let us derive some helpful inequalities.

  \begin{lemma}
    Let $\infw{t}$ be a regular episturmian word of period $d$ with directive word $\Delta$ and intercept
    $c_1 c_2 \dotsm$. Let $y$ be a letter occurring in $\Delta$ and $Q$ a nonnegative integer. Then
    \begin{align}
      \frac{\abs{u_{r_{k+1}}}}{q_k + \abs{\tau_{k+1}(y)} - \val{c_1 \dotsm c_{k+1}}} \geq \frac{a_{k+1} + q_{k-(d+1)}/q_k}{a_{k+1} - c_{k+1} + 2}, \label{eq:helpful_1} \\
      \frac{\abs{u_{r_{k+1}}} - \val{c_1 \dotsm c_{k+1}}}{q_k} \geq a_{k+1} - c_{k+1} - 1 + q_{k-(d+1)} / q_k, \label{eq:helpful_2} \\
      \frac{\abs{u_{r_{k+1}}}}{Q q_k + \abs{\tau_k(y)} - \val{c_1 \dotsm c_k}} \geq \frac{1}{Q+1}\left( a_{k+1} + q_{k-(d+1)} / q_k \right) \label{eq:helpful_3}
    \end{align}
    for all $k$ large enough.
  \end{lemma}
  \begin{proof}
    From \autoref{lem:ostrowski_ub}, \eqref{eq:reg_q}, \autoref{lem:regular_help}, and \eqref{eq:u_5}, we find
    that
    \begin{align*}
      0 &\leq \val{c_1 \dotsm c_k} < q_k, \\
      q_{k+1} &< (a_{k+1} + 1)q_k, \\
      \abs{\tau_k(y)} &\leq q_k, \text{ and} \\
      a_k q_{k-1} + q_{k-(d+2)} &\leq \abs{u_{r_k}}
    \end{align*}
    for all $k$ large enough. Therefore
    \begin{align*}
      \frac{\abs{u_{r_{k+1}}}}{q_k + \abs{\tau_{k+1}(y)} - \val{c_1 \dotsm c_{k+1}}}
      &\geq \frac{a_{k+1} q_k + q_{k-(d+1)}}{q_k + q_{k+1} - \val{c_1 \dotsm c_k} - c_{k+1} q_k} \\
      &\geq \frac{a_{k+1} q_k + q_{k-(d+1)}}{q_k + (a_{k+1} + 1)q_k - c_{k+1} q_k} \\
      &= \frac{a_{k+1} + q_{k-(d+1)} / q_k}{a_{k+1} - c_{k+1} + 2}
    \end{align*}
    for all $k$ large enough. This proves the first inequality. The second inequality is derived similarly:
    \begin{align*}
      \frac{\abs{u_{r_{k+1}}} - \val{c_1 \dotsm c_{k+1}}}{q_k}
      &\geq \frac{a_{k+1} q_k + q_{k-(d+1)} - \val{c_1 \dotsm c_{k+1}}}{q_k} \\
      &\geq \frac{a_{k+1} q_k + q_{k-(d+1)} - c_{k+1}q_k - q_k}{q_k} \\
      &= a_{k+1} - c_{k+1} - 1 + q_{k-(d+1)} / q_k
    \end{align*}
    for all $k$ large enough. The third inequality is derived analogously by using the inequality
    $\abs{u_{r_{k+1}}} \geq a_{k+1} q_k + q_{k-(d+1)}$ in the numerator and the inequalities
    $\abs{\tau_k(y)} \leq q_k$ and $\val{c_1 \dotsm c_k} \geq 0$ in the numerator.
  \end{proof}

  \begin{proof}[Proof of \autoref{thm:dio_upq}]
    Say $\infw{t}$ has directive word as in \eqref{eq:dw_multiplicative} and intercept $c_1 c_2 \dotsm$. If $\infw{t}$
    has bounded partial quotients, then $\ind{\infw{t}} < \infty$
    \cite[Sect.~4.1]{2002:episturmian_words_and_episturmian_morphisms}, so
    $\dio{\infw{t}} \leq \ind{\infw{t}} < \infty$. We may thus assume that $\infw{t}$ has unbounded partial quotients.
    Thus there exists a sequence $(k_j)$ such that $\smash[b]{a_{k_j+1} \to \infty}$ as $j \to \infty$. Below we impose
    restrictions on $(k_j)$ and show that under each restriction $\dio{\infw{t}} = \infty$. It is straightforward to
    see that (by taking an appropriate subsequence) some $(k_j)$ must satisfy at least one of the restrictions or
    $c_k = 0$ for $k$ large enough. When $c_k = 0$ for all $k$ large enough, the word $\infw{t}$ is a shift of
    $\infw{c}_\Delta$, and Propositions \ref{prp:standard_dio} and \ref{prp:dio_shift_invariant} imply that
    $\dio{\infw{t}} = \infty$.

    \textbf{Case A.} Assume that $0 < c_{k_j+1} < a_{k_j+1} - 1$ for all $j$. Suppose additionally that there exists a
    nonnegative constant $M$ such that $a_{k_j+1} - c_{k_j+1} \leq M$ for all $j$. By the fourth implication of
    \autoref{thm:main} (i), it suffices to show that
    \begin{equation*}
      \frac{\abs{u_{r_{k_j+1}}}}{q_{k_j} + \abs{\tau_{k_j+1}(y_i)} - \val{c_1 \dotsm c_{k_j+1}}} \xrightarrow{j\to\infty} \infty
    \end{equation*}
    in order to conclude that $\dio{\infw{t}} = \infty$ (here $y_i$ is as in \autoref{thm:main}). By
    \eqref{eq:helpful_1}, we have
    \begin{equation*}
       \frac{\abs{u_{r_{k_j+1}}}}{q_{k_j} + \abs{\tau_{k_j+1}(y_i)} - \val{c_1 \dotsm c_{k_j+1}}} \geq \frac{a_{k_j+1} + q_{k_j-(d+1)}/q_{k_j}}{a_{k_j+1} - c_{k_j+1} + 2} \geq \frac{1}{M+2} a_{k_j+1}
    \end{equation*}
    for $j$ large enough. As $a_{k_j+1} \to \infty$ as $j \to \infty$, it follows that $\dio{\infw{t}} = \infty$.
    Assume then that the difference $a_{k_j+1} - c_{k_j+1}$ is unbounded. This time around, we consider the first
    implication of \autoref{thm:main} (i). Notice that we need to ensure that the interval
    $\smash[b]{\{\abs{u_{r_{k_j}}}, \ldots, \abs{u_{r_{k_j+1}}} - \val{c_1 \dotsm c_{k_j+1}}\}}$ is nonempty, but this
    is straightforward to do. From \eqref{eq:helpful_2}, we have
    \begin{equation*}
      \frac{\abs{u_{r_{k_j+1}}} - \val{c_1 \dotsm c_{k_j+1}}}{q_{k_j}} \geq a_{k_j+1} - c_{k_j+1} - 1 + q_{k_j-(d+1)} / q_{k_j} \xrightarrow{j\to\infty} \infty
    \end{equation*}
    because the difference $a_{k_j+1} - c_{k_j+1}$ is unbounded. Hence $\dio{\infw{t}} = \infty$ in this case as well.

    \textbf{Case B.} If $\smash[b]{c_{k_j+1} = a_{k_j+1} - 1 > 0}$ and $\smash[b]{c_{k_j} = 0}$ for all $j$, then the
    difference $a_{k_j+1} - c_{k_j+1}$ is bounded. By repeating the arguments of Case A, we see that
    $\dio{\infw{t}} = \infty$.

    \textbf{Case C.} Suppose that $c_{k_j+1} = a_{k_j+1}$ for all $j$. From \eqref{eq:helpful_3}, we see that
    \begin{equation*}
      \frac{\abs{u_{r_{k_j+1}}}}{q_{k_j} + \abs{\tau_{k_j}(y_i)} - \val{c_1 \dotsm c_k}} \geq \frac12(a_{k_j+1} + q_{k-(d+1)}/q_k)
    \end{equation*}
    so $\dio{\infw{t}} = \infty$ by the last implication of \autoref{thm:main} (ii).

    \textbf{Case D.} Assume that $c_{k_j+1} = a_{k_j+1} - 1 > 0$ and $c_{k_j} \neq 0$ for all $j$. We apply the last
    implication of \autoref{thm:main_sturmian} (iii) and obtain similar to the Case C that
    \begin{equation*}
      \frac{\abs{u_{r_{k_j+1}}}}{2q_{k_j} + \abs{\tau_{k_j}(y_i)} - \val{c_1 \dotsm c_{k_j}}} \geq \frac13(a_{k_j+1} + q_{k-(d+1)}/q_k).
    \end{equation*}
    Hence again $\dio{\infw{t}} = \infty$.
  \end{proof}

  Interestingly a statement analogous to \autoref{thm:dio_upq} is not true for the initial critical exponent. It is
  shown in \cite[Prop.~4.1]{2006:initial_powers_of_sturmian_sequences} that every Sturmian subshift contains a word
  $\infw{t}$ such that $\ice{\infw{t}} \leq 1 + \varphi \approx 2.6180$ where $\varphi$ is the Golden ratio. Even more
  interestingly it is possible that $\ice{\infw{t}} = 2$ for certain Sturmian words $\infw{t}$ with unbounded partial
  quotients \cite[Thm~1.1]{2006:initial_powers_of_sturmian_sequences}. We suspect that it is equally possible that the
  initial critical exponent is finite while the Diophantine exponent is infinite when $d > 2$.

  Next we show that $\dio{\infw{t}} > 2$ for essentially all regular episturmian words. For Sturmian words, this result
  can be inferred from the results of \cite{2006:initial_powers_of_sturmian_sequences} as indicated in the proof of
  \cite[Prop.~4]{2010:on_the_expansion_of_some_exponential_periods_in_an_integer}.

  \begin{theorem}\label{thm:dio_lower_bound}
    Let $\infw{t}$ be a regular episturmian word of period $d$ with directive word as in \eqref{eq:dw_multiplicative}.
    If $d = 2$ or $\limsup_k a_k \geq 3$, then $\dio{\infw{t}} > 2$.
  \end{theorem}
  \begin{proof}
    The claim follows from \autoref{thm:dio_upq} if $\infw{t}$ has unbounded partial quotients. Suppose that $\infw{t}$
    has bounded partial quotients and intercept $c_1 c_2 \dotsm$. Let $C = M + 1$ where $M = \limsup_k a_k$. It follows
    from \eqref{eq:reg_q} that $q_{k+1} / q_k \leq C$ for all $k$ large enough.

    \textbf{Case A.} Assume first that there exists infinitely many $k$ such that $0 < c_{k+1} < a_{k+1} - 1$. By the
    first implication of \autoref{thm:main} (i), we have
    \begin{equation*}
      \dio{\infw{t}} - 1 \geq \limsup_{k\to\infty} \frac{\abs{u_{r_{k+1}}} - \val{c_1 \dots c_{k+1}}}{q_k}
    \end{equation*}
    where we consider the limit superior over an appropriate subsequence like in the proof of \autoref{thm:dio_upq}.
    From \eqref{eq:helpful_2}, we find that
    \begin{equation*}
      \frac{\abs{u_{r_{k+1}}} - \val{c_1 \dots c_{k+1}}}{q_k} \geq a_{k+1} - c_{k+1} - 1 + q_{k-(d+1)} / q_k \geq 1 +
      C^{-(d+1)}.
    \end{equation*}
    The claim follows.

    \textbf{Case B.} Let us assume that $c_{k+1} = a_{k+1} - 1 > 0$ and $c_k = 0$ for infinitely many $k$. Suppose in
    addition that $a_{k+1} \geq 3$. By the final implication of \autoref{thm:main} (i), we have
    \begin{equation*}
      \dio{\infw{t}} - 1 \geq \limsup_{k\to\infty} \frac{\abs{u_{r_{k+1}}}}{q_k + \abs{\tau_{k+1}(y_i)} - \val{c_1 \dotsm c_{k+1}}}.
    \end{equation*}
    By applying \eqref{eq:helpful_1}, we see that
    \begin{align*}
      \frac{\abs{u_{r_{k+1}}}}{q_k + \abs{\tau_{k+1}(y)} - \val{c_1 \dotsm c_{k+1}}}
      &\geq \frac{a_{k+1} + q_{k-(d+1)}/q_k}{a_{k+1} - c_{k+1} + 2} \\
      &= \frac{a_{k+1} + q_{k-(d+1)}/q_k}{3} \\
      &\geq 1 + C^{-(d+1)}/3,
    \end{align*}
    so the claim follows.

    \textbf{Case C.} Assume that $c_{k+1} = a_{k+1} - 1 > 0$ and $c_k \neq 0$ for infinitely many $k$. Assume moreover
    that $a_{k+1} \geq 3$. The last implication of \autoref{thm:main} (iii) gives
    \begin{equation*}
      \dio{\infw{t}} - 1 \geq \limsup_{k\to\infty} \frac{\abs{u_{r_{k+1}}}}{2q_k + \abs{\tau_k(y_i)} - \val{c_1 \dotsm c_k}}.
    \end{equation*}
    From \eqref{eq:helpful_3}, we obtain that
    \begin{equation*}
      \frac{\abs{u_{r_{k+1}}}}{2q_k + \abs{\tau_k(y_i)} - \val{c_1 \dotsm c_k}}
      \geq \frac13 \left( a_{k+1} + q_{k-(d+1)} / q_k \right)
      \geq 1 + C^{-(d+1)}/3,
    \end{equation*}
    so the claim holds in this case as well.

    \textbf{Case D.} Suppose that $c_{k+1} = a_{k+1}$ for infinitely many $k$. Assume moreover that $a_{k+1} \geq 2$.
    By the second implication of \autoref{thm:main} (ii), we have
    \begin{equation*}
      \dio{\infw{t}} - 1 \geq \limsup_{k\to\infty} \frac{\abs{u_{r_{k+1}}}}{q_k + \abs{\tau_k(y_i)} - \val{c_1 \dotsm c_k}}.
    \end{equation*}
    Again, from \eqref{eq:helpful_3}, we see that
    \begin{equation*}
      \frac{\abs{u_{r_{k+1}}}}{q_k + \abs{\tau_k(y_i)} - \val{c_1 \dotsm c_k}}
      \geq \frac12 \left( a_{k+1} + q_{k-(d+1)} / q_k \right)
      \geq 1 + C^{-(d+1)}/2.
    \end{equation*}

    \textbf{Case E.} Suppose that $c_{k+1} = c_k = 0$ for infinitely many $k$. Suppose additionally that
    $a_{k+1} \geq 2$. Then the first implication of \autoref{thm:main} (iv.a) gives
    \begin{equation*}
      \dio{\infw{t}} - 1 \geq \limsup_{k\to\infty} \frac{\abs{u_{r_{k+1}}} - \val{c_1 \dotsm c_{k-1}}}{q_k}.
    \end{equation*}
    Now from \eqref{eq:helpful_2}, we find that
    \begin{equation*}
      \frac{\abs{u_{r_{k+1}}} - \val{c_1 \dotsm c_{k-1}}}{q_k} \geq a_{k+1} - 1 + C^{-(d+1)},
    \end{equation*}
    so the claim follows.
    
    \textbf{Case F.} Assume finally that $c_{k+1} = 0$ and $c_k \neq 0$ for infinitely many $k$. Suppose in addition
    that $a_{k+1} \geq 2$. We deduce from the first implication of \autoref{thm:main} (iv.c) that
    \begin{equation*}
      \dio{\infw{t}} - 1 \geq \limsup_{k\to\infty} \frac{\abs{u_{r_{k+1}}} - \val{c_1 \dotsm c_k}}{q_k}.
    \end{equation*}
    Thus the claim follows by an application of \eqref{eq:helpful_2} as in the Case E.

    The Cases A--F prove the claim when $\limsup_k a_k \geq 3$, so we may consider the special case $d = 2$. Next we
    handle the Cases B--D unconditionally. Consider the Case D first. The Ostrowski conditions imply that $c_k = 0$.
    Thus the first implication of
    \autoref{thm:main_sturmian} (ii) gives
    \begin{align*}
      \frac{q_{k+1} - 2 - \val{c_1 \dotsm c_k}}{q_{k-1}}
      &= \frac{a_{k+1}q_k + q_{k-1} - 2 - \val{c_1 \dotsm c_{k-1}}}{q_{k-1}} \\
      &\geq \frac{a_{k+1}q_k - 2}{q_{k-1}} \\
      &= a_{k+1} a_k + \frac{a_{k+1} q_{k-2} - 2}{q_{k-1}},
    \end{align*}
    so
    \begin{equation*}
      \dio{\infw{t}} - 1 \geq \limsup_{k\to\infty} \left( a_{k+1} a_k + \frac{a_{k+1} q_{k-2} - 2}{q_{k-1}} \right)
      \geq 1 + C^{-1}.
    \end{equation*}
    Suppose then that $c_{k+1} = a_{k+1} - 1 > 0$ and $c_k = 0$. Taking into account what we have already proved, we
    see that we may assume that $c_{k-1} < a_{k-1}$. Now the first implication of \autoref{thm:main_sturmian} (i)
    yields
    \begin{align*}
      \frac{q_{k+1} - 2 - \val{c_1 \dotsm c_{k+1}}}{q_k}
      &= \frac{q_k + q_{k-1} - 2 - \val{c_1 \dotsm c_{k-1}}}{q_k} \\
      &= 1 + \frac{q_{k-1} - 2 - c_{k-1}q_{k-2} - \val{c_1 \dotsm c_{k-2}}}{q_k} \\
      &\geq 1 + \frac{q_{k-1} - 2 - (c_{k-1} + 1)q_{k-2}}{q_k} \\
      &\geq 1 + \frac{q_{k-3} - 2}{q_k},
    \end{align*}
    so $\dio{\infw{t}} - 1 \geq 1 + C^{-3}$. Suppose finally that $c_{k+1} = a_{k+1} - 1 > 0$ and $c_k \neq 0$. We may
    again assume that $c_k < a_k$. From the second implication of \autoref{thm:main_sturmian} (iii), we have
    \begin{align*}
      \frac{2q_k + q_{k-1} - 2 - \val{c_1 \dotsm c_k}}{q_k + q_{k-1}}
      &= 1 + \frac{q_k - 2 - c_k q_{k-1} - \val{c_1 \dotsm c_{k-1}}}{q_k + q_{k-1}} \\
      &\geq 1 + \frac{q_k - (c_k + 1)q_{k-1} - 2}{q_k + q_{k-1}} \\
      &\geq 1 + \frac{q_{k-2} - 2}{q_k + q_{k-1}} \\
      &\geq 1 + \frac{q_{k-2} - 2}{q_{k+1}},
    \end{align*}
    so $\dio{\infw{t}} - 1 \geq 1 + C^{-3}$ in this case as well. As the cases A--D are now handled when $d = 2$, the
    claim is true or $c_k = 0$ for all $k$ large enough. In the latter case the word $\infw{t}$ is a shift of
    $\infw{c}_\Delta$, and the claim follows from Propositions \ref{prp:standard_dio} and
    \ref{prp:dio_shift_invariant}.
  \end{proof}

  For a fixed $d$, the conclusion of \autoref{thm:dio_lower_bound} can be improved. For example, it is shown in
  \cite[Thm.~4.3]{2019:a_new_complexity_function_repetitions_in_sturmian} that
  \begin{equation*}
    \dio{\infw{t}} \geq \frac53 + \frac{4\sqrt{10}}{15} \approx 2.5099
  \end{equation*}
  for a Sturmian word $\infw{t}$. We have not attempted such a study for $d > 2$. However, if $d$ is not fixed, it is
  unlikely that the lower bound $2$ can be improved. The main indication to this is the fact that the Diophantine
  exponent of the standard $d$-bonacci tends to $2$ as $d \to \infty$ (see the discussion after
  \autoref{prp:bonacci_dio}).

  \autoref{thm:dio_lower_bound} leaves open if $\dio{\infw{t}} > 2$ when $\limsup_k a_k \leq 2$ and $d > 2$. We show
  next that this is not necessarily true by providing explicit counterexamples. Before this, we prove the following
  lemmas that narrow down the exceptional directive words and intercepts even further.

  \begin{lemma}\label{lem:narrow}
    Let $\infw{t}$ be an episturmian word of period $d$ with directive word as in \eqref{eq:dw_multiplicative} and
    intercept $c_1 c_2 \dotsm$. If one of the conditions
    \begin{enumerate}[(i)]
      \item $0^d$ occurs infinitely many times in $c_1 c_2 \dotsm$,
      \item there exist infinitely many $k$ such that $c_i = a_i$ for all $i$ such that $k \leq i \leq k + d-2$
    \end{enumerate}
    is satisfied, then $\dio{\infw{t}} > 2$.
  \end{lemma}
  \begin{proof}
    We may assume that $\infw{t}$ has bounded partial quotients and let $C = 1 + \limsup_k a_k$. Suppose that $0^d$
    occurs infinitely many times in $c_1 c_2 \dotsm$. Let $k$ be such that $c_i = 0$ for all $i$ such that
    $k-(d-1) \leq i \leq k$. Suppose that $c_{k+1} = 0$. Then we apply the first implication of \autoref{thm:main}
    (iv.a), \eqref{eq:u_4}, \autoref{lem:ostrowski_ub}, and \eqref{eq:reg_q} as follows:
    \begin{align*}
      \dio{\infw{t}} - 1
      &\geq \limsup_{k\to\infty} \frac{\abs{u_{r_{k+1}}} - \val{c_1 \dotsm c_{k-d}}}{q_k} \\
      &\geq \limsup_{k\to\infty} \frac{q_k + \abs{s_{k-d}^{\s{a}-1} \dotsm s_0^{\s{a}}} - q_{k-d}}{q_k} \\
      &\geq 1 + \limsup_{k\to\infty} \frac{q_{k-d} + \abs{s_{k-2d}^{\s{a}-1} \dotsm s_0^{\s{a}}} - q_{k-d}}{q_k} \\
      &\geq 1 + C^{-(2d+1)}.
    \end{align*}
    Consider then the case $c_{k+1} \neq 0$. By inspecting the first
    implications of \autoref{thm:main} (i) and (ii), we find that
    \begin{equation*}
      \dio{\infw{t}} - 1
      \geq \limsup_{k\to\infty} \frac{\abs{u_{r_k+1}} - \val{c_1 \dotsm c_k}}{q_k},
    \end{equation*}
    so $\dio{\infw{t}} - 1 \geq 1 + C^{-(2d+1)}$ like above.

    Assume then that there exist infinitely many $k$ such that $c_i = a_i$ for all $i$ such that
    $k + 2 \leq i \leq k + d$. We must have $c_{k+1} = 0$ by the Ostrowski conditions. We infer from \autoref{thm:main}
    (v) that
    \begin{equation*}
      \dio{\infw{t}} - 1 \geq \limsup_{k\to\infty} \frac{\abs{u_{r_{k+1}}}}{q_k} \geq \limsup_{k\to\infty} \left( a_{k+1} + \frac{\abs{s_{k-d}^{\s{a}-1} \dotsm s_0^{\s{a}}}}{q_k} \right) \geq 1 + C^{-(d+1)} \qedhere
    \end{equation*}
  \end{proof}

  \begin{lemma}\label{lem:3_ls1}
    Let $\infw{t}$ be a regular episturmian word with directive word as in \eqref{eq:dw_multiplicative} with $d = 3$.
    If $\limsup_k a_k = 1$, then $\dio{\infw{t}} > 2$.
  \end{lemma}
  \begin{proof}
    Say $\limsup_k a_k = 1$ and that $\infw{t}$ has intercept $c_1 c_2 \dotsm$. By \autoref{lem:narrow}, the claim
    follows if the intercept contains $0^3$ infinitely many times or $11$ infinitely many times. Therefore we may
    assume that $c_1 c_2 \dotsm$ is eventually a product of the words $01$ and $001$.

    Assume that $001$ occurs infinitely often in the intercept. Let $k$ be large enough such that we are guaranteed
    that $a_{k+1} = a_k = a_{k-1} = 1$, $c_{k+1} = 1$, and $c_k = c_{k-1} = 0$. By the first implication of
    \autoref{thm:main} (ii), we have
    \begin{equation*}
      \dio{\infw{t}} - 1 \geq \limsup_{k\to\infty} \frac{\abs{u_{r_k+1}} - \val{c_1 \dotsm c_k}}{\abs{\tau_k(y_i)}}.
    \end{equation*}
    By \autoref{lem:regular_help}, we have $\abs{\tau_k(y_i)} \leq q_{k-1} + q_{k-2}$. Therefore
    \begin{align*}
      \frac{\abs{u_{r_k+1}} - \val{c_1 \dotsm c_k}}{\abs{\tau_k(y_i)}}
      &\geq \frac{\abs{s_{k-1} s_{k-2} s_{k-3}^{\s{a}} \dotsm s_0^{\s{a}}} - \val{c_1 \dotsm c_{k-2}}}{q_{k-1} + q_{k-2}} \\
      &\geq \frac{q_{k-1} + \abs{s_{k-3}^{\s{a}} \dotsm s_0^{\s{a}}}}{q_{k-1} + q_{k-2}} \\
      &= \frac{q_{k-1} + q_{k-2} + \abs{s_{k-5}^{\s{a}-1} s_{k-6}^{\s{a}} \dotsm s_0^{\s{a}}}}{q_{k-1} + q_{k-2}} \\
      &\geq 1 + \frac{q_{k-6}}{q_k} \\
      &\geq 1 + C^{-6}
    \end{align*}
    where $C = 1 + \limsup_k a_k = 2$. Therefore $\dio{\infw{t}} > 2$.

    The only case left is the case where $c_1 c_2 \dotsm$ has suffix $(01)^\omega$. Say $k$ is large enough and
    $c_{k+1} = 0$ and $c_{k+2} = c_k = 1$. Since $c_{k+2} = 1$, we have $y_i \neq x_{k+2}$, and so
    $\abs{\tau_{k+1}(y_i)} = q_k + q_{k-1}$. From \autoref{thm:main} (iv.b), we have
    \begin{align*}
      \dio{\infw{t}} - 1
      &\geq \limsup_{k\to\infty} \frac{\abs{u_{r_{k+1}}}}{\abs{\tau_{k+1}(y_i)} - \val{c_1 \dotsm c_k}} \\
      &\geq \limsup_{k\to\infty} \frac{\abs{s_{k-1}^{\s{a}} \dotsm s_0^{\s{a}}}}{q_k + q_{k-1} - q_{k-1}} \\
      &\geq 1 + \limsup_{k\to\infty} \frac{\abs{s_{k-4}^{\s{a}} \dots s_0^{\s{a}}}}{q_k} \\
      &\geq 1 + C^{-4}. \qedhere
    \end{align*}
  \end{proof}

  By \autoref{lem:3_ls1}, the conclusion of \autoref{thm:dio_lower_bound} can fail only if $\limsup_k a_k = 2$ when
  $d = 3$. The next proposition shows that this may happen. For the proof, recall that the Stolz-Ces\`{a}ro Theorem
  states that
  \begin{equation*}
    \lim_{k\to\infty} \frac{a_k}{b_k} = \lim_{k\to\infty} \frac{a_k - a_{k-1}}{b_k - b_{k-1}}
  \end{equation*}
  whenever the right side limit exists and the sequence $(b_k)$ is strictly monotone.

  \begin{proposition}\label{prp:3_ce_dio}
    Let $\infw{t}$ be the episturmian word with directive word $(001122)^\omega$ having intercept $1^\omega$. Then
    \begin{equation*}
      \dio{\infw{t}} = 1 + \frac12(\beta - 1) \approx 1.9156
    \end{equation*}
    where $\beta$, approximately $2.8312$, is the real root of the polynomial $x^3 - 2x^2 - 2x - 1$.
  \end{proposition}
  \begin{proof}
    Suppose that $k$ is such that $k \geq 4$. It follows from \autoref{thm:main} (iii) that $\dio{\infw{t}} - 1$ equals
    the largest of the limits of the ratios
    \begin{equation*}
      \frac{\abs{u_{r_k+1}}}{q_k + \abs{\tau_k(y_i)} - \val{c_1 \dotsm c_k}},
      \frac{\abs{u_{r_k+1}} + q_k - \val{c_1 \dotsm c_k}}{q_k + \abs{\tau_k(y_i)}},
      \frac{\abs{u_{r_{k+1}}}}{2q_k + \abs{\tau_k(y_i)} - \val{c_1 \dotsm c_k}}
    \end{equation*}
    as $k \to \infty$. Here $\Delta' = T^{r_{k+1}}(\Delta) = T^{2(k+1)}(\Delta)$, so the words $\Delta$ and $\Delta'$
    are isomorphic. Therefore the Ostrowski numeration systems associated with $\Delta$ and $\Delta'$ are the same.
    It follows that $\val[\Delta']{c_{k+2} \dotsm c_{k+d-1}} = \val[\Delta]{1} = 1$ meaning that $y_i$ is the second
    letter of the word $\infw{c}_{\Delta'}$, that is, $y_i = x_{k+2}$. Thus $\tau_k(y_i) = s_{k-1}^2 s_{k-2}$. The
    numbers $q_i$ satisfy the linear recurrence $q_i = 2q_{i-1} + 2q_{i-2} + q_{i-3}$, so it follows from the theory of
    linear recurrences that $\smash[t]{q_{k+\ell}/q_k \to \beta^\ell}$ as $k \to \infty$.

    Let us compute the last limit. From \eqref{eq:u_2}, we find that
    $\abs{u_{r_{k+1}}} - \abs{u_{r_k}} = q_k + q_{k-1}$. In addition, we have
    $2q_k + \abs{\tau_k(y_i)} = q_{k+1}$ and $\val{1^k} - \val{1^{k-1}} = q_{k-1}$. Therefore the Stolz-Ces\`{a}ro
    Theorem implies that
    \begin{align*}
      \lim_{k\to\infty} \frac{\abs{u_{r_{k+1}}}}{2q_k + \abs{\tau_k(y_i)} - \val{c_1 \dotsm c_k}}
      &= \lim_{k\to\infty} \frac{q_k + q_{k-1}}{q_{k+1} - q_k - q_{k-1}} \\
      &= \lim_{k\to\infty} \frac{q_k + q_{k-1}}{q_k + q_{k-1} + q_{k-2}} \\
      &= \frac{\beta^2 + \beta}{\beta^2 + \beta + 1} \\
      &= \frac12(\beta - 1) \\
      &\approx 0.9156.
    \end{align*}
    Therefore $\dio{\infw{t}}$ is at least as large as claimed. It thus suffices to show that the other two limits do
    not exceed this value. For the first ratio, we find like above that
    \begin{equation*}
      \lim_{k\to\infty} \frac{\abs{u_{r_k+1}}}{q_k + \abs{\tau_k(y_i)} - \val{c_1 \dotsm c_k}}
      = \lim_{k\to\infty} \frac{2q_{k-1}}{2q_{k-1} + q_{k-2}}
      = \frac{2\beta}{2\beta + 1}
      \approx 0.8499
    \end{equation*}
    and, for the second ratio, we have
    \begin{equation*}
      \lim_{k\to\infty} \frac{\abs{u_{r_k+1}} + q_k - \val{c_1 \dotsm c_k}}{q_k + \abs{\tau_k(y_i)}}
      = \lim_{k\to\infty} \frac{q_k}{3q_{k-1} + q_{k-2}}
      = \frac{\beta}{3\beta + 1}
      \approx 0.8443. \qedhere
    \end{equation*}
  \end{proof}

  Notice that \autoref{prp:3_ce_dio} shows that the Diophantine exponent can be less than that of the corresponding
  standard word. The next proposition demonstrates that \autoref{lem:3_ls1} does not generalize to $d > 3$. Therefore
  the assumptions of \autoref{thm:dio_lower_bound} are necessary.

  \begin{proposition}\label{prp:4-bonacci_dio}
    Let $\infw{t}$ be the episturmian word with directive word $(0123)^\omega$ having intercept $(001)^\omega$,
    $(010)^\omega$, or $(100)^\omega$. Then
    \begin{equation*}
      \dio{\infw{t}} = 1 + \frac{1}{27}(-7 \zeta_4^3 + 15 \zeta_4^2 + 13 \zeta_4 - 4) \approx 1.9873
    \end{equation*}
    where $\zeta_4$, approximately $1.9276$, is the positive real root of the polynomial $x^4 - x^3 - x^2 - x - 1$.
  \end{proposition}
  \begin{proof}
    Assume that the intercept $c_1 c_2 \dotsm$ equals $(001)^\omega$. Suppose that $k$ is such that $k \geq 3$ and
    $c_{k+1} = 1$. It follows from \autoref{thm:main} (ii) that $\dio{\infw{t}} - 1$ is it least as large as the limits
    of
    \begin{equation*}
      \frac{\abs{u_{r_{k+1}}} - \val{c_1 \dotsm c_k}}{\abs{\tau_k(y_i)}} \quad \text{and} \quad \frac{\abs{u_{r_{k+1}}}}{q_k + \abs{\tau_k(y_i)} - \val{c_1 \dotsm c_k}}
    \end{equation*}
    as $k \to \infty$ (along appropriate subsequences). Now $c_{k+2} \dotsm c_{k+d-2} = 00$, so $y_i$ is the first letter of the epistandard word with
    intercept $x_{k+2} x_{k+3} \dotsm$, that is, $y_i = x_{k+2}$ (notice that the word $\Delta'$ of \autoref{thm:main}
    is isomorphic to the directive word $(0123)^\omega$, so the numeration systems associated to both directive words
    are the same). It follows that $\tau_k(y_i) = s_{k-1} s_{k-2} s_{k-3}$. The numbers $q_i$ satisfy the linear
    recurrence $q_i = q_{i-1} + q_{i-2} + q_{i-3} + q_{i-4}$, so it follows from the theory of linear recurrences that
    $q_{k+\ell}/q_k \to \zeta_4^\ell$ as $k \to \infty$.

    From \eqref{eq:u_4}, we find that $\abs{u_{r_{k+1}}} - \abs{u_{r_{k-2}}} = q_k - q_{k-4}$. In addition,
    $\val{c_1 \dotsm c_k} - \val{c_1 \dotsm c_{k-3}} = q_{k-3}$. Using the Stolz-Ces\`{a}ro Theorem, the limit of the
    first ratio is found as follows:
    \begin{align*}
      \lim_{k\to\infty} \frac{\abs{u_{r_{k+1}}} - \val{c_1 \dotsm c_k}}{\abs{\tau_k(y_i)}}
      &= \lim_{k\to\infty} \frac{\abs{u_{r_{k+1}}} - \abs{u_{r_{k-2}}} - (\val{c_1 \dotsm c_k} - \val{c_1 \dotsm c_{k-3}})}{\abs{\tau_k(y_i)} - \abs{\tau_{k-3}(y_i)}} \\
      &= \lim_{k\to\infty} \frac{q_{k-1} + q_{k-2}}{q_{k-1} + q_{k-2} + q_{k-7}} \\
      &= \frac{\zeta_4^6 + \zeta_4^5}{\zeta_4^6 + \zeta_4^5 + 1} \\
      &= \frac{1}{27}(-7 \zeta_4^3 + 15 \zeta_4^2 + 13 \zeta_4 - 4) \\
      &\approx 0.9873.
    \end{align*}
    Therefore $\dio{\infw{t}}$ is at least as large as claimed. It thus suffices to show that the analogous limits do
    not exceed this value in the other cases. For the latter ratio, we find that
    \begin{align*}
      \lim_{k\to\infty} \frac{\abs{u_{r_{k+1}}}}{q_k + \abs{\tau_k(y_i)} - \val{c_1 \dotsm c_k}}
      = \lim_{k\to\infty} \frac{q_k - q_{k-4}}{q_k + q_{k-1} + q_{k-2} - 2q_{k-3} + q_{k-7}}
      \approx 0.6107.
    \end{align*}

    Assume then that $c_{k+1} = 0$ and $c_k = 1$. By \autoref{thm:main} (iv.b), we need to find the limit of
    \begin{equation*}
      \frac{\abs{u_{r_{k+1}}}}{\abs{\tau_{k+1}(y_i)} - \val{c_1 \dotsm c_k}}.
    \end{equation*}
    Now $c_{k+2} \dotsm c_{k+d-1} = 01$ and $\val{c_{k+2} \dotsm c_{k+d-1}} = 2$, so again $y_i = x_{k+2}$ and
    $\abs{\tau_{k+1}(y_i)} = q_{k+1}$. Moreover, we have $\val{c_1 \dotsm c_k} - \val{c_1 \dotsm c_{k-2}} = q_{k-1}$.
    Hence
    \begin{equation*}
      \lim_{k\to\infty} \frac{\abs{u_{r_{k+1}}}}{\abs{\tau_{k+1}(y_i)} - \val{c_1 \dotsm c_k}}
      = \lim_{k\to\infty} \frac{q_k - q_{k-4}}{q_{k+1} - q_{k-2} - q_{k-1}}
      \approx 0.8139.
    \end{equation*}

    Suppose finally that $c_{k+1} = c_k = 0$. By \autoref{thm:main} (iv.a), we need to consider the ratios
    \begin{equation*}
      \frac{\abs{u_{r_{k+1}}} - \val{c_1 \dotsm c_{k-1}}}{q_k} \quad \text{and} \quad
      \frac{\abs{u_{r_{k+1}}}}{\abs{\tau_{k+1}(y_i)} - \val{c_1 \dotsm c_{k-1}}}.
    \end{equation*}
    This time around $c_{k+2} \dotsm c_{k+d-1} = 10$, $\val{c_{k+2} \dotsm c_{k+d-1}} = 1$, $y_i = x_{k+3}$,
    $\tau_{k+1}(y_i) = s_k s_{k-1} s_{k-2}$, and $\val{c_1 \dotsm c_k} - \val{c_1 \dotsm c_{k-2}} = q_{k-2}$.
    Proceeding as above, it is straightforward to show that the limits are approximately $0.7653$ and $0.7309$. This
    proves the claim. It is straightforward to check that the intercepts $(010)^\omega$ and $(100)^\omega$ lead to
    exactly the same result.
  \end{proof}

  The intercepts $(001)^\omega$, $(010)^\omega$, and $(100)^\omega$ are not the only interesting ones for words in the
  $4$-bonacci subshift. It can be computed (like in the proof of \autoref{prp:4-bonacci_dio}) that the words with
  intercepts $(01)^\omega$ or $(10)^\omega$ have Diophantine exponent $2$. The word with intercept $(011)^\omega$ and
  the words with conjugate intercepts $(110)^\omega$ and $(101)^\omega$ have Diophantine exponent $1 + \zeta_4^2 -
  \zeta_4 \approx 2.7879$. The intercepts $(0001)^\omega$ and $(0011)^\omega$ and their conjugate intercepts give
  Diophantine exponent $\dio{\infw{c}_{(0123)^\omega}} \approx 2.0781$. We have not attempted to compute Diophantine
  exponents for aperiodic intercepts. Computer experiments suggest that the Diophantine exponent of the word in the
  $5$-bonacci subshift with intercept $(001)^\omega$ is approximately $1.9148$ and approximately $1.8535$ for the
  intercept $(01)^\omega$. We believe that similar results are obtained for $d$-bonacci words for $d > 4$ and for words
  that satisfy $\limsup_k a_k = 1$.

  Since $\ice{\infw{x}} \leq \dio{\infw{x}}$ for any infinite word $\infw{x}$, \autoref{prp:3_ce_dio} has the following
  remarkable consequence.

  \begin{corollary}
    There exists an episturmian word over a $3$-letter alphabet having only finitely many square prefixes.
  \end{corollary}

  This is indeed unexpected since every Sturmian word and every regular epistandard word has arbitrarily long square
  prefixes \cite[Lemma~6.5]{2007:powers_in_a_class_of_A_strict_standard_episturmian}. We expect that a regular
  episturmian word has infinitely many square prefixes when $\limsup_{k\to\infty} a_k \geq 3$, but we have not
  attempted to prove this. We also expect that every infinite word in the Tribonacci subshift has arbitrarily long
  square prefixes.

  \subsection{Diophantine Exponents of $d$-bonacci Words}
  Here we prove more explicit results on the Diophantine exponents of $d$-bonacci words.

  \begin{definition}
    The \emph{$d$-bonacci constant} $\zeta_d$ is the positive real root of the polynomial
    $x^d - x^{d-1} - \dotsm - x - 1$.
  \end{definition}

  \begin{proposition}\label{prp:bonacci_dio}
    For the $d$-bonacci word $\infw{f}_d$ with directive word $(01 \dotsm (d-1))^\omega$ and intercept $0^\omega$, we
    have $\dio{\infw{f}_d} = \ice{\infw{f}_d} = \ind{\infw{f}_d} - 1 = 1 + 1/(\zeta_d - 1)$.
  \end{proposition}
  \begin{proof}
    We deduce from \autoref{thm:main} (iv.a) and the Stolz-Ces\`{a}ro Theorem that
    \begin{equation*}
      \dio{\infw{f}_d} - 1
      = \lim_{k\to\infty} \frac{\abs{u_{r_k+1}}}{q_k}
      = \lim_{k\to\infty} \frac{q_k}{q_k - q_{k-1}}
      = \frac{1}{\zeta_d - 1}. \qedhere
    \end{equation*}
  \end{proof}

  In particular, we have $\dio{\infw{f}_2} \approx 2.6180$, $\dio{\infw{f}_3} \approx 2.1915$, and
  $\dio{\infw{f}_4} \approx 2.0781$. Moreover, we have $\dio{\infw{f}_d} \to 2$ as $d \to \infty$ since it is
  well-known that $\zeta_d \to 2$ as $d \to \infty$ (see, e.g.,
  \cite{2014:three_series_for_the_generalized_golden_mean}).

  \begin{lemma}\label{lem:bonacci_conditions}
    Let $\infw{t}$ be a word in the $d$-bonacci subshift with intercept $c_1 c_2 \dotsm$ that is not in the
    shift orbit of $\infw{f}_d$. If $0^d$ occurs infinitely many times in $c_1 c_2 \dotsm$, then
    \begin{equation*}
      \dio{\infw{t}} \geq 2 + \frac{1}{\zeta_d^d - \zeta_d} > \dio{\infw{f}_d}.
    \end{equation*}
    If $0^\ell$ occurs in $c_1 c_2 \dotsm$ for arbitrarily large $\ell$, then
    \begin{equation*}
      \dio{\infw{t}} \geq 1 + \frac{\zeta_d^d}{(\zeta_d - 1)(\zeta_d^d - 1)}.
    \end{equation*}
  \end{lemma}
  \begin{proof}
    Since $\infw{t}$ is not in the shift orbit of $\infw{f}_d$, the intercept does not end with $0^\omega$. Therefore
    it contains $0^\ell 1$ infinitely many times for some $\ell$. Let $k$ be such that $c_{k+1} = 1$ and
    $c_k = \ldots = c_{k-(\ell-1)} = 0$. Using \autoref{lem:regular_tau_images}, we see that
    \begin{align}\label{eq:foo_1}
      \frac{\abs{u_{r_k+1}} - \val{c_1 \dotsm c_k}}{\abs{\tau_k(y_i)}}
      &\geq \frac{\abs{s_{k-1} \dotsm s_0} - \val{c_1 \dotsm c_{k-\ell}}}{\abs{s_{k-1} \dotsm s_{k-(d-1)}}} \nonumber \\
      &= \frac{\abs{s_{k-1} \dotsm s_0} - \val{c_1 \dotsm c_{k-\ell}}}{q_k - q_{k-d}}.
    \end{align}
    Now
    $-\val{c_1 \dotsm c_{k-\ell}}/(q_k - q_{k-d}) \geq -q_{k-\ell}/(q_k - q_{k-d}) \xrightarrow{k\to\infty} -\zeta_d^{-\ell}/(1 - \zeta_d^{-d})$,
    so if $\ell$ can be taken arbitrarily large, we have from \autoref{thm:main} (ii) that
    \begin{equation*}
      \dio{\infw{t}} - 1 
      \geq \lim_{k\to\infty} \frac{\abs{s_{k-1} \dotsm s_0}}{q_k - q_{k-d}}
      = \lim_{k\to\infty} \frac{q_k + \abs{u_{r_{k-d+1}}}}{q_k - q_{k-d}}
      = \frac{\dio{\infw{f}_d} - 1}{1 - \zeta_d^{-d}}
    \end{equation*}
    where the last equality follows from the proof of \autoref{prp:standard_dio}. The latter claim now follows from
    \autoref{prp:bonacci_dio} and simplification. If $\ell \geq d$, then we deduce from \eqref{eq:foo_1} that
    \begin{equation*}
      \frac{\abs{u_{r_k+1}} - \val{c_1 \dotsm c_k}}{\abs{\tau_k(y_i)}}
      \geq \frac{\abs{s_{k-1} \dotsm s_{k-(d-1)}} + \abs{s_{k-(d+1)} \dotsm s_0}}{q_k - q_{k-d}}
      = 1 + \frac{\abs{u_{r_{k-d+1}}}}{q_k - q_{k-d}},
    \end{equation*}
    so, like above,
    \begin{equation*}
      \dio{\infw{t}} - 1 
      \geq 1 + \frac{\dio{\infw{f}_d} - 2}{1 - \zeta_d^{-d}} = 1 + \frac{1}{\zeta_d^d - \zeta_d}.
    \end{equation*}
    It is easy to check that this lower bound is larger than $\dio{\infw{f}_d}$. Thus the first claim is proved.
  \end{proof}

  We believe that there always exists a word in the subshift of the $d$-bonacci word such that
  $\dio{\infw{t}} = 2 + 1/(\zeta_d^d - \zeta_d)$, but we have not attempted to prove this rigorously.

  \begin{corollary}
    Let $\infw{t}$ be a word in the Fibonacci subshift with intercept $c_1 c_2 \dotsm$. If $\infw{t}$ is in the shift
    orbit of the standard word $\infw{c}_\Delta$, then $\dio{\infw{t}} = \ice{\infw{t}} = 1 + \zeta_2$. Otherwise
    $\dio{\infw{t}}, \ice{\infw{t}} \in [3, 2+\zeta_2]$. Moreover, if $c_1 c_2 \dotsm$ contains $10^\ell 1$ for
    arbitrarily large $\ell$, then $\dio{\infw{t}} = \ice{\infw{t}} = 2+\zeta_2 = \ind{\infw{t}}$.
  \end{corollary}
  \begin{proof}
    If $0^2$ occurs infinitely many times in $c_1 c_2 \dotsm$, then \autoref{lem:bonacci_conditions} implies that
    $\dio{\infw{t}} \geq 3$. If $0^2$ does not occur in $c_1 c_2 \dotsm$ infinitely many times, then the intercept ends
    with $1^\omega$ or $(01)^\omega$. The former case is impossible since the intercept satisfies the Ostrowski
    conditions by \autoref{lem:regular_ostrowski_conditions}. It follows from \autoref{lem:left_shift_intercept} that
    \begin{equation*}
      \infw{t} = T^{c_1} L_{x_1}^{a_1} \circ \dotsm \circ T^{c_k} L_{x_k}^{a_k}(y \infw{c}_{\Delta'}) = T^{c_1} L_{x_1}^{a_1}
      \circ \dotsm \circ T^{c_k} L_{x_k}^{a_k}(y) \tau_k(\infw{c}_{\Delta'})
    \end{equation*}
    for some integer $k$, letter $y$, and $\Delta' = T^k(\Delta)$. Clearly $\tau_k(\infw{c}_{\Delta'})$ is an
    episturmian word with directive word $(01)^\omega$ and intercept $0^\omega$, so
    $\tau_k(\infw{c}_{\Delta'}) = \infw{f}_2$. Therefore $\infw{t} = w \infw{f}_2$ for some finite word $w$. It follows
    from Propositions \ref{prp:dio_shift_invariant} and \ref{prp:bonacci_dio} that
    $\dio{\infw{t}} = \dio{\infw{f}_2} = 1 + \zeta_2$. The last claim follows from \autoref{lem:bonacci_conditions}
    because $1 + \zeta_2^2/((\zeta_2 - 1)(\zeta_2^2 - 1)) = \ind{\infw{f}_2} = 2 + \zeta_2$. The claim for
    $\ice{\infw{t}}$ was established in the proof of \cite[Prop.~4.3]{2006:initial_powers_of_sturmian_sequences}.
  \end{proof}

  We can prove a result like the preceding corollary for the Tribonacci subshift, but this result does not anymore
  generalize to $d > 3$ as indicated by \autoref{prp:4-bonacci_dio}.

  \begin{proposition}\label{prp:tribonacci_dio}
    If $\infw{t}$ is a word in the Tribonacci subshift, then
    $\dio{\infw{f}_3} \leq \dio{\infw{t}} \leq \ind{\infw{f}_3}$.
  \end{proposition}
  \begin{proof}
    Say $\infw{t}$ has intercept $c_1 c_2 \dotsm$. If the intercept has suffix $0^\omega$, then
    $\dio{\infw{t}} = \dio{\infw{f}_3}$, so we may assume that $1$ occurs infinitely often in $c_1 c_2 \dotsm$. By
    \autoref{lem:bonacci_conditions}, the claim is clear if the intercept contains $0^3$ infinitely many times.

    Suppose that $0011$ occurs infinitely often in $c_1 c_2 \dotsm$. Then there exists infinitely many $k$ such that
    $c_{k+2} = c_{k+1} = 1$ and $c_k = c_{k-1} = 0$. Then the first implication of \autoref{thm:main} (ii) yields
    \begin{equation*}
      \dio{\infw{t}} - 1 \geq \limsup_{k\to\infty} \frac{\abs{u_{r_k+1}} - \val{c_1 \dotsm c_k}}{\abs{\tau_k(y_i)}}.
    \end{equation*}
    Since $c_{k+2} = 1$, we have $y_i \neq x_{k+2}$, so $\abs{\tau_k(y_i)} = q_{k-1}$ by \autoref{lem:regular_help}.
    Therefore
    \begin{equation*}
      \frac{\abs{u_{r_k+1}} - \val{c_1 \dotsm c_k}}{\abs{\tau_k(y_i)}} = \frac{\abs{u_{r_k+1}} - \val{c_1 \dotsm c_{k-2}}}{q_{k-1}}
      \geq 1 + \frac{q_{k-3}}{q_{k-1}}
      \xrightarrow{k\to\infty} 1 + \zeta_3^{-2}.
    \end{equation*}
    Hence $\dio{\infw{t}} \geq 2 + \zeta_3^{-2} > \dio{\infw{f}_3}$. We may thus assume that eventually all occurrences
    of $11$ in $c_1 c_2 \dotsm$ are preceded by $10$. Assume that $1101$ occurs infinitely many times in
    $c_1 c_2 \dotsm$. Letting $k + 1$ to correspond to the third letter of this pattern, we have from
    \autoref{thm:main} (iv.b) that
    \begin{equation*}
      \dio{\infw{t}} - 1 \geq \limsup_{k\to\infty} \frac{\abs{u_{r_{k+1}}}}{\abs{\tau_{k+1}(y_i)} - \val{c_1 \dotsm c_k}}.
    \end{equation*}
    Again $c_{k+2} = 1$, so $\abs{\tau_{k+1}(y_i)} = q_k + q_{k-1}$. Using the fact that $1101$ is preceded by $10$, we
    have
    \begingroup
    \allowdisplaybreaks
    \begin{align*}
      \frac{\abs{u_{r_{k+1}}}}{\abs{\tau_{k+1}(y_i)} - \val{c_1 \dotsm c_k}}
      &= \frac{\abs{u_{r_{k+1}}}}{q_k + q_{k-1} - q_{k-1} - q_{k-2} - q_{k-4} - \val{c_1 \dotsm c_{k-4}}} \\
      &\geq \frac{\abs{u_{r_{k+1}}}}{q_k - q_{k-2} - q_{k-4}} \\
      &\geq \frac{q_k}{q_k - q_{k-2} - q_{k-4}} \\
      &\xrightarrow{k\to\infty} \frac{1}{1 - \zeta_3^{-2} - \zeta_3^{-4}} \\
      &\geq 1.6206.
    \end{align*}
    \endgroup
    Hence $\dio{\infw{t}} \geq 2.6206 > \dio{\infw{f}_3}$. We may thus assume that $11$ is eventually always followed
    by $00$. Suppose then that $100101$ occurs infinitely many times in the intercept. Letting $k + 1$ to correspond to
    the fifth letter of this pattern, we have from \autoref{thm:main} (iv.b) that
    \begin{equation*}
      \dio{\infw{t}} - 1 \geq \limsup_{k\to\infty} \frac{\abs{u_{r_{k+1}}}}{\abs{\tau_{k+1}(y_i)} - \val{c_1 \dotsm c_k}}.
    \end{equation*}
    Now $\abs{\tau_{k+1}(y_i)} = q_k + q_{k-1}$ and so
    \begin{align*}
      \frac{\abs{u_{r_{k+1}}}}{\abs{\tau_{k+1}(y_i)} - \val{c_1 \dotsm c_k}}
      \geq \frac{\abs{u_{r_{k+1}}}}{q_k + q_{k-1} - q_{k-1} - q_{k-4}}
      \geq \frac{q_k + q_{k-4}}{q_k - q_{k-4}}
      \xrightarrow{k\to\infty} \dio{\infw{f}_3} - 1,
    \end{align*}
    meaning that $\dio{\infw{t}} \geq \dio{\infw{f}_3}$. We may thus assume that $100$ is eventually always followed by
    $100$. We may now prove that the intercept contains only finitely many occurrences of $11$. Indeed, the preceding
    assumptions imply that late enough occurrences of $11$ must be followed by $00$ and each late enough occurrence of
    $100$ must be followed by $100$. Thus $c_1 c_2 \dotsm$ has suffix $(100)^\omega$ and $11$ occurs only finitely many
    times.

    We have now argued that the claim holds or the intercept is eventually a product of the words $10$ and $100$.
    Moreover, the pattern $100101$ does not eventually occur, so the intercept has suffix $(10)^\omega$ or
    $(100)^\omega$. Say it has suffix $(10)^\omega$. Let $k$ be large enough such that $c_{k+1} = c_{k-1} = 0$ and
    $c_{k+2} = c_k = c_{k-2} = 1$. Again $\abs{\tau_{k+1}(y_i)} = q_k + q_{k-1}$, and
    \begin{align*}
      \frac{\abs{u_{r_{k+1}}}}{\abs{\tau_{k+1}(y_i)} - \val{c_1 \dotsm c_k}}
      &= \frac{\abs{u_{r_{k+1}}}}{q_k + q_{k-1} - q_{k-1} - q_{k-3} - \val{c_1 \dotsm c_{k-3}}} \\
      &\geq \frac{\abs{s_{k-1} \dotsm s_0}}{q_{k-1} + q_{k-2}} \\
      &\geq \frac{q_{k-1} + q_{k-2} + q_{k-3}}{q_{k-1} + q_{k-2}} \\
      &\xrightarrow{k\to\infty} \dio{\infw{f}_3} - 1,
    \end{align*}
    so $\dio{\infw{t}} \geq \dio{\infw{f}_3}$ by \autoref{thm:main} (iv.b).

    We can now assume that the intercept has suffix $(100)^\omega$. Since the Diophantine exponent is shift-invariant,
    we may focus on the intercept $(100)^\omega$ and its conjugate intercepts $(010)^\omega$, and $(001)^\omega$.

    Let us first consider the intercept $(100)^\omega$. Let $k$ be such that $c_{k+1} = 1$. Here
    $\abs{\tau_k(y_i)} = q_{k-1} + q_{k-2}$, $\abs{u_{r_k+1}} - \abs{u_{r_{k-3}+1}} = q_k$, and
    $\val{c_1 \dotsm c_k} - \val{c_1 \dotsm c_{k-3}} = q_{k-3}$, so we find with the help of the Stolz-Ces\`{a}ro
    Theorem that (along an appropriate subsequence)
    \begin{equation*}
      \lim_{k\to\infty} \frac{\abs{u_{r_k+1}} - \val{c_1 \dotsm c_k}}{\abs{\tau_k(y_i)}}
      = \lim_{k\to\infty} \frac{q_{k-1} + q_{k-2}}{q_{k-1} + q_{k-3}}
      = \frac{\zeta_3^2 + \zeta_3}{\zeta_3^2 + 1}
      = \frac{1}{\zeta_3 - 1},
    \end{equation*}
    so $\dio{\infw{t}} \geq \dio{\infw{f}_3}$ by \autoref{thm:main} (ii). In fact, it can be showed that
    $\dio{\infw{t}} = \dio{\infw{f}_3}$. The remaining cases lead to the same result as is straightforward to show.
  \end{proof}

  \section{Irrationality Exponents}\label{sec:ie}
  The results of the previous section on Diophantine exponents allows us to obtain novel results on irrationality
  exponents of numbers whose fractional parts are regular episturmian.

  \begin{definition}
    The \emph{irrationality exponent} $\mu(\xi)$ of a real number $\xi$ is the supremum of the real numbers $\rho$ such
    that the inequality
    \begin{equation*}
      \Abs{\xi - \frac{p}{q}} < \frac{1}{q^\rho}
    \end{equation*}
    has infinitely many rational solutions $p/q$. If $\mu(\xi) = \infty$, then we say that $\xi$ is a \emph{Liouville
    number}.
  \end{definition}

  Recall that $\mu(\xi) \geq 2$ whenever $\xi$ is irrational and that $\mu(\xi) = 2$ for almost all real numbers $\xi$
  with respect to the Lebesgue measure. Roth's theorem states that $\mu(\xi) = 2$ when $\xi$ is algebraic.

  \begin{definition}
    Let $b$ be an integer such that $b \geq 2$ and $\infw{x}$ be an infinite word over the alphabet
    $\{0, 1, \ldots, b-1\}$ such that $\infw{x} = x_1 x_2 \dotsm$. Then $\xi_{\infw{x},b}$ is the real number
    $\sum_{k \geq 1} x_k/b^k$.
  \end{definition}

  The following result is the key result linking combinatorial properties of $\infw{x}$ to the Diophantine properties
  of $\xi_{\infw{x},b}$. For its proof, see, e.g.,
  \cite[Sect.~3]{2010:on_the_expansion_of_some_exponential_periods_in_an_integer} and
  \cite[Thm.~2.1]{2011:nombres_reels_de_complexite_sous-lineaire}. The function $p$ of the statement is the factor
  complexity function that counts the number of factors of length $n$, that is,
  $p(\infw{x}, n) = \abs{\Lang[n]{\infw{x}}}$ for all $n$.

  \begin{proposition}\label{prp:ie_bounds}
    Let $b$ be an integer such that $b \geq 2$ and $\infw{x}$ be an aperiodic infinite word. Then
    $\mu(\xi_{\infw{x},b}) \geq \dio{\infw{x}}$. If, moreover, there exists an integer $K$ such that
    $p(\infw{x}, n) \leq Kn$ for all $n$ large enough, then $\mu(\xi_{\infw{x},b}) \leq (2K+1)^3 (\dio{\infw{x}} + 1)$.
  \end{proposition}

  Showing that $\mu(\xi_{\infw{x},b}) > 2$ implies that $\xi_{\infw{x},b}$ is transcendental by Roth's theorem, so
  computing Diophantine exponents of infinite words can be used to show transcendence results. While the following
  results allow us to conclude the transcendence of certain numbers, these facts are not new. One of the main results
  of \cite{2007:on_the_complexity_of_algebraic_numbers_I_expansions} states that if $\infw{x}$ is aperiodic and has
  sublinear factor complexity, then $\xi_{\infw{x},b}$ is transcendental. Since episturmian words have sublinear factor
  complexity \cite[Thm.~7]{2001:episturmian_words_and_some_constructions_of_de}, it follows that $\xi_{\infw{t},b}$ is
  transcendental for an \emph{arbitrary} aperiodic episturmian word $\infw{t}$.

  \autoref{thm:dio_lower_bound} together with \autoref{prp:ie_bounds} directly implies the following result. The
  result was obtained in \cite{2010:on_the_expansion_of_some_exponential_periods_in_an_integer} when $d = 2$ based on
  the results of \cite{2006:initial_powers_of_sturmian_sequences}.

  \begin{theorem}\label{thm:ie_bound}
    Let $\infw{t}$ be a regular episturmian word of period $d$ with directive word as in \eqref{eq:dw_multiplicative}.
    If $d = 2$ or $\limsup_k a_k \geq 3$, then $\mu(\xi_{\infw{t},b}) > 2$.
  \end{theorem}

  Hence we have identified a new class of transcendental numbers whose irrationality exponents are greater than $2$.
  This class is uncountable because the set of directive words in the statement is uncountable and an episturmian word
  has a unique directive word.

  The next theorem identifies a new uncountable class of Liouville numbers. The case $d = 2$ was first established by
  Komatsu \cite{1996:a_certain_power_series_and_the_inhomogeneous_continued}.

  \begin{theorem}\label{thm:ie_liouville}
    Let $\infw{t}$ be a regular episturmian word. Then $\xi_{\infw{t},b}$ is a Liouville number if and only if
    $\infw{t}$ has unbounded partial quotients.
  \end{theorem}
  \begin{proof}
    If $\infw{t}$ has unbounded partial quotients, then $\dio{\infw{t}} = \infty$ by \autoref{thm:dio_upq} and
    $\xi_{\infw{t},b}$ is a Liouville number by \autoref{prp:ie_bounds}. If $\infw{t}$ has bounded partial quotients,
    then $\dio{\infw{t}} < \infty$ by \autoref{thm:dio_upq}, and \autoref{prp:ie_bounds} yields a finite upper bound
    for $\mu(\xi_{\infw{t},b})$ since a regular episturmian word over $d$ letters has $(d-1)n + 1$ factors of length
    $n$ for all $n$ \cite[Thm.~7]{2001:episturmian_words_and_some_constructions_of_de}.
  \end{proof}

  Adamczewski and Cassaigne established in \cite{2006:diophantine_properties_of_real_numbers_generated} that
  $\xi_{\infw{x},b}$ is not a Liouville number if $\infw{x}$ is a $k$-automatic word. \autoref{thm:ie_liouville}
  establishes the following weak analogue of this result: if $\infw{t}$ is a regular epistandard word (or its shift)
  and the sequence $(q_k)$ satisfies a linear recurrence, then $\xi_{\infw{x},b}$ is not a Liouville number. Indeed,
  when the sequence $(q_k)$ satisfies a linear recurrence, the corresponding Ostrowski numeration system can be viewed
  as a positional numeration system and the corresponding standard word is automatic with respect to this numeration
  system. For more on Ostrowski-automatic words, see the recent paper
  \cite{2021:ostrowski_automatic_sequences_theory_and_applications}. For a general reference on these topics, see,
  e.g., \cite[Ch.~2]{2014:formal_languages_automata_and_numeration_systems_2}.

  In \cite{2019:a_new_complexity_function_repetitions_in_sturmian}, Bugeaud and Kim prove the following remarkable
  result.

  \begin{proposition}\label{prp:sturmian_ie_dio}\cite[Thm.~4.5]{2019:a_new_complexity_function_repetitions_in_sturmian}
    If $\infw{t}$ is a Sturmian word, then $\mu(\xi_{\infw{t},b}) = \dio{\infw{t}}$.
  \end{proposition}

  This result states that the irrationality exponent of $\xi_{\infw{t},b}$ can be read off its base-$b$ expansion for a
  Sturmian word $\infw{t}$. For most real numbers, this is not possible
  \cite[p.~3287]{2019:a_new_complexity_function_repetitions_in_sturmian}. Since many properties of Sturmian words
  transfer to all episturmian words, it is natural to wonder if this result can be generalized. As a consequence to
  \autoref{prp:4-bonacci_dio}, we obtain a negative answer to this question.

  \begin{proposition}\label{prp:counter_example}
    There exists an episturmian word $\infw{t}$ over a $3$-letter alphabet such that
    $\mu(\xi_{\infw{t},b}) > \dio{\infw{t}}$.
  \end{proposition}
  \begin{proof}
    Let $\infw{t}$ be the episturmian word with directive word $(001122)^\omega$ and intercept $1^\omega$. Since
    $\infw{t}$ is aperiodic, the number $\xi_{\infw{t},b}$ is irrational, and hence $\mu(\xi_{\infw{t},b}) \geq 2$. On
    the other hand, we showed in \autoref{prp:3_ce_dio} that $\dio{\infw{t}} < 2$. The claim follows.
  \end{proof}

  This leaves open if the conclusion of \autoref{prp:sturmian_ie_dio} is true for regular episturmian words satisfying
  the assumptions of \autoref{thm:ie_bound}. We see no reason to believe this as the proof of
  \autoref{prp:sturmian_ie_dio} relies heavily on the theory of continued fractions and no such theory is available for
  general episturmian words.

  Bugeaud and Kim discuss on p. 3288 of \cite{2019:a_new_complexity_function_repetitions_in_sturmian} the base-$b$
  expansions of the numbers of the form $\log(1+1/a)$. Their absolute lower bound for the Diophantine exponents of
  Sturmian words implies that the fractional part of the base-$b$ expansion of $\log(1+1/a)$ cannot be a Sturmian word
  for $a \geq 34$. This result is a direct consequence of \cite[Cor.~1]{1980:legendre_polynomials_and_irrationality}
  which provides an upper bound for the irrationality exponent of $\log(1+1/a)$ which tends to $2$ as $a \to \infty$.
  Applying this corollary to \autoref{prp:tribonacci_dio} yields the following.

  \begin{corollary}
    For every integer $b$ such that $b \geq 2$ and every integer $a \geq 23347$, the fractional part of the base-$b$
    expansion of $\log(1+1/a)$ is not isomorphic to any word in the Tribonacci subshift.
  \end{corollary}

  Results like this can be generated by finding results on the irrationality exponents of specific numbers. Since the
  irrationality exponent of $e$ is $2$ (see, e.g., the proof of
  \cite[Cor.~2]{2010:on_the_expansion_of_some_exponential_periods_in_an_integer}), we obtain the following corollary.
  While this result is very minor, it seems that, besides the results of
  \cite{2010:on_the_expansion_of_some_exponential_periods_in_an_integer}, very little is known about the base-$b$
  expansions of $e$.

  \begin{corollary}
    Let $\infw{t}$ be a regular episturmian word of period $d$ with directive word as in \eqref{eq:dw_multiplicative}
    such that $d = 2$ or $\limsup_k a_k \geq 3$. For every integer $b$ such that $b \geq 2$, the fractional part of the
    base-$b$ expansion of $e$ is not isomorphic to $\infw{t}$.
  \end{corollary}

  Obviously we conjecture that the conclusion is true when $d > 2$. We obtain the same conclusion if we replace the
  number $e$ by a badly approximable number. This observation is rather interesting since there exist badly
  approximable numbers of sublinear complexity; see
  \cite[Sect.~8.5]{2012:Distribution_modulo_one_and_diophantine_approximation}.

  \section{Open Problems}
  We have characterized the initial nonrepetitive complexity of regular episturmian words in \autoref{thm:main}. It is
  unclear how this result can be generalized to all episturmian words.

  \begin{problem}
    Characterize the initial nonrepetitive complexity of a general episturmian word.
  \end{problem}

  Such a characterization would allow to determine the Diophantine exponent of a general episturmian word, but it is
  possible that this problem can be attacked in some other way.

  \begin{problem}
    Determine the Diophantine exponent of an episturmian word. Determine conditions that ensure that the Diophantine
    exponent of an episturmian word is strictly greater than $2$.
  \end{problem}

  The Rauzy graph approach should work for studying the prefix nonrepetitive complexity function and the initial
  critical exponent. However, we find this more difficult than characterizing initial nonrepetitive complexity because
  we need to detect returning to the initial vertex in the Rauzy graph, not a return to some vertex.

  \begin{problem}
    Adjust the methods of \autoref{sec:inrc} for the study of prefix nonrepetitive complexity function of regular
    episturmian words.
  \end{problem}

  We believe \autoref{thm:dio_upq} generalizes to all aperiodic episturmian words.

  \begin{conjecture}
    Let $\infw{t}$ be an aperiodic episturmian word. Then $\dio{\infw{t}} < \infty$ if and only if $\infw{t}$ has
    bounded partial quotients.
  \end{conjecture}

  Let $X_d$ be the set of regular episturmian words with period $d$. Define
  \begin{align*}
    \mathcal{B}(d) &= \{\dio{\infw{t}} : \infw{t} \in X_d\} \quad \text{and} \\
    \mathcal{K}(d) &= \{\dio{\infw{t}} : \text{$\infw{t} \in X_d$ and $\textstyle \limsup_k \, a_k \geq 3$ or $d = 2$}\}.
  \end{align*}
  It is proved in \cite[Thm.~4.3]{2019:a_new_complexity_function_repetitions_in_sturmian} that
  $\mathcal{B}(2) \geq \frac53 + \frac{4\sqrt{10}}{15}$ and that this lower bound is optimal. This is a remarkable
  result, and we propose the following problems and questions.

  \begin{problem}
    Reprove that $\mathcal{B}(2) \geq \frac53 + \frac{4\sqrt{10}}{15}$ using \autoref{thm:main}. Find and prove an
    optimal lower bound for $\mathcal{K}(d)$ when $d > 2$.
  \end{problem}

  \begin{question}
    What is the least element of $\mathcal{K}(d)$ and what is the least accumulation point of $\mathcal{K}(d)$? Is it
    true that $\inf \mathcal{B}(d) = 1$ when $d > 2$?
  \end{question}

  \autoref{thm:dio_lower_bound} was not enough to conclude that $\mu(\xi_{\infw{t},b}) > 2$ for all regular episturmian
  words $\infw{t}$. It would be interesting to know if this is true.

  \begin{question}
    Is it true that $\mu(\xi_{\infw{t},b}) > 2$ for all regular episturmian words $\infw{t}$? What about all aperiodic
    episturmian words?
  \end{question}

  Our results provide lower bounds for $\mu(\xi_{\infw{t},b})$ when $\infw{t}$ is a regular episturmian word, and
  \autoref{prp:counter_example} suggests that the lower bounds are strict. Hence we propose the following problem.

  \begin{problem}
    Find a better upper bound than that of \autoref{prp:ie_bounds} for the irrationality exponent of a regular
    episturmian word with bounded partial quotients.
  \end{problem}

  Finally we propose the following problem related to the discussion after \autoref{thm:ie_liouville}. This is an
  analogue of a conjecture of Shallit \cite{1999:number_theory_and_formal_languages} that was settled in
  \cite{2006:diophantine_properties_of_real_numbers_generated}.

  \begin{problem}
    Develop rigorously the notion of an Ostrowski-automatic word in the setting $d > 2$ and settle the conjecture that
    if $\infw{t}$ is an Ostrowski-automatic word, then $\xi_{\infw{t},b}$ is not a Liouville number.
  \end{problem}

  \section*{Acknowledgments}
  We thank the reviewers for pointing out several problems in the original manuscript. Due to your help, these issues
  have been identified and fixed. We are especially thankful for \autoref{ex:greedy_fail} which indicates that not all
  Ostrowski expansions equal the corresponding greedy expansions.

  \printbibliography[heading=bibintoc]
  
\end{document}